\documentclass[10pt,journal,final,twocolumn,twoside]{IEEEtran}

\usepackage[dvips]{graphicx}
\usepackage{epsfig}
\usepackage[cmex10]{amsmath}
\usepackage{amssymb}
\usepackage{amsthm}
\usepackage{amsfonts}
\usepackage{bm}
\usepackage{dsfont} 
\usepackage{cite}
\usepackage[svgnames]{xcolor} 
\usepackage[tight,footnotesize]{subfigure}
\usepackage{algpseudocode}
\usepackage{algorithm}
\usepackage{graphics}
\usepackage{mathrsfs}
\usepackage{psfrag}
\usepackage{tikz}
\usepackage{accents}
\usepackage{bbm}
\graphicspath{{figs/}}
\usepackage{pgfplots} 
\pgfplotsset{compat=newest} 
\pgfplotsset{plot coordinates/math parser=false} 
\newlength\figureheight 
\newlength\figurewidth

\input{parham.def}

\date{\today}

\begin{document}

\title{Mismatched Binary Hypothesis\\ Testing: Error Exponent Sensitivity}

\author{Parham Boroumand and Albert Guill\'en i F\`abregas

%

\thanks{P. Bouroumand is with the Department of Engineering, 
University
of Cambridge, Cambridge CB2 1PZ, U.K. (e-mail: pb702@cam.ac.uk).

A.~Guill\'en i F\`abregas is with the Department of Engineering, 
University
of Cambridge, Cambridge CB2 1PZ, U.K. and with the Department of Information and 
Communication
Technologies, Universitat Pompeu Fabra, Barcelona 08018, Spain. (e-mail: guillen@ieee.org).
}
\thanks{This work was supported in part by the European Research Council under 
Grant 725411.
}
\thanks{This work has been presented in part at the 2020 International Z\"urich Seminar on Information and Communication, Z\"urich, Switzerland, Feb. 2020 and at the 2021 IEEE International Symposium on Information Theory, Melbourne, Australia, Jul. 2021.}

\thanks{Copyright (c) 2017 IEEE. Personal use of this material is permitted.  However, permission to use this material for any other purposes must be obtained from the IEEE by sending a request to pubs-permissions@ieee.org.}
}

\maketitle

\begin{abstract} 

We study the problem of mismatched binary hypothesis testing between i.i.d. distributions. We analyze the tradeoff between the pairwise error probability exponents when the actual distributions generating the observation are different from the distributions used in the likelihood ratio test, sequential probability ratio test, and Hoeffding's generalized likelihood ratio test in the composite setting. When the real distributions are within a small divergence ball of the test distributions, we find the deviation of the worst-case error exponent of each test with respect to the matched error exponent. In addition, we consider the case where an adversary tampers with the observation, again within a divergence ball of the observation type. We show that the tests are more sensitive to distribution mismatch than to adversarial observation tampering.
\end{abstract}

\section{Introduction }
\label{sec:intro}

Hypothesis testing, or the problem of deciding the probability distribution that generated a given observation, is one of the main problems in statistics, finding applications in social, biological, medical and data sciences, including information theory, image processing and signal processing.  Depending on the specific subject and underlying assumptions, hypothesis testing has been termed classification, model selection, discrimination, or signal detection. The simplest instance of this problem is the binary case, i.e., deciding which of the two probability distributions have generated the observation.
Depending on whether or not priors on the hypotheses are available, the problem is referred to as Bayesian or non-Bayesian. In the Bayesian setting, the average probability of error emerges as the key performance measure. Instead, when priors are not available the average error probability cannot be computed and one must consider a tradeoff among the pairwise error probabilities (see e.g. \cite{Lehmann,poor2013introduction}).

In \cite{Neyman}, Neyman and Pearson studied non-Bayesian binary case and considered the minimization of one pairwise error probability subject to the other being upper-bounded by a constant. In this setting, they proved that a hypothesis test that computes the ratio between the two probability distributions for the given observation and compares it to a threshold is optimal. In order to implement the optimal likelihood ratio test, one must wait to process the whole block of observation data at once. In some applications, it may be preferable to attempt to make a decision as promptly as possible. In \cite{Wald1}, Wald proposed a sequential test that attempts to make a decision at every time instant, or waits one time instant for the arrival of a new observation sample; the optimality of the above test was established in \cite{Wald}.

A critical underlying assumption in the above and follow-up works, is that the probability distributions of each of the hypotheses are known and thus, can be employed by testing algorithms. While highly desirable, this is an optimistic assumption that is difficult to guarantee in practice. A number of solutions have been proposed in the literature. Composite hypothesis testing considers the case where the distributions that generate the observation belong to known families of distributions. Hoeffding proposed an asymptotically optimal test for this setting in \cite{Hoeffding}. Classification assumes no knowledge of the underlying probability distributions but assumes the availability of training data (see e.g. \cite{ziv1988classification,CoverClass,Gutman}). Robust hypothesis testing assumes a statistical model of the variability of the true distribution, which is then used to design the optimal test for that robustness model \cite{Huber, Kassam, Levy, poor2013introduction}.

The simplicity of the Neyman and Pearson's likelihood ratio or Wald's sequential tests has brought them into practice, even in settings where the distributions are unknown.  In this work, we consider an alternative to the above methods. We assume that the true probability distributions $P_0$ and $P_1$ are unknown, but instead,  two fixed probability distributions $\hat P_0$ and $\hat P_1$ are used for testing using the optimal tests for the cases where the distributions are known. We refer to this case to as mismatched  hypothesis testing. 


In this work, we study the exponential decay of the probability of error, or error exponents in short, as a proxy for the performance of hypothesis testing. In particular, we consider the error exponent tradeoff between both pairwise error probabilities. We consider the worst case error exponent tradeoff when the actual distributions generating the observation are within a certain distance of the test distributions. As a measure of distance, we use the family of R\'enyi relative entropies \cite{renyi1961measures} (see also \cite{Harremos}) as well as $f$-divergences \cite{ali1966general,csiszar1967information} with $f(t)$ such that its second derivative at $t=1$ is bounded. We study the behavior of the worst-case tradeoff when the distance between the test and the true distributions is small and provide an approximation of the worst-case exponent as an expansion around the matched exponent, i.e., the exponent attained when the distributions are known. In addition, we extend the results to the case where the mismatch is not in the testing distributions but in the observation: an adversary has modified the observation data within a certain divergence of the observation type.
 
This paper is structured as follows. Section \ref{sec:pre} introduces notation and reviews the preliminaries about likelihood ratio, Hoeffding's generalized likelihood ratio and sequential  testing. Sections \ref{sec:lrt}, \ref{sec:glrt} and \ref{sec:sequentialHT} discuss our main results for the likelihood ratio, generalized likelihood ratio and sequential probability ratio tests in the mismatched setting. Section \ref{sec:adv} discusses the case where the mismatch is not in the distribution, but in the observation, i.e., the cases where the observation data has been tampered by an adversary. Proofs of the main results can be found in the Appendices.

\section{Preliminaries}
\label{sec:pre}

Consider the binary hypothesis testing problem \cite{Lehmann} where an observation $\xv=(x_1,\dotsc,x_k)\in\Xc^k$ is generated from two possible distributions $P^k_0$ and $P^k_1$ defined on the probability simplex $\Pc(\Xc^k)$, for some observation alphabet $\Xc$. We assume that $P^k_0$ and $P^k_1$ are product distributions, i.e., $P_0^k(\xv)=\prod_{i=1}^k P_0(x_i)$, and similarly for $P_1^k$. For simplicity, we assume that both $P_0(x)>0$ and $P_1(x)>0$ for each $x\in\Xc$. In the following, we describe the settings considered in the paper.

\subsection{Likelihood Ratio Test}

In this setting $k$ is assumed to be a fixed integer $n$. Let  $\phi: \Xc^n \rightarrow \{0,1\}$ be a hypothesis test that decides which distribution generated the observation $\xv$. We consider deterministic tests $\phi$ that decide in favor of $P_0$ if $\xv\in \Ac_0$, where $\Ac_0\subset \Xc^n$ is the decision region for the first hypothesis, and decides in favor of $P_1$ otherwise. We define $\Ac_1=\Xc^n \setminus \Ac_0$ to be the decision region for the second hypothesis. The test performance is determined by the pairwise error probabilities
\begin{equation}\label{eq:e1}
\epsilon_0 (\phi)= \sum_{\bx \in \Ac_1} P_0^n(\bx),~~~~ \epsilon_1 (\phi)= \sum_{\bx \in \Ac_0} P_1^n(\bx).  
\end{equation}
A hypothesis test is said to be optimal whenever it achieves the optimal error probability tradeoff given by 
\begin{equation}\label{eq:trade}
 \min_{\phi: \epsilon_1 (\phi) \leq \xi} \epsilon_0 (\phi),
\end{equation}
where $\xi\in[0,1]$.

The likelihood ratio test defined as
\begin{equation}
\phi(\xv)= \mathbbm{1} \bigg \{ \frac{P_1^n(\bx)}{P_0^n(\bx)}  \geq e^{n\gamma} \bigg\},
\end{equation}
was shown in  \cite{Neyman} to attain the optimal tradeoff \eqref{eq:trade} for every $\gamma$. The type of a sequence $\bx= (x_1,\ldots,x_n)$ is defined as $\Th(a)=\frac{N(a|\bx)}{n}$, where $N(a|\bx)$ is the number of occurrences of the symbol $a\in\Xc$ in the string. The likelihood ratio test can also be expressed as a function of the type of the observation $\Th$ as \cite{Cover}
\begin{align}\label{eq:LRTtype}
\phi(\Th)= \mathbbm{1} \big\{ D(\Th\|P_0)-D(\Th\|P_1)    \geq  \gamma \big\},
\end{align}
where $D(P\|Q)= \sum_{x\in\Xc} P(x) \log \frac{P(x)}{Q(x)}$ is the relative entropy between distributions $P$ and $Q$. This expression of the likelihood ratio test will be used extensively throughout the paper.

In this paper, we  study the asymptotic exponential decay of the pairwise error probabilities as the observation length $n$ tends to infinity, i.e., 
\begin{equation}
E_0 \triangleq \liminf_{n\to\infty} -\frac{1}{n} \log \epsilon_0 (\phi), ~~~E_1 \triangleq \liminf_{n\to\infty} -\frac{1}{n} \log \epsilon_1 (\phi).
\label{eq:exponents}
\end{equation}
These limits are known to exist for i.i.d. observations, as is the case of this paper. In order to study the tradeoff between error exponents, it is sufficient to consider deterministic tests. The optimal error exponent tradeoff $(E_0,E_1)$ is defined as
\begin{align}\label{eq:tradefix}
E_1(E_0) \triangleq \sup \big\{&E_1\in \mathbb{R}_{+}: \exists \phi , \exists N_0 \in \ZZ_+  \  \text{s.t.} \  \forall   n>N_0,  \nonumber \\
&\epsilon_0(\phi) \leq e^{-nE_0}  ~ \text{and} ~  \epsilon_1(\phi) \leq e^{-nE_1}\big\}.
\end{align}


By using Sanov's Theorem \cite{Cover,Dembo}, the optimal error exponent tradeoff $(E_0,E_1)$, attained by the likelihood ratio test, can be shown to be \cite{Blahut}
\begin{align}
E_0=\min_{Q \in \Qc_0} D(Q\|P_0)\label{eq:min1},\\
E_1=\min_{Q \in \Qc_1} D(Q\|P_1)\label{eq:min2},
\end{align}
where 
\begin{align}
\Qc_0&=  \big\{Q\in \Pc(\Xc): D(Q\| P_0)-D(Q\|P_1) \geq  \gamma  \big\}, \label{eq:constraint1} \\
\Qc_1&= \big\{Q\in \Pc(\Xc): D(Q\| P_0)-D(Q\|P_1) \leq  \gamma    \big\}. \label{eq:constraint2}
\end{align}

The minimizing distribution in \eqref{eq:min1}, \eqref{eq:min2} is the tilted distribution
\begin{equation}\label{eq:tilted}
Q_{\lambda}(x)= \frac{ P_{0}^{1-\lambda}(x) P_{1}^{\lambda}(x) } {\sum_{a \in \Xc }  P_{0}^{1-\lambda}(a) P_{1}^{\lambda}(a) }, ~~~0\leq\lambda \leq 1,
\end{equation}
whenever $\gamma$ satisfies $-D(P_0\|P_1) \leq \gamma  \leq D(P_1\|P_0)$.
In this case, $\lambda$ is the solution of
\begin{equation}\label{eq:KKTgamma} 
D(Q_{\lambda}\| P_0)-D(Q_{\lambda} \| P_1) = \gamma.
\end{equation}
Instead, if $\gamma<-D(P_0\|P_1)$, the optimal distribution in \eqref{eq:min1} is $Q_\lambda(x)= P_0(x)$ and $E_0=0$, and if $\gamma>D(P_1\|P_0)$, the optimal distribution in \eqref{eq:min2} is $Q_\lambda(x)= P_1(x)$ and $E_1=0$. 

Equivalently, the dual expressions of \eqref{eq:min1} and \eqref{eq:min2} can be derived by substituting the minimizing distribution \eqref{eq:tilted} into the Lagrangian yielding \cite{Blahut,Dembo} 
\begin{align}
E_0&=\max_{\lambda \geq 0 } \lambda \gamma - \log \Big (  \sum_{x\in \Xc} P_0^{1-\lambda}(x) P_1^{\lambda}(x) \Big ),  \\
E_1&=\max_{\lambda \geq 0 } -\lambda \gamma - \log \Big (  \sum_{x\in \Xc} P_0^{\lambda}(x) P_1^{1-\lambda}(x)  \Big ).  
\end{align}

The Stein regime is defined as the highest error exponent  under one hypothesis  when the  error probability under the other hypothesis is at most some fixed $ \epsilon \in (0,\frac{1}{2})$ \cite{Cover}
\begin{align}\label{eq:steindef}
E_1^{(\epsilon)}  \triangleq \sup \big \{ & E_1\in \mathbb{R}_{+}: \exists \phi , \exists n_0 \in \ZZ_+  \  \text{s.t.} \  \forall   n>n_0 \nonumber \\
&\epsilon_0 (\phi)\leq \epsilon  \quad \text{and} \quad \epsilon_1(\phi) \leq e^{-nE_1} \big \}.
\end{align} 
The optimal $E_1^{(\epsilon)}$, given by \cite{Cover}
\begin{equation}\label{eq:stein2}
E_1^{(\epsilon)} = D(P_0\|P_1),
\end{equation}
can be achieved by setting the threshold in \eqref{eq:LRTtype} to be ${\gamma} = -D(P_0\|P_1)+\frac{c}{\sqrt{n}}$, where $c$ is a constant that depends on distributions $P_0, P_1$ and $\epsilon$.

\subsection{Generalized Likelihood Ratio Test} 

In this setting, $k$ is also a fixed integer $n$ similarly to the likelihood ratio test case. We consider the composite hypothesis testing problem where only the first distribution $P_0$ is known, and no prior information is available regarding the second distribution $P_1$. Hoeffding proposed in \cite{Hoeffding} the following generalized likelihood ratio test for when $P_0$ is known, while the second distribution is restricted to $\Pc(\Xc)$,
\begin{equation}
\phi(\bx)= \mathbbm{1} \bigg \{ \frac{\sup_{\substack{ P_1 \in  \Pc(\Xc) } }P_1^n (\xv)}{P_0^n(\bx)}  \geq e^{n\gamma} \bigg\}.
\end{equation}
Similarly to \eqref{eq:LRTtype}, Hoeffding's generalized likelihood ratio test can be expressed as a function of the type of the observation $\Th$ as \cite{Cover,Merhav}
\begin{align}\label{eq:GLRTtype}
\phi(\Th)= \mathbbm{1} \big\{ D(\Th\|P_0)    \geq  \gamma \big\}.
\end{align}
The pairwise error probabilities are defined as exactly as in \eqref{eq:e1} where $\Ac_0$ and $\Ac_1$ are the corresponding decision regions of Hoeffding's test. The optimal error exponent tradeoff is defined  as in \eqref{eq:tradefix}.

By using Sanov's Theorem, for any $0 \leq \gamma \leq D(P_1\|P_0)$ the error exponents of Hoeffding's generalized likelihood ratio test are given by
\begin{align}
E_0&=\min_{Q: D(Q\|P_0) \geq \gamma } D(Q\|P_0)=\gamma\label{eq:minH1},\\
E_1&=\min_{Q: D(Q\|P_0) \leq \gamma} D(Q\|P_1)\label{eq:minH2}.
\end{align}
The minimizing distribution in \eqref{eq:minH2} is the tilted distribution
\begin{equation}\label{eq:tiltedH}
Q_{\mu}(x)= \frac{ P_{0}^{\frac{\mu}{1+\mu}}(x) P_{1}^{\frac{1}{1+\mu}}(x) } {\sum_{a \in \Xc }  P_{0}^{\frac{\mu}{1+\mu}}(a) P_{1}^{\frac{1}{1+\mu}}(a) }, ~~~0\leq\mu,
\end{equation}
and the $\mu$ is the solution to $D(Q_\mu\|P_0)=\gamma$. Therefore, by comparing \eqref{eq:tilted} and \eqref{eq:tiltedH}, the optimizing distributions have the same form and there exist some thresholds for likelihood ratio test and Hoeffding's test such that $Q_\mu =Q_\lambda$. Hence Hoeffding's test can achieve the optimal error exponent tradeoff \cite{Merhav}.

\subsection{Sequential Probability Ratio Test}

In the sequential setting, the number of samples $k$ is a random variable called the stopping time $\tau$ taking values in $\ZZ_{+}$. A sequential hypothesis test is a pair $\Phi=(\phi: \Xc^\tau \rightarrow \{0,1\},\tau)$ where for every $n\geq 0$ the event $\{\tau\leq n\} \in \mathscr{F}_n$ where $\mathscr{F}_n$ is the sigma algebra induced by random variables $X_1,\ldots,X_n$, i.e., $ \sigma(X_1,\ldots,X_n)$. Moreover, $\phi$ is a $\mathscr{F}_{\tau}$ measurable decision rule, i.e., the decision rule determined by causally observing the sequence $X_i$. In other words, at each time instant, the test attempts to make a decision in favor of one of the hypotheses or chooses to take a new sample. 

The two possible pairwise error probabilities that measure the performance of the test are defined as
\begin{equation}\label{eq:errprob}
\epsilon_0(\Phi)=\PP_0\big [\phi(X^{\tau})\neq 0  \big]   ~ \text{,} ~ 	\epsilon_1(\Phi)=\PP_1\big[\phi(X^{\tau})\neq 1  \big],
\end{equation}
where the probabilities are over $P_0, P_1$, respectively. 

There are two definitions of achievable error exponents in the literature. According to \cite{Csiszar}  the optimal error exponent tradeoff is defined as
\begin{align}\label{eq:tradeseq1}
 E_1(E_0)& \triangleq \sup \Big \{E_1\in \mathbb{R}^{+}: \exists \Phi ,\ \exists \ n  \in \ZZ_{+}   \text{ s.t.} \     \mathbb{E}_{P_0} [\tau] \leq n, \nonumber \\
 &\mathbb{E}_{P_1} [\tau]  \leq n, \   \epsilon_0(\Phi) \leq 2^{- n E_0}    ~ \text{and} ~  \epsilon_1(\Phi) \leq 2^{- n E_1} \Big \} .
\end{align}
Alternatively, the expected stopping time can be different under each hypothesis by design to increase the reliability under one of the hypotheses by taking a larger number of samples compared to the number of samples the alternative hypothesis. Accordingly, \cite{Poly}  defined the error exponent tradeoff as
\begin{align}\label{eq:tradeseq2}
 &E_1(E_0)\notag\\
 &~~\triangleq  \sup \Big \{E_1\in \mathbb{R}^{+}: \exists \Phi ,\ \exists \ n_0, n_1  \in \ZZ_{+}, \text{s.t.}  \   \mathbb{E}_{P_0} [\tau]   \leq n_0, \nonumber \\
&~~~~ \mathbb{E}_{P_1}[\tau]  \leq n_1, \   \epsilon_0(\Phi) \leq 2^{- n_0 E_0}  ~ \text{and} ~   \epsilon_1(\Phi) \leq 2^{- n_1 E_1} \Big \},
\end{align}
which allows different stopping times under different hypothesis. 

%

The sequential probability ratio test (SPRT) $\Phi=(\phi,\tau)$ was proposed by Wald in \cite{Wald1}. The stopping time is defined as follows
\begin{align}
\tau=\inf \big\{n\geq1:S_n& \geq  \gamma_0 \  \text{or} \  S_n \leq -\gamma_1\big\},
\label{eq:deftau}
\end{align} 
where 
\begin{align}\label{eq:LLR}
S_n=\sum_{i=1}^n \log \frac{P_0(x_i)}{P_1(x_i)},
\end{align} 
is the the accumulated log-likelihood ratio (LLR) of the observed sequence $\xv$ and $\gamma_0, \gamma_1$ are two positive real numbers. Moreover, the test makes a decision according to the rule
\begin{align}
\phi= 
\begin{cases}
0&  \text{if } S_\tau \geq \gamma_0  \\
1 &  \text{if } S_\tau \leq  - \gamma_1,\\
\end{cases}
\label{eq:defsprt}
\end{align}    
It is shown in  \cite{Berk} that the above test attains the optimal error exponent tradeoff, i.e., as thresholds $\gamma_0, \gamma_1$  approach infinity, the test achieves the best error exponent trade-off in \eqref{eq:tradeseq1} and \eqref{eq:tradeseq2}. It is known that the error probabilities of sequential probability ratio test as a function of $\gamma_0$ and  $\gamma_1$ are \cite{Woodroofe} 
\begin{align}
\epsilon_0 = c_0 \cdot e^{-\gamma_1 }  \quad , \quad  \epsilon_1 =c_1 \cdot e^{-\gamma_0 } ,
\end{align}
as $\gamma_0, \gamma_1 \rightarrow \infty$ where $c_0, c_1$ are positive constants. Moreover, it can also be shown that
\begin{align}
\mathbb{E}_{P_0} [\tau] &= \frac{  \gamma_0}{D(P_0\|P_1)}(1+o(1))     \label{eq:SPRT1}, \\ 
\mathbb{E}_{P_1} [\tau] &= \frac{  \gamma_1}{D(P_1\|P_0)}(1+o(1))   \label{eq:SPRT2}.
\end{align}
Therefore, according to definition \eqref{eq:tradeseq1}, the optimal error exponent tradeoff is given by,
\begin{equation}
E_0 = D(P_1\|P_0)  , ~ E_1= D(P_0\|P_1),
\end{equation}
where thresholds $\gamma_0, \gamma_1$  are chosen as
\begin{equation}\label{eq:thresh}
\gamma_0=n\big (D(P_0\|P_1)+o(1)\big),~ \gamma_1= n\big(D(P_1\|P_0)+o(1)\big). 
\end{equation}
Hence, the sequential probability ratio test achieves the Stein regime error exponents of the standard likelihood ratio test simultaneously.  Moreover, according to definition \eqref{eq:tradeseq2} the optimal error exponent tradeoff is given by
\begin{equation}
E_0 = \ell D(P_1\|P_0)  , ~ E_1= \frac{1}{\ell} D(P_0\|P_1),
\end{equation}
where $\ell= \frac{n_1}{n_0} $. Equivalently, we have 
\begin{equation} \label{eq:seqexp}
E_0 E_1 = D(P_0\|P_1)D(P_1\|P_0).
\end{equation}
To achieve \eqref{eq:seqexp}, thresholds $\gamma_0, \gamma_1$ should be chosen as
\begin{equation}\label{eq:threshell}
\gamma_0=n_0\big (D(P_0\|P_1)+o(1)\big) ,~ \gamma_1=  n_1\big(D(P_1\|P_0)+o(1)\big).
\end{equation}


\section{Likelihood Ratio Testing Sensitivity}
\label{sec:lrt}

In this section, we study the robustness of mismatched likelihood ratio testing. We first derive the optimal error exponent tradeoff and find the worst-case error exponent when the true distribution lies in a small relative entropy ball around the testing distribution. Then, we study the deviation of the worst case error exponent around the matched likelihood ratio test exponent for small divergence balls, where the divergence is either the R\'enyi divergence of order $\alpha$ or the $f$-divergence with $\frac{d^2f(t)}{dt^2}\big|_{t=1}=\alpha$.

Let $\hat{P}_0(x)$ and $\hat{P}_1(x)$ be the testing distributions used in the likelihood ratio test with threshold $\hat{\gamma}$ given by
\begin{align}\label{eq:LRTtypeMM}
\hat{\phi}(\Th)= \mathbbm{1} \big\{ D(\Th\|\hat{P}_0)-D(\Th\|\hat{P}_1)    \geq \hat{\gamma} \big\}.
\end{align}
For simplicity, we assume that both $\hat P_0(x)>0$ and $\hat P_1(x)>0$ for each $x\in\Xc$. We are interested in the achievable error exponent tradeoff of the mismatched likelihood ratio test, i.e., 
\begin{align}\label{eq:opttestmism}
\hat{E}_1(\hat{E}_0) \triangleq \sup \big\{&\hat{E}_1\in \mathbb{R}_{+}: \exists \hat{\gamma}  , \exists N_0 \in \ZZ _+ \  \text{s.t.} \  \forall  n>N_0, \nonumber \\
&\epsilon_0 \leq e^{-n\hat{E}_0}  ~ \text{and} ~  \epsilon_1 \leq e^{-n\hat{E}_1}\big\}.
\end{align}

\begin{theorem}\label{thm:mismatchLRT}
	For fixed $\hat{P}_0, \hat{P}_1 \in \Pc(X)$ the optimal error exponent tradeoff in \eqref{eq:opttestmism} is given by		
	\begin{align}
	\hat{E}_0&=\min_{Q \in \hat{\Qc}_0} D(Q\|P_0) \label{eq:LRTmis1},\\
	\hat{E}_1&=\min_{Q \in \hat{\Qc}_1} D(Q\|P_1)\label{eq:LRTmis2},
	\end{align}	
	where 
	\begin{align}
	\hat{\Qc}_0&=  \big\{Q\in \Pc(\Xc): D(Q\| \hat{P}_0)-D(Q\|\hat{P}_1) \geq \hat{\gamma} \big \}, \label{eq:qhat1}\\
	\hat{\Qc}_1&= \big\{Q\in \Pc(\Xc): D(Q\| \hat{P}_0)-D(Q\|\hat{P}_1) \leq  \hat{\gamma}   \big \}. \label{eq:qhat2}
	\end{align}
	The minimizing distributions in \eqref{eq:LRTmis1} and \eqref{eq:LRTmis2} are 
	\begin{equation}\label{eq:tiltedMM1}
	\hat{Q}_{\lambda_0}(x)= \frac{ P_0(x)  \hat{P}_{0}^{-\lambda_0}(x) \hat{P}_{1}^{\lambda_0}(x) } {\sum_{a \in \Xc } P_0 (a) \hat{P}_{0}^{-\lambda_0}(a) \hat{P}_{1}^{\lambda_0}(a) },~~\lambda_0\geq0,
	\end{equation}
	\begin{equation}\label{eq:tiltedMM2}
	\hat{Q}_{\lambda_1}(x)= \frac{ P_1(x)  \hat{P}_{1}^{-\lambda_1}(x) \hat{P}_{0}^{\lambda_1}(x) } {\sum_{a \in \Xc } P_1 (a) \hat{P}_{1}^{-\lambda_1}(a) \hat{P}_{0}^{\lambda_1}(a) },~~\lambda_1\geq0
	\end{equation}
	respectively, where  $  \lambda_0$ is chosen so that	
	\begin{equation}\label{eq:KKTgamma1} 
	D(\hat{Q}_{\lambda_0}\|\hat{P}_0)-D(\hat{Q}_{\lambda_0} \| \hat{P}_1) = \hat{\gamma},
	\end{equation}
	whenever  
	$D(P_0\|\hat{P}_0)- D(P_0\|\hat{P}_1) \leq \hat{\gamma}$,
	and otherwise, $\hat{Q}_{\lambda_0}(x)=P_0(x)$ and  $\hat{E}_0=0$. Similarly,  $ \lambda_1 \geq 0$ is chosen so that
	\begin{equation}\label{eq:KKTgamma2} 
	D(\hat{Q}_{\lambda_1}\|\hat{P}_0)-D(\hat{Q}_{\lambda_1} \| \hat{P}_1) = \hat{\gamma},
	\end{equation}
	whenever
	$   D(P_1 \|\hat{P}_0) -D(P_1 \|\hat{P}_1) \geq\hat{\gamma},$    and otherwise, $\hat{Q}_{\lambda_1}(x)=P_1(x)$ and  $\hat{E}_1=0$. Furthermore,  the dual expressions for the type-\RNum{1} and type-\RNum{2} error exponents  are
	\begin{align}\label{eq:dual}
	\hat{E}_0&=\max_{\lambda \geq 0 } \lambda \hat{\gamma} - \log \Big (  \sum_{x\in \Xc} P_0(x) \hat{P}_0^{-\lambda}(x) \hat{P}_1^{\lambda}(x) \Big ),  \\
	\hat{E}_1&=\max_{\lambda \geq 0 } -\lambda \hat{\gamma} - \log \Big (  \sum_{x\in \Xc} \hat{P}_0^{\lambda}(x) P_1(x) \hat{P}_1^{-\lambda}(x)  \Big ).  
	\end{align} 
\end{theorem}


\begin{proof}
	Theorem \ref{thm:mismatchLRT}, proved in Appendix \ref{apx:TmismatchLRT} follows from a direct application of Sanov's Theorem.
\end{proof}

\begin{remark}
	For mismatched likelihood ratio testing, the optimizing distributions $\hat{Q}_{\lambda_0}, \hat{Q}_{\lambda_1}$ can be different, since the decision regions only depend on the mismatched distributions. However, if $\hat{P}_0, \hat{P}_1$ are tilted with respect to $P_0$ and $P_1$, then both $\hat{Q}_{\lambda_0}, \hat{Q}_{\lambda_1}$ are also tilted respect to $P_0$ and $P_1$. {This implies that for any set of mismatched distributions $\hat{P}_0, \hat{P}_1$ that  are tilted with respect to generating distributions, there exists  a threshold $\hat{\gamma}$ such that the mismatched likelihood ratio test achieves the optimal error exponent tradeoff in \eqref{eq:tradefix}. However, the thresholds to achieve the same type-\RNum{1} and type-\RNum{2} error exponents, under the matched and mismatched tests are different, and the difference can be found by equating the likelihood ratio test for $\hat{Q}_{\lambda}$}
\end{remark}

\begin{theorem}\label{thm:stein}
	
	In the Stein regime, the mismatched likelihood ratio test achieves 
	\begin{equation}\label{eq:steinMM2}
	\hat{E}_1^{(\epsilon)}=\min_{Q \in \hat{\Qc}_1} D(Q\|P_1),
	\end{equation}
	with threshold  
	\begin{equation}\label{eq:steinthresh2}
	\hat{\gamma}=D(P_0\|\hat{P}_0) -D(P_0\|\hat{P}_1) +\sqrt{\frac{V(P_0,\hat{P}_0,\hat{P}_1)}{n}}\Qsf^{-1}(\epsilon),
	\end{equation}
	where
	\begin{equation}
V(P_0,\hat{P}_0,\hat{P}_1) = {\rm Var}_{P_0}\bigg[\log \frac{\hat{P}_0(X)}{\hat{P}_1(X)} \bigg ],
\end{equation}
is the variance of the random variable $\log \frac{\hat{P}_0(X)}{\hat{P}_1(X)}$ where $X$ is distributed according to $P_0$, and 
$\Qsf^{-1}(\epsilon)$ is the inverse cumulative distribution function of a zero-mean unit-variance Gaussian random variable.
\end{theorem} 

\begin{proof}
	Theorem \ref{thm:stein}, proved in Appendix \ref{apx:Tstein} follows from Central Limit Theorem.
\end{proof}

\begin{remark}	
	Note that since $P_0$ satisfies the constraint in \eqref{eq:steinMM2} then $\hat{E}_1^{(\epsilon)} \leq {E}_1^{(\epsilon)}$. In fact, if $\hat{P}_0, \hat{P}_1$ are tilted respect to $P_0, P_1$ then this inequality is met with equality. Moreover, it is easy to find a set of data and test distributions where  $\hat{E}_1^{(\epsilon)} < {E}_1^{(\epsilon)}$. 
	
\end{remark}


Next, we study the worst-case error-exponent performance of mismatched likelihood ratio testing when the distributions generating the observation fulfil
 \begin{equation}\label{eq:ball2}
 P_0 \in \Bc(\hat P_0,r_0),~~ P_1 \in \Bc(\hat P_1, r_1),
 \end{equation}
 where
  \begin{equation}\label{eq:ball}
 \Bc(Q,r)= \big\{P\in\Pc(\Xc): d(Q,P) \leq r \big\},
 \end{equation} 
 is a ball centred at distribution $Q$ containing all distributions whose  distance is smaller or equal than radius $r$, and for the R\'enyi divergence of positive order $\alpha$ where $\alpha\neq 1$ we set $d(Q,P)=D_\alpha(Q\|P)=\frac{1}{\alpha-1} \log  \sum_{x\in\Xc} Q(x)^\alpha P(x)^{1-\alpha}$, and for $\alpha=1$, the continuity in $\alpha$ leads to defining the R\'enyi divergence of order $1$ to be the relative entropy. Similarly, given a convex and twice differentiable function $f$, we set $d(Q,P)=D_{f}(Q\|P)=\sum_{x\in\Xc} P(x) f\Big( \frac{Q(x)}{P(x)} \Big) $ to be the $f$-divergence, and we set $\alpha=\frac{d^2f(t)}{dt^2}\Big|_{t=1}$.

 	For every $P_0$ the achievable error exponent  $\hat{E}_0$  does not depend on $P_1$ therefore, for every  $r_0, r_1 \geq 0$ the least favorable exponents  $\underline{\hat{E}}_0(r_0), \underline{\hat{E}}_1(r_1)$ can be written as
	\begin{align}
	\underline{\hat{E}}_0(r_0)&=\min_{P_0 \in \Bc(\hat P_0,r_0)    } \ \ \min_{Q \in \hat{\Qc}_0} D(Q\|P_0),\label{eq:MMlower1}\\
	\underline{\hat{E}}_1(r_1)&=\min_{P_1 \in  \Bc(\hat P_1,r_1)} \ \ \min_{ Q \in \hat{\Qc}_1  } D(Q\|P_1), \label{eq:MMlower2}
	\end{align}
	where $\hat{\Qc}_0, \hat{\Qc}_1 $ are defined in \eqref{eq:qhat1}, \eqref{eq:qhat2}. Then, for any distribution pair $P_0 \in\Bc (\hat P_0,r_0), P_1 \in \Bc(\hat P_1,r_1)$, the corresponding error exponent pair $(\hat{E}_0, \hat{E}_1)$  satisfies
	\begin{equation}\label{eq:LU1}
	\underline{\hat{E}}_0(r_0)  \leq		\hat{E}_0  \text{,} \quad \underline{\hat{E}}_1(r_1)  \leq		\hat{E}_1.
	\end{equation}

	Figure \ref{fig:mismatch} depicts the mismatched probability distributions and the mismatched likelihood ratio test as a hyperplane dividing the probability space into the two decision regions. The worst-case achievable error exponents of mismatched likelihood ratio testing for data distributions in a divergence ball are essentially the minimum relative entropy between two sets of convex probability distributions. Specifically, the minimum relative entropy between $\Bc(\hat P_0,r_0)$ and $\hat \Qc_1$ gives $\underline{\hat{E}}_0(r_0)$, and similarly for $\underline{\hat{E}}_1(r_1)$. Observe that in the matched case, i.e., $\hat P_0=P_0$ and $\hat P_1=P_1$, $\hat{Q}_{\lambda_0}= \hat{Q}_{\lambda_1}$.
	
	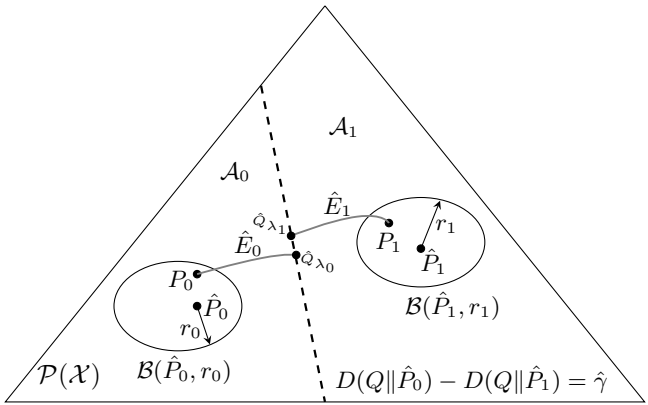
\begin{figure}[!h]
		\centering
		\begin{tikzpicture}[scale=0.85]
		\draw (5,1.2) -- (0,-5) -- (10,-5) --(5,1.2)  ;
		\draw [line width=0.3mm, dashed] (4,-0.05) --   (5.,-5)  ;
		\node at (7.29,-4.65) {\small $D(Q\|\hat{P}_0)-D(Q\|\hat{P}_1) =  \hat{\gamma}$};
		\node[draw,circle,inner sep=1pt,fill] at (3,-3.5) {};
		\node at (3.3,-3.5) {\small $ \hat{P}_0$};
		\node[draw,circle,inner sep=1pt,fill] at (6.5,-2.6) {};
		\node at (6.7,-2.8) {\small $\hat{P}_1$};
		\draw (2.7,-3.5) ellipse (1cm and 0.7cm);
		\draw (6.5,-2.5) ellipse (1 cm and 0.7 cm);
		\node at (1,-4.6) {$\Pc({\Xc})$};
		\draw [->,>=stealth] (6.5,-2.6) -- (6.8,-1.85);
		\draw [->,>=stealth] (3,-3.5) -- (3.2,-4.1);
		\node at (2.9,-3.9) {\small $r_0$};
		\node at (6.9,-2.25) {\small $r_1$};
		\node[draw,circle,inner sep=1pt,fill] at (3,-3) {};
		\node[draw,circle,inner sep=1pt,fill] at (6,-2.2) {};
		\draw [line width=0.25mm, gray]plot [smooth, tension=1] coordinates{(3,-3)  (4,-2.75) (4.55,-2.7)} ;
		\draw [line width=0.25mm, gray]plot [smooth, tension=1] coordinates{ (6,-2.2)  (5.5,-2.1)  (4.47,-2.4)} ;
		\node[draw,circle,inner sep=1pt,fill] at (4.55,-2.7) {};
		\node[draw,circle,inner sep=1pt,fill] at (4.47,-2.4) {};
		\node at (4.9,-2.8) {\tiny $\hat{Q}_{\lambda_0}$};
		\node at (4.15,-2.2) {\tiny $\hat{Q}_{\lambda_1}$};
		\node at (3.8,-2.52) {\small $\hat E_0$};
		\node at (5.2,-1.9) {\small $\hat E_1$};
		\node at (2.75,-3.1) {\small $P_0$};
		\node at (6,-2.5) {\small $P_1$};
		\node at (3.6,-1.4) {\small $\Ac_0$};
		\node at (5.3,-0.7) {\small $\Ac_1$};
		\node at (2.8,-4.5) {\small $\Bc(\hat P_0,r_0)$};
		\node at (7,-3.5) {\small $\Bc(\hat P_1,r_1)$};
		\end{tikzpicture}
		\caption{Mismatched likelihood ratio test with real distributions in divergence balls $\Bc(\hat P_0,r_0), \Bc(\hat P_1,r_1)$. } 
		\label{fig:mismatch}
	\end{figure}

	Furthermore, since the R\'enyi divergence $D_\alpha(Q\|P)$ is convex in $P$ for $\alpha\geq0$ \cite{Harremos}, and $f$-divergence $D_{f}(Q\|P)$ is convex in $P$ \cite{Csiszar}, then \eqref{eq:MMlower1} is a convex optimization problem and the KKT conditions are also sufficient. In addition, for the relative entropy, writing the Lagrangian gives
	\begin{align}\label{eq:lagrange}
	L(Q,P_0,&\lambda_0,\lambda_0', \nu_0, \nu_0')= D(Q\|P_0) + \lambda_0 \big( D(Q\|\hat{P}_1) \nonumber \\
	&-D(Q\|\hat{P}_0) +\hat{\gamma} \big) +  \lambda_0' \big ( D(\hat{P}_0\|P_0) -r_0\big )  \nonumber \\
	&+ \nu_0 \Big ( \sum_{x\in \Xc} Q(x)-1\Big )+ \nu_0' \Big ( \sum_{x\in \Xc} P_0(x)-1\Big ).
	\end{align}
	where $\lambda_0,\lambda_1, \nu_0,\nu_1$ are the Lagrange multipliers corresponding to the optimization \eqref{eq:MMlower1} constraints. Differentiating with respect to $Q(x)$ and  $P_0(x)$ and setting the derivatives to zero we have
	\begin{align}
	1+\log \frac{Q(x)}{P_0(x)} +\lambda_0 \log \frac{\hat{P}_0(x)}{\hat{P}_1(x)} + \nu_0&=0,\label{eq:lagrange1}\\
	-\frac{Q(x)}{P_0(x)}-\lambda_0' \frac{\hat{P}_0(x)}{P_0(x)}+\nu_0'&=0, \label{eq:lagrange2}
	\end{align} 
	respectively. Solving equations \eqref{eq:lagrange1}, \eqref{eq:lagrange2} for every $x\in\Xc$  we obtain
	\begin{align}\label{eq:lowerworstKKT1}
	\underline{Q}_{\lambda_0}(x)&= \frac{ \underline{P}_0(x) \hat{P}_0^{-\lambda_0}(x) \hat{P}_1^{\lambda_0}(x) } {\sum_{a \in \Xc } 	\underline{P}_0(a) \hat{P}_0^{-\lambda_0}(a) \hat{P}_1^{\lambda_0}(a) },\\
	\underline{P}_0(x)&=\frac{1}{1+\lambda_0'} \underline{Q}_{\lambda_0}(x) + \Big(1-\frac{1}{1+\lambda_0'}\Big) \hat{P}_0(x),	\label{eq:lowerworstKKT11}
	\end{align}
     where $ \lambda_0 \geq 0$. Moreover, from the complementary slackness condition \cite{Boyd} if  for all  $P_0$ in $\Bc(\hat{P}_0,r_0)$,   $D(P_0\|\hat{P}_0)- D(P_0\|\hat{P}_1) < \hat{\gamma}$  then 	     
     \begin{align}
     D(\underline{Q}_{\lambda_0}\|\hat{P}_0)-D(\underline{Q}_{\lambda_0} \| \hat{P}_1) &=\hat{ \gamma},\\
     D(\hat{P}_0\| \underline{P}_0) &= r_0,
     \end{align}
     Otherwise,  if there exists a  $\underline{P}_0$  in $\Bc(\hat P_0,r_0) $ such that $D(\underline{P}_0\|\hat{P}_0)- D(\underline{P}_0\|\hat{P}_1) \geq \hat{\gamma}$, then for this distribution $\hat{E}_0=0$. Therefore, if
     	\begin{equation}\label{eq:gammaworsL}
     \max_{P_0 \in\Bc(\hat P_0,r_0) }  D(P_0\| \hat{P}_0) - D(P_0\|\hat{P}_1)  < \hat{\gamma}
     \end{equation}
      holds, for all $P_0$ in the relative entropy ball, then  $\underline{\hat{E}}_0(r_0) >0$, otherwise $\underline{\hat{E}}_0(r_0) =0$.
	 Similar steps hold for the second hypothesis by only substituting the distributions.

Next we will study how the worst-case error exponents $(\underline{\hat{E}}_0, \underline{\hat{E}}_1)$ behave when the divergence ball radii $r_0,r_1$ are small. In particular, we  derive a Taylor series expansion of the worst-case error exponent, when the true distributions $P_0, P_1$ are within a R\'enyi entropy ball of radii $r_0,r_1$ centered at the testing distributions $\hat P_0,\hat P_1$. This approximation can also be interpreted as the worst-case sensitivity of the test, i.e., how does the test perform when actual distributions are very close to the mismatched distributions.    

\begin{theorem}\label{thm:lowerworst}
Consider a hypothesis testing setting with mismatch, with true distributions $P_0, P_1$ and testing distributions $\hat P_0,\hat P_1$. 
	For every 	$r_i \geq 0$, where $i\in \{0,1\}$, and threshold $\hat \gamma$ satisfying
	\begin{equation}\label{eq:threshcodsen}
	-D(\hat{P}_0\|\hat{P}_1) \leq \hat{\gamma} \leq D(\hat{P}_1\|\hat{P}_0),
	\end{equation}
	 the worst-case error exponents $\underline{\hat{E}}_i(r_i)$ can be expressed as
	\begin{equation}\label{eq:worstapproxlrt}
	\underline{\hat{E}}_i (r_i) = E_i -  \sqrt{r_i \cdot \theta_i(\hat{P}_0,\hat{P}_1,\hat{\gamma})}+ o \big(\sqrt{r_i} \big),  
	\end{equation}
	where	
	\begin{equation}\label{eq:sensitivity}
	\theta_i(\hat{P}_0,\hat{P}_1,\hat{\gamma}) =\frac{2}{\alpha}  {\rm Var}_{\hat{P}_i} \bigg(\frac{\hat{Q}_{\lambda}(X)}{\hat{P}_i(X)}  \bigg) 
	\end{equation}
	and $\hat{Q}_{\lambda}$ is the minimizing distribution in \eqref{eq:tilted} for test $\hat{\phi}$. 
\end{theorem}

Observe that as a result of the $\sqrt{r}$ expansion, the slope of the exponent for small $r$ tends to infinity, which implies that the likelihood ratio test is very sensitive to mismatch. In addition, observe that the sensitivity terms $\theta_i(\hat{P}_0,\hat{P}_1,\hat{\gamma})$ in \eqref{eq:sensitivity} can also be expressed as the chi-squared distance between $\hat{Q}_{\lambda}$ and $\hat{P}_i$.

\begin{proposition}\label{cor:varderivative}
	For every $\hat{P}_0,\hat P_1 \in \Pc(\Xc)$,  and $\hat{\gamma}$ satisfying \eqref{eq:threshcodsen}
	\begin{align}
	\frac{\partial }{\partial \hat{\gamma}}\theta_0(\hat{P}_0,\hat{P}_1,\hat{\gamma}) \geq 0, ~~~ 	\frac{\partial }{\partial \hat{\gamma}}\theta_1(\hat{P}_0,\hat{P}_1,\hat{\gamma}) \leq 0.
	\end{align}
\end{proposition}

This proposition shows that $\theta_0(\hat{P}_0,\hat{P}_1,\hat{\gamma})$  is a non-decreasing  function of $\hat{\gamma}$, i.e., as $\hat{\gamma}$ increases from $-D(\hat{P}_0\|\hat{P}_1) $ to $D(\hat{P}_1\|\hat{P}_0)$, the worst-case exponent $\underline{\hat E}_0(r_0)$ becomes more sensitive to mismatch. Conversely, $\theta_1(\hat{P}_0,\hat{P}_1,\hat{\gamma})$  is a non-increasing  function of $\hat{\gamma}$, i.e., as $\hat{\gamma}$ increases from $-D(\hat{P}_0\|\hat{P}_1) $ to $D(\hat{P}_1\|\hat{P}_0)$, the worst-case exponent $ \underline{\hat E}_1(r_1)$ becomes less sensitive (more robust) to mismatch. Moreover, when $\lambda=\frac{1}{2}$, we have 	
\begin{equation}
\hat{Q}_{\frac{1}{2}}(x)=\frac{\sqrt{\hat{P}_0(x) \hat{P}_1(x) }}{ \sum_{a \in \Xc } \sqrt{\hat{P}_0(a) \hat{P}_1(a) }    },
\end{equation}
and then $\theta_0(\hat{P}_0,\hat{P}_1,\hat{\gamma})=\theta_1(\hat{P}_0,\hat{P}_1,\hat{\gamma})$. In addition,  $\hat{Q}_{\frac{1}{2}}$ minimizes ${E}_0 + {E}_1$  yielding \cite{Veeravalli}
\begin{align}
{E}_0 + {E}_1&=D(\hat{Q}_{\frac{1}{2}}\|\hat{P}_0) + D(\hat{Q}_{\frac{1}{2}}\| \hat{P}_1)\\
&=\min_{Q \in \Pc(\Xc)} D(Q\|\hat{P}_0) +  D(Q\|\hat{P}_1)\\& = 2 B(\hat{P}_0,\hat{P}_1) 
\end{align}
where $B(\hat{P}_0,\hat{P}_1)=-\log\sum_{x\in\Xc}\sqrt{\hat{P}_0(x)\hat{P}_1(x)}$ is the Bhattacharyya distance between the mismatched distributions $\hat{P}_0$ and $\hat{P}_1$. This suggests that having equal sensitivity (or robustness) for both hypotheses minimizes the sum of the exponents.


\begin{example}
	When $\gamma=0$ the likelihood ratio test becomes the maximum-likelihood test, which is known to achieve the lowest average probability of error in the Bayes setting for equal priors. For fixed priors $\pi_0,\pi_1$, the error probability in the Bayes setting is
	$\bar\epsilon= \pi_0\epsilon_0 +\pi_1\epsilon_1$,
	resulting in the following error exponent \cite{Cover}
	\begin{equation}
	\bar E= \lim_{n\rightarrow \infty} -\frac{1}{n} \log \bar \epsilon = \min \{E_0,E_1\},
	\end{equation}
	assuming that the priors $\pi_0,\pi_1$ are independent of $n$.
	Consider $\hat{P}_0 =\text{Bern}(0.1)$ ,    $\hat{P}_1 =\text{Bern}(0.8)$. Also, assume $r_0=r_1=r$. Figure \ref{fig:worstRsen} shows the worst-case error exponent in the Bayes setting given by $\min \{\underline{\hat E}_0,\underline{\hat E}_1\}$ by solving \eqref{eq:MMlower1} and \eqref{eq:MMlower2} as well as $ \min \{\underline{\tilde{E}}_0,\underline{\tilde{E}}_1\}  $ by  the approximations in   \eqref{eq:worstapproxlrt} for the R\'enyi divergence balls of order $\alpha\in \{\frac{1}{2},1,2 \}$. Similarly, Figure \ref{fig:worstRsenfdiver} shows the worst-case error exponent in the Bayes setting for two $f$-divergences, $\chi^2$, and Hellinger divergences.  We can see that the approximation is good, especially for small radii $r$. 	
	
	Moreover, it can be seen that error exponents are very sensitive to mismatch for small $r$, i.e., the slope of the worst-case exponent goes to infinity as $r$ approaches to zero.
	
	\begin{figure}[htp]
		 \centering
		\input{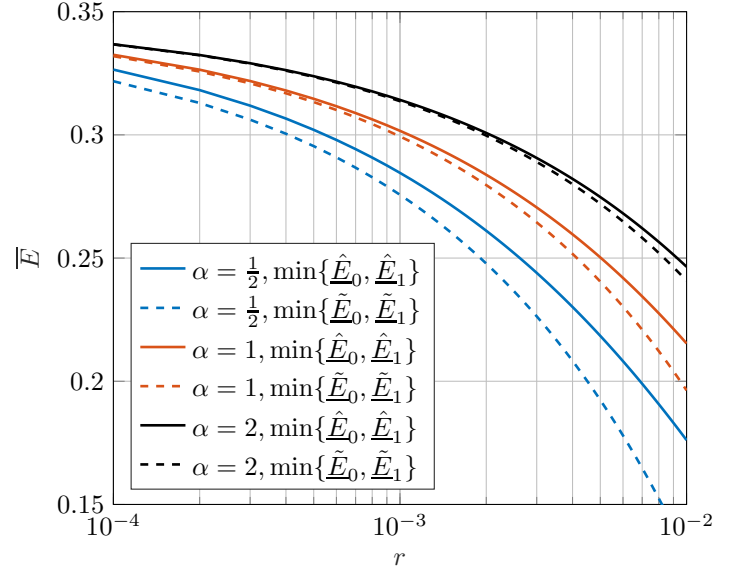} 
		\caption{Worst-case achievable Bayes error exponent for the R\'enyi divergence balls of order $\alpha\in \{\frac{1}{2},1,2 \}$. The solid lines correspond to the optimization problems in \eqref{eq:MMlower1}, \eqref{eq:MMlower2} and the dashed lines correspond to the approximated Bayes exponent using Theorem \ref{thm:lowerworst}. } 
		\label{fig:worstRsen}
	\end{figure} 

	\begin{figure}[htp]
		 \centering
	\input{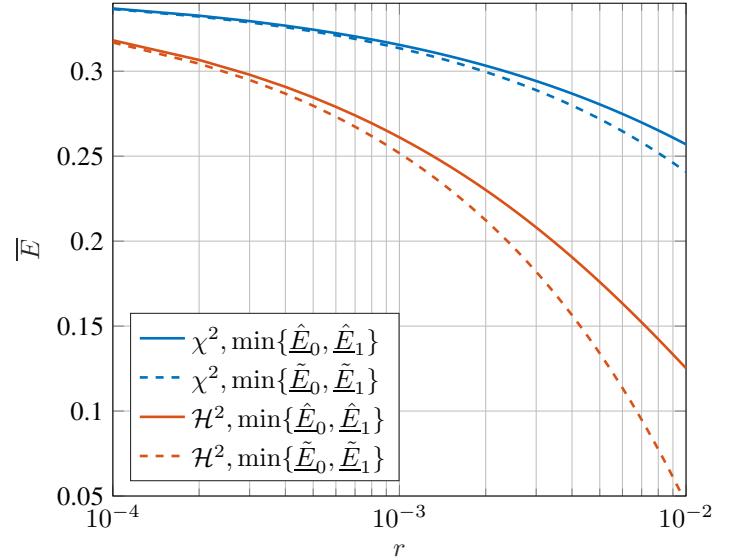} 
	\caption{Worst-case achievable Bayes error exponent for $\chi^2$ and Hellinger divergence balls. The solid lines correspond to the optimization problems in \eqref{eq:MMlower1}, \eqref{eq:MMlower2} and the dashed lines correspond to the approximated Bayes exponent using Theorem \ref{thm:lowerworst}. } 
	\label{fig:worstRsenfdiver}
\end{figure}

\end{example}


\section{Generalized Likelihood Ratio Testing Sensitivity }
\label{sec:glrt}

Next, we study the performance of Hoeffding's test under mismatch. Similarly to previous section, $P_0$ denotes the actual distribution that generated the observation and $\hat{P}_0$ indicates the mismatched distribution used in the test. Hoeffding's test using the mismatched distribution $\hat{P}_0$ with threshold $\hat{\gamma}$ is given by
\begin{equation}
\hat{\phi}(\Th)= \mathbbm{1}\{ D(\Th\|\hat{P}_0) \geq \hat \gamma     \}. 
\end{equation}

Note that since the original  test by Hoeffding does not depend on the second distribution (cf. \eqref{eq:GLRTtype}), the test is also independent of the second probability distribution in the mismatched case. Therefore, by Sanov's theorem, for every $P_0, P_1$ the error exponent $\hat{E}_1$ is equal to
\begin{equation}
\hat{E}_1=\min_{Q\in \Pc(\Xc): D(Q\|\hat{P}_0)\leq \hat{\gamma}} D(Q\|P_1). 
\end{equation}
The above optimization is a convex problem and by KKT conditions the minimizer is the tilted distribution between $\hat{P}_0$ and $P_1$ given by,
\begin{equation}\label{eq:tiltedmu}
\hat{Q}_{\mu}(x)= \frac{ \hat{P}_{0}^{\frac{\mu}{1-\mu}}(x) P_{1}^{\frac{1}{1-\mu}}(x) } {\sum_{a \in \Xc }  \hat{P}_{0}^{\frac{\mu}{1-\mu}}(a) P_{1}^{\frac{1}{1-\mu}}(a) },
\end{equation}
and where $\mu$ is the solution to
\begin{equation}\label{eq:KKTgamma} 
D(\hat{Q}_{\mu}\| \hat{P}_0)= \hat{\gamma}.
\end{equation}	
Similarly, by Sanov's theorem, the error exponent $\hat{E}_0$ is 
\begin{equation}
\hat{E}_0=\min_{Q\in \Pc(\Xc): D(Q\|\hat{P}_0)\geq \hat{\gamma}} D(Q\|P_0). 
\label{eq:e0hoeffmism}
\end{equation}
It is clear that if $D(P_0\|\hat{P}_0)\geq \hat{\gamma}$ then $\hat{E}_0=0$. Hence, we assume $D(P_0\|\hat{P}_0) < \hat{\gamma}$. Unfortunately, the solution to the above error exponent cannot be derived by convex optimization since the constraint is the complement of a convex set. In the following, we introduce an upper bound to the achievable exponents.

\begin{theorem}\label{thm:hoeffupper} 
	For fixed $\hat{P}_0, P_0 \in \Pc(\Xc)$ the error exponent $\hat{E}_0$ of Hoeffding's test with mismatch is upper bounded by 
	\begin{equation}
	\hat{E}_0 \leq (\hat{\gamma} - \sqrt{2\hat{\gamma}} \|\hat{P}_0-P_0\|_{{\rm TV}})^+, 
	\end{equation} 
	where $\|\cdot\|_{{\rm TV}}$ is the total variational distance, and $(x)^+=\max(0,x)$. 
\end{theorem}

Observe from Theorem \ref{thm:hoeffupper} that the highest achievable exponent in Hoeffding's test is equal to the achievable exponent when $\hat{P}_0=P_0$, i.e., the mismatch will always result in suboptimal error exponent tradeoff. However, the  likelihood ratio testing can still achieve an optimal error exponent tradeoff under mismatch if the mismatched distributions are tilted version of actual distributions. The universality of Hoeffding's test can explain the higher sensitivity of Hoeffding's test toward mismatch.

For the small thresholds $\hat{\gamma}$ we can also approximate \eqref{eq:e0hoeffmism} by Taylor expanding the $D(Q\|P_0), D(Q\| \hat{P}_0)$ to get  
\begin{equation}
\hat{E}_0=\min_{\substack{ \frac{1}{2} \hat{\thetav}^T \hat{\Jm} \hat{\thetav}  \geq \hat{\gamma} \\ \onev^T \hat{\thetav}=0} } \frac{1}{2} \thetav^T \Jm \thetav + o(\| \thetav\|^2)+o(\| \hat{\thetav} \|^2)
\end{equation}
where 
\begin{align}
\thetav&= \big(Q(x_1)-{P}_0(x_1),\dotsc,Q(x_{|\Xc|})-{P}_0(x_{|\Xc|})\big)^T,\\
\hat{\thetav}&= \big(Q(x_1)-\hat{P}_0(x_1),\dotsc,Q(x_{|\Xc|})-\hat{P}_0(x_{|\Xc|})\big)^T,\\
\Jm&=\diag\bigg( \frac{\alpha}{{P}_0(x_1)},\dotsc,\frac{\alpha}{{P}_0(x_{|\Xc|})}\bigg),\\
\hat{\Jm}&=\diag\bigg( \frac{\alpha}{\hat{P}_0(x_1)},\dotsc,\frac{\alpha}{\hat{P}_0(x_{|\Xc|})}\bigg).
\end{align}	
Observe that, since \eqref{eq:e0hoeffmism} is minimizing a convex function over the complement of the convex set, the optimizing distribution is always on the boundary of the set and hence $D(Q^*\|\hat{P}_0)=\hat{\gamma}$ where $Q^*$ is the optimizing distribution. Therefore, we can conclude that $\|Q^*-\hat{P}_0\| = O(\sqrt{\hat{\gamma}})$. Similarly by the assumption  $D(P_0 \|\hat{P}_0) \leq \hat{\gamma}$, we get that $\|P_0-\hat{P}_0\| = O(\sqrt{\hat{\gamma}})$, and hence by the triangle inequality $\|Q^*-{P}_0\| = O(\sqrt{\hat{\gamma}})$. Substituting these into the error term we get that the approximation error is $o(\hat{\gamma})$. Next by replacing $Q(x_{|\Xc|})=1- \sum_{x=1}^{|\Xc|-1} Q(x)$ and  dropping the equality condition we have
\begin{equation}
\hat{E}_0=\min_{ -\frac{1}{2} {\Qv}^T \hat{\Hm} {\Qv} - \hat{\hv}^T \Qv +\hat{\gamma}+1 \leq 0 } \frac{1}{2} {\Qv}^T {\Hm} {\Qv} + \hv^T \Qv -1  + o(\hat{\gamma}) 
\end{equation}
where
\begin{align}
\Hm&^{(|\Xc|-1 \times |\Xc|-1)} \nonumber \\ 
&=\begin{pmatrix}
\frac{1}{P_0(1)} + \frac{1}{P_0(|\Xc|)} &  \ldots& \frac{1}{P_0(|\Xc|)}  \\
\frac{1}{P_0(|\Xc|)} &  \ddots& \frac{1}{P_0(|\Xc|)}   \\
\vdots & \ldots  & \vdots  \\
\frac{1}{P_0(|\Xc|)}   & \ldots   & \frac{1}{P_0(|\Xc|-1)} + \frac{1}{P_0(|\Xc|)}   \\
\end{pmatrix},\nonumber \\ 
&\hv^{(|\Xc|-1 \times 1)}= \Big( -\frac{1}{P_0(|\Xc|)}, \ldots , - \frac{1}{P_0(|\Xc|)} \Big)^T,   
\end{align}
and $\hat{\Hm}, \hat{\hv}$ are  ${\Hm}, {\hv}$ with $P_0$ replaced by $\hat{P}_0$. This optimization problem is quadratic optimization with a single quadratic constraint. Using the Schur complement to express the dual problem, we obtain \cite[Appendix B]{Boyd} 
\begin{equation}
\hat{E}_0=\max_{ \substack{\lambda \geq0 \\ \Rm \succeq 0} } \nu 
\end{equation}
where 
\begin{equation}
\Rm =\frac{1}{2}  
\begin{pmatrix}
\Hm + \lambda \hat{\Hm} & \hv+ \lambda \hat{\hv}\\
  \hv^T+ \lambda \hat{\hv}^T & -2 +2\lambda( 1+\hat{\gamma}) -\nu 
\end{pmatrix}.
\end{equation}
This optimization problem is convex as the dual problem is always concave. In addition the strong duality holds for this problem when the Slater's condition is met that is there exist a $\Qv$ such that the $-\frac{1}{2} {\Qv}^T \hat{\Hm} {\Qv} - \hat{\hv}^T \Qv +\hat{\gamma}+1 < 0$ \cite{Boyd}. Therefore, one can find the mismatched type-\RNum{1} error exponent for $\hat{\gamma}$ small enough by conventional convex optimization methods.

We now focus on the worst-case error-exponent performance of the mismatched Hoeffding test when the distributions generating the observation fulfil \eqref{eq:ball2}, i.e., they are inside a divergence ball of radii $r_0, r_1$.  Figure \ref{fig:mismatchGLRT} illustrates the mismatched probability distributions and the mismatched Hoeffding test as the relative entropy ball centered at $\hat{P}_0$ divides the probability space into the two decision regions.

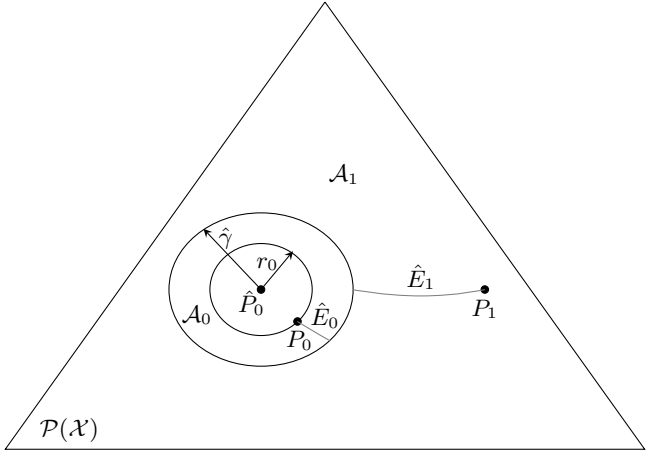
\begin{figure}[!h]
	\centering
	\begin{tikzpicture}[scale=0.85]
	\draw (5,2) -- (0,-5) -- (10,-5) --(5,2)  ;
	\node[draw,circle,inner sep=1pt,fill] at (4,-2.5) {};
	\node at (3.85,-2.7) {\small $ \hat{P}_0$};
	\node[draw,circle,inner sep=1pt,fill] at (7.5,-2.5) {};
	\node at (7.5,-2.8) {\small $P_1$};
	\draw (4,-2.5) ellipse (1*0.8 cm and 0.9*0.8cm);
	\draw (4,-2.5) ellipse (1.8*0.8 cm and 1.5*0.8cm);
	\node at (1.,-4.7) {\small $\Pc({\Xc})$};
	\draw [->,>=stealth] (4,-2.5) -- (4.5,-1.9);
	\node at (4.1,-2.1) {\small $r_0$};
	\draw [->,>=stealth] (4,-2.5) -- (3.1,-1.55);
	\node at (3.45,-1.7) {\small $\hat{\gamma}$};
	\node[draw,circle,inner sep=1pt,fill] at (4.57,-3) {};
	\draw [gray]plot [smooth, tension=1] coordinates{(4.57,-3)  (4.9,-3.2) (5.08,-3.3)} ;
	\draw [gray]plot [smooth, tension=1] coordinates{ (5.43,-2.5)  (6.5,-2.6)  (7.5,-2.5)} ;
	\node at (5,-2.9) {\small $\hat{E}_0$};
	\node at (6.5,-2.3) {\small $\hat{E}_1$};
	\node at (4.6,-3.3) {\small $P_0$};
	\node at (3.,-2.9) {\small $\Ac_0$};
	\node at (5.3,-0.7) {\small $\Ac_1$};
	\end{tikzpicture}
	\caption{Mismatched Hoeffding's test with real distribution $P_0$ from a divergence ball $\Bc(\hat{P}_0,r_0)$.} 
	\label{fig:mismatchGLRT}
\end{figure}

	For every $P_0$ the achievable error exponent $\hat{E}_0$  does not depend on $P_1$ therefore, for every  $r_0$, the least favorable exponents  $\underline{\hat{E}}_0(r_0),$ defined in \eqref{eq:e0hoeffmism} can be written as
\begin{align}\label{eq:worstHoef}
\underline{\hat{E}}_0(r_0)&=\min_{\substack{P_0:P_0 \in \Bc(\hat P_0,r_0) \\ Q\in \Pc(\Xc):  D(Q\|\hat{P}_0) \geq \hat{\gamma}}} D(Q\|P_0),
\end{align}
where $\Bc(\hat P_0,r_0)$ is the divergence ball centered at $\hat P_0$ with radius $r_0$ \eqref{eq:ball2} with divergence measure parametrized by $\alpha$. 
As opposed to the mismatched likelihood ratio test where the worst achievable exponent could be found by solving a convex problem, here, the optimization problem in \eqref{eq:worstHoef} is non-convex and  in principle difficult to solve. However, as the next theorem states, we are still able to perform a Taylor series expansion to find the behavior of the worst exponent $\underline{\hat{E}}_0(r_0)$  when the radius of  the relative entropy ball  $r_0$ is small. 

\begin{theorem}\label{thm:lowerworstHoef}
	Consider a mismatched generalized likelihood ratio test with real and test distributions $P_0,P_1$ and $\hat P_0$, respectively. For every $ r_0 \geq 0$, we have that the error exponent $\underline{\hat{E}}_0(r_0)$ can be approximated as
	\begin{equation}\label{eq:worstapproxHoef}
	\underline{\hat{E}}_0(r_0) = E_0 -  \sqrt{r_0\cdot \theta_0(\hat{P}_0,\hat{\gamma})}+o(\sqrt{r_0}),  
	\end{equation}
	where	
	\begin{equation}\label{eq:hoefsen}
	\theta_0(\hat{P}_0,\hat{\gamma}) = \max_{\substack{ Q: D(Q\|\hat{P}_0) = \hat{\gamma}}}\frac{2}{\alpha}  {\rm Var}_{\hat{P}_0} \bigg(\frac{Q(X)}{\hat{P}_0(X)}  \bigg),
	\end{equation}	
	is the sensitivity of Hoeffding's test with mismatch.
\end{theorem}

Observe that whole the expressions \eqref{eq:sensitivity} and \eqref{eq:hoefsen} are structurally similar,  \eqref{eq:hoefsen}  has an additional optimization step.  The following result compares the sensitivities of the worst-case mismatched likelihood ratio and Hoeffding's test sensitivities.

\begin{proposition}\label{cor:comparing}
	Let $\hat{P}_0$ be fixed and $\hat{P}_1$ be some distribution used in the likelihood ratio test. Also, let Hoeffding's test sensitivity denoted by $\theta_0^{\rm h}(\hat{P}_0,\hat\gamma^{\rm h})$, and $\theta_0^{\rm lrt}(\hat{P}_0,\hat{P}_1,\hat\gamma^{\rm lrt})$ be the sensitivity of likelihood ratio test when the threshold $\hat\gamma^{\rm lrt}$ is chosen such that the type-\RNum{1} error exponent is equal to $\hat \gamma^{\rm h}$. Then, we have
	\begin{equation}\label{eq:lrttohoef}
	1 \leq \frac{\theta_0^{\rm h}(\hat P_0,\hat\gamma^{\rm h})}{\theta_0^{\rm lrt}(\hat P_0,\hat P_1,\hat\gamma^{\rm lrt})}\leq \sqrt{\frac{(h-1)^2} {h\log h +1-h} }.
	\end{equation}
	where 
	\begin{equation}
	h=\frac{\min\big\{1, \hat{\gamma}^{\rm h} +\min_{x\in \Xc} \hat{P}_0(x)\big\}  }{\min_{x\in \Xc} \hat{P}_0(x) }.
	\end{equation}
	Also, we have the weaker inequality
	\begin{equation}\label{eq:cor}
	\frac{\theta_0^{\rm h}(\hat P_0,\hat\gamma^{\rm h})}{\theta_0^{\rm lrt}(\hat P_0,\hat P_1,\hat\gamma^{\rm lrt})}\leq \sqrt{\frac{4}{\min_{x\in \Xc} \hat{P}_0(x)}}.
	\end{equation}
\end{proposition}


\section{Sequential Probability Ratio Testing Sensitivity }
\label{sec:sequentialHT}

The sequential probability ratio test $\hat{\Phi}=(\hat{\phi},\hta)$ with thresholds $\hgo, \hgt$ and mismatched distributions $\hpo, \hpt$ is given by
\begin{align}
\hta=\inf \{n\geq1: \hat{S}_n& \geq  \hgo \  \mathrm{or} \  \hat{S}_n \leq -\hgt\},  
\end{align} 
where  
\begin{align}
\hat{S}_n=\sum_{i=1}^n \log \frac{\hat{P}_0(x_i)}{\hat{P}_1(x_i)},
\end{align} 
and
\begin{align}
\hat{\phi}= 
\begin{cases}
0 &  \text{if } \hat{S}_{\hta} \geq \hgo  \\
1 &  \text{if } \hat{S}_{\hta} \leq  - \hgt.\\
\end{cases}
\end{align}   

Similarly to the previous sections, in order to study the sensitivity of the mismatched sequential ratio test, we first study the highest achievable error exponents, i.e,
\begin{align}\label{eq:tradeseqMM1}
 \hEt(&\hEo) \triangleq \sup \Big \{\hEt \in \mathbb{R}_{+}: \exists \hgo,\hgt , \exists \ n  \in \ZZ_{+}  \text{ s.t.}   \mathbb{E}_{P_0} [\hta]    \leq n, \nonumber \\ 
  &\mathbb{E}_{P_1} [\hta]  \leq n, \   \epsilon_0(\hat{\Phi}) \leq 2^{- n \hEo}    \quad \text{and} \quad  \epsilon_1(\hat{\Phi}) \leq 2^{- n\hEt } \Big \} ,
\end{align}
which is analogous to the definition in \eqref{eq:tradeseq1}. Similarly to \eqref{eq:tradeseq2}, we can also define the following tradeoff
\begin{align}\label{eq:tradeseqMM2}
 \hEt&(\hEo)\notag\\ & \triangleq \sup \Big \{\hEt \in \mathbb{R}_{+}: \exists \hgo, \hgt , \exists n_0, n_1  \in \ZZ_{+}, \text{s.t.}  \mathbb{E}_{P_0} [\hta]   \leq n_0, \nonumber \\  
 & \mathbb{E}_{P_1}[\hta]  \leq n_1, \   \epsilon_0(\hat{\Phi}) \leq 2^{- n_0 \hEo}  ~ \text{and} ~  \epsilon_1(\hat{\Phi}) \leq 2^{- n_1 \hEt} \Big \}.
\end{align}

The next theorem provides the error exponents $\hEo,\hEt$ and the average stopping time $\mathbb{E}_{P_0}[\hta], \mathbb{E}_{P_1}[\hta]$ of the mismatched sequential probability ratio test as a function of thresholds $\hgo, \hgt$.

\begin{theorem}\label{thm:seqMM}
	For fixed probability measures $\hat{P}_0, \hat{P}_1$, let $P_0$ and $P_1$ be such that
\begin{equation}\label{eq:posdrift}
0<D(P_0\|\hat{P}_1)-D(P_0\|\hat{P}_0),  0<D(P_1\|\hat{P}_0)-D(P_1\|\hat{P}_1).
\end{equation}
Then, as  $\hgo,\hgt \rightarrow \infty$, the pairwise probabilities  of error $\heo,\het$ are given by
\begin{align}
\heo &= \hat c_0 \cdot e^{-\frac{D(P_0\|P_1)}{D(P_0\|\hat{P}_1)-D(P_0\|\hat{P}_0)}\hgt  } \label{eq:MMexp1}, \\  
\het &= \hat c_1 \cdot e^{-\frac{D(P_1\|P_0)}{D(P_1\|\hat{P}_0)-D(P_1\|\hat{P}_1)}\hgo  } \label{eq:MMexp2},
\end{align}
where $\hat c_0, \hat c_1$ are positive constants. Furthermore, the expected stopping times are given by
\begin{align} \label{eq:SPRTthresh}
\mathbb{E}_{P_0}[\hta]&=\frac{ \hgo}{D(P_0\|\hat{P}_1)-D(P_0\|\hat{P}_0)}(1+o(1)),\\ 
\mathbb{E}_{P_1}[\hta]&=\frac{ \hgt}{D(P_1\|\hat{P}_0)-D(P_1\|\hat{P}_1)} (1+o(1)).
\end{align}
\end{theorem}

The next result states that if the average drift of the likelihood ratio changes sign under mismatch, the probability of error under that hypothesis tends to one.

	\begin{theorem} \label{thm:negdrift}
	For fixed $\hat{P}_0, \hat{P}_1$, let $P_0$ be such that 
	\begin{equation}\label{eq:negdrift}
	D(P_0\|\hat{P}_1)-D(P_0\|\hat{P}_0)<0.
	\end{equation}
	Then, as  thresholds $\hgo, \hgt$  approach infinity, $\hat{\epsilon}_0 \rightarrow 1$. Similarly, letting $P_1$ such that 
	\begin{equation}
	D(P_1\|\hat{P}_0)-D(P_1\|\hat{P}_1)<0,
	\end{equation}
	 $\hat{\epsilon}_1  \rightarrow 1$ as  thresholds $\hgo, \hgt$  approach infinity. 
\end{theorem}

\begin{corollary}\label{cor:exp}
Under the conditions of Theorem \ref{thm:seqMM}, the achievable error exponent tradeoff according to \eqref{eq:tradeseqMM1} is given by
\begin{align}
\hEo &= D(P_0\|P_1) \frac{  D(P_1\|\hat{P}_0)-D(P_1\|\hat{P}_1) }{D(P_0\|\hat{P}_1)-D(P_0\|\hat{P}_0) }  , \\
\hEt&= D(P_1\|P_0)\frac{D(P_0\|\hat{P}_1)-D(P_0\|\hat{P}_0) }{D(P_1\|\hat{P}_0)-D(P_1\|\hat{P}_1)} ,
\end{align}	
where to achieve these exponents thresholds $\hgo, \hgt$ should be chosen as
\begin{align}\label{eq:thresh1}
\hgo&=n\big ( D(P_0\|\hat{P}_1)-D(P_0\|\hat{P}_0) +o(1)\big), \\
\hgt&= n\big(D(P_1\|\hat{P}_0)-D(P_1\|\hat{P}_1) +o(1)\big). \label{eq:thresh2}
\end{align}
Moreover, the achievable error exponents according to \eqref{eq:tradeseqMM2} satisfy
\begin{equation}
\hEo = \ell D(P_0\|P_1)  , \quad \hEt= \frac{1}{\ell} D(P_1\|P_0),
\end{equation}
where $\ell= \frac{D(P_1\|\hat{P}_0)-D(P_1\|\hat{P}_1)}{D(P_0\|\hat{P}_1)-D(P_0\|\hat{P}_0)}\frac{n_1}{n_0} $. Equivalently, we have that
\begin{equation} \label{eq:seqexpMM}
\hEo 	\hEt = D(P_0\|P_1)D(P_1\|P_0).
\end{equation}
To achieve \eqref{eq:seqexpMM}, thresholds $\hgo, \hgt$ should be chosen as
\begin{align}
\hgo&=n_0\big (D(P_0\|\hat{P}_1)-D(P_0\|\hat{P}_0)+o(1)\big) \label{eq:threshMM1},\\
\hgt&=  n_1\big(D(P_1\|\hat{P}_0)-D(P_1\|\hat{P}_1))+o(1)\big) \label{eq:threshMM2}.
\end{align}
\end{corollary}

By comparing \eqref{eq:seqexp}, \eqref{eq:seqexpMM} we can conclude that mismatched sequential probability ratio test has the same performance as the case with no mismatch, i.e., there exist thresholds $\hgo, \hgt$ such that the expected stopping time condition is met, and the error exponents satisfy \eqref{eq:seqexpMM}. The intuition behind the existence of thresholds such that the optimal tradeoff is achievable relies on the fact that the mismatched distributions only cause a change in the drifts of the random walk generated by $\hat{S}_n$, and hence one can choose the thresholds appropriately to rescale the random walk behavior to achieve the optimal exponents. However, choosing $\hgo, \hgt$ to achieve \eqref{eq:seqexpMM} requires the knowledge of true probability measures $P_0, P_1$ by \eqref{eq:threshMM1}, \eqref{eq:threshMM2}, which might not be  possible. 
Note that, in the case of the likelihood ratio test, the number of samples is fixed. Hence, we cannot change the speed of decision-making when using the mismatched distributions, and mismatched distributions change the drift of the log-likelihood ratio, but rescaling the number of samples described for the mismatched likelihood ratio is not possible; therefore, mismatch can reduce the exponents. Having this in mind, we consider the performance of the mismatched probability ratio test when the thresholds are selected from \eqref{eq:thresh1}, \eqref{eq:thresh2}, \eqref{eq:threshMM1} and \eqref{eq:threshMM2} but replacing $P_0,P_1$ by the mismatched measures $\hat P_0,\hat P_1$. In fact, this is precisely the relevant practical scenario where only the testing probability measures $\hat P_0,\hat P_1$ are available. In this scenario, mismatch in probability measures will induce a mismatch in expected stopping time and the error exponents. Consider the case where \eqref{eq:thresh} is used with mismatched measures $\hat P_0,\hat P_1$,
\begin{equation}\label{eq:thresh3}
\hgo=n\big (D(\hat{P}_0\|\hat{P}_1)+o(1)\big) ~,~ \hgt= n\big(D(\hat{P}_1\|\hat{P}_0)+o(1)\big). 
\end{equation}
Using \eqref{eq:thresh3} and by Theorem \ref{thm:seqMM} we obtain
\begin{align} 
\mathbb{E}_{P_0}[\hta]&=n\frac{ D(\hat{P}_0\|\hat{P}_1)}{D(P_0\|\hat{P}_1)-D(P_0\|\hat{P}_0)}(1+o(1)),\\ 
\mathbb{E}_{P_1}[\hta]&=n\frac{ D(\hat{P}_1\|\hat{P}_0)}{D(P_1\|\hat{P}_0)-D(P_1\|\hat{P}_1)} (1+o(1)).
\end{align}
Therefore, the mismatch in the thresholds, induces expected stopping times that may be larger than $n$. Letting $\eta^{-1}=\max \bigg \{\frac{ D(\hat{P}_0\|\hat{P}_1)}{D(P_0\|\hat{P}_1)-D(P_0\|\hat{P}_0)}, \frac{ D(\hat{P}_1\|\hat{P}_0)}{D(P_1\|\hat{P}_0)-D(P_1\|\hat{P}_1)} \bigg  \}$, and according to definition \eqref{eq:tradeseqMM1} we have the following exponents,
\begin{align}
\hEo &= \frac{D(P_0\|P_1)D(\hat{P}_1\|\hat{P}_0)}{D(P_0\|\hat{P}_1)-D(P_0\|\hat{P}_0)} \eta, \label{eq:threshMMM1} \\
\hEt  &= \frac{D(P_1\|P_0)D(\hat{P}_0\|\hat{P}_1)}{D(P_1\|\hat{P}_0)-D(P_1\|\hat{P}_1)} \eta   \label{eq:threshMMM2}.
\end{align}	
Similarly to \eqref{eq:threshell}, for the second definition of exponent, we need to multiply one of the thresholds by $\ell$, and the corresponding exponents will be equal to $\ell$ and $\frac{1}{\ell}$ times the above exponents.

We now analyze the worst-case error exponents, defined as
\begin{align}
\underline{\hat{E}}_i(r_i) \triangleq \min_{P_i\in\Bc(\hat P_i,r_i)} \hat E_i,~~ i\in \{0,1 \},
\label{eq:worst_case}
\end{align}
where $\Bc(Q,r)$ is the divergence ball of radius $r$ centered at distribution $Q$ defined in \eqref{eq:ball}.
From \eqref{eq:threshMMM1}, we observe  that error exponents of mismatched sequential probability ratio test are a function of both data distributions $ P_0,  P_1$, as opposed to the fixed sample-size setting where  $\hat E_0$ is independent of ${P}_1$. The next theorem shows the behavior of the worst-case exponents when the true distributions are within a small divergence ball of radii $r_0,r_1$ and center $\hat{P}_0,\hat{P}_1$, respectively.

\begin{theorem}\label{thm:lowerworstseq}
Let $P_i, \hat{P}_i$ are defined on the probability simplex $\Pc(\Xc)$ and  $r_i \geq 0$, for $i\in\{0,1\}$. Define $\bar\imath=1-i$ to be the complement of index $i$. Then, the worst-case error exponents can be approximated as 
	\begin{align}\label{eq:worstapproxseq}
	\underline{\hat{E}}_i (r_i) = E_i - \min\bigg \{ \sum_{j=0}^{1} & \sqrt{r_j\cdot \theta_{i,j}(\hat{P}_0,\hat{P}_1)},  \sqrt{r_{\bar\imath}\cdot \theta_{\bar\imath}(\hat{P}_0,\hat{P}_1)} \bigg  \} \nonumber \\ 
	&+ o \big(\sqrt{r_0} +\sqrt{r_1} \ \big),  
	\end{align}
	where	
	\begin{align}\label{eq:senseq}
	\theta_{i,j}(\hat{P}_0,\hat{P}_1) &= 
	\begin{cases}
	\frac{2}{\alpha}  {\rm Var}_{\hat{P}_i} \Big( \rho_i \log \frac{\hat{P}_i(X)}{\hat{P}_{\bar\imath}(X)}  \Big)  &i=j\\
	\frac{2}{\alpha}  {\rm Var}_{\hat{P}_j} \Big(  \rho_i \frac{\hat{P}_i(X)}{\hat{P}_j(X)}  \Big)  & i\neq j
	\end{cases}\\
	\theta_{\bar\imath}(\hat{P}_0,\hat{P}_1)&=  \frac{2}{\alpha} {\rm Var}_{\hat{P}_{\bar\imath}} \bigg( \log \frac{\hat{P}_{i}(X)}{\hat{P}_{\bar\imath}(X)} +\rho_i \frac{\hat{P}_i(X)}{\hat{P}_{\bar\imath}(X)}  \bigg),\\
	\rho_i &= \frac{D(\hat{P}_{\bar\imath}\|\hat{P}_i)}{D(\hat{P}_i\|\hat{P}_{\bar\imath})}.  
	\end{align}
\end{theorem}

Next assuming $r_0=r_1=r$, we obtain the following result.

\begin{corollary}\label{cor:seq}
	For every 	$r=r_0= r_1 \geq 0$,  and $i\in\{0,1\}, \bar\imath=1-i$,  
	\begin{align}\label{eq:seqcor}
	&	\underline{\hat{E}}_i (r) = E_i - \sqrt{r\cdot \theta_{\bar\imath}(\hat{P}_0,\hat{P}_1)} + o \big(\sqrt{r}  \big),  
	\end{align}
	where	
	\begin{equation}
	\theta_{\bar\imath}(\hat{P}_0,\hat{P}_1)=  \frac{2}{\alpha} {\rm Var}_{\hat{P}_{\bar\imath}} \bigg( \log \frac{\hat{P}_{i}(X)}{\hat{P}_{\bar\imath}(X)} +\rho_i \frac{\hat{P}_i(X)}{\hat{P}_{\bar\imath}(X)}  \bigg). 
	\end{equation}
\end{corollary}

As an example, consider $\hat{P}_0 =\text{Bern}(0.1)$ ,    $\hat{P}_1 =\text{Bern}(0.8)$, and $r=r_0=r_1$ and the relative entropy is used as the $f$-divergence ball measure of distance. Figure \ref{fig:worstRsenseq} shows the worst-case error exponent given by solving non-convex optimization problem in  \eqref{eq:worst_case} with precision of $10^{-3}$ as well as the approximation $ \underline{\tilde{E}}_0$ obtained from  \eqref{eq:seqcor} by ignoring the $o(\sqrt{r})$ term. Observe that there exists some gap between the approximation $\underline{\tilde{E}}_0$ and the actual exponent $\underline{\hat{E}}_0$ in \eqref{eq:worst_case}. The approximation consists of a linear approximation of the objective and second order approximation of constraints and computing it is straightforward for arbitrary distributions and radii. Instead, computing the  exact optimization problem $\underline{\hat{E}}_0$ \eqref{eq:worst_case} is difficult, as it is a non-convex optimization problem involving a highly nonlinear objective, cf. Eqs. \eqref{eq:threshMMM1}--\eqref{eq:worst_case}.


\begin{figure}[htp]
	\centering  
%
%
\definecolor{mycolor1}{rgb}{0.00000,0.44700,0.74100}%
\definecolor{mycolor2}{rgb}{0.85000,0.32500,0.09800}%
\begin{tikzpicture}[scale=1]

\begin{axis}[%
width=3.0in,
height=2.5812in,
at={(1.011in,0.742in)},
scale only axis,
xmode=log,
xmin=0.0001,
xmax=0.01,
xminorticks=true,
xlabel style={font=\color{white!15!black}},
xlabel={$r$},
ymin=0.8,
ymax=1.4,
ylabel style={font=\color{white!15!black}},
ylabel={$\underline{\hat{E}}_0$},
axis background/.style={fill=white},
xmajorgrids,
xminorgrids,
ymajorgrids,
legend style={legend cell align=left, align=left, draw=white!15!black}
]
\addplot [color=mycolor1, line width=1.0pt]
  table[row sep=crcr]{%
0.0001	1.31227716062863\\
0.0002	1.29195651417509\\
0.0003	1.27605519917042\\
0.0004	1.26361583567816\\
0.0005	1.252095683621\\
0.0006	1.24228963172206\\
0.0007	1.23254611697923\\
0.0008	1.22366904174256\\
0.0009	1.21564379645542\\
0.001	1.20845733322501\\
0.0011	1.20130497043724\\
0.0012	1.19418650300406\\
0.0013	1.1871017139778\\
0.0014	1.18083221787454\\
0.0015	1.17536798175278\\
0.0016	1.16914764216693\\
0.0017	1.16372619367038\\
0.0018	1.15755452583283\\
0.0019	1.1529426599594\\
0.002	1.14758038281358\\
0.0021	1.14223763789805\\
0.0022	1.13767360521721\\
0.0023	1.1323668859929\\
0.0024	1.12783359654455\\
0.0025	1.12331440914501\\
0.0026	1.11880926699545\\
0.0027	1.11431811373639\\
0.0028	1.10984088462205\\
0.0029	1.10612044870831\\
0.003	1.10166861269888\\
0.0031	1.09796926403956\\
0.0032	1.09354262765506\\
0.0033	1.08986422115381\\
0.0034	1.08619525845829\\
0.0035	1.08180496776186\\
0.0036	1.07815672018793\\
0.0037	1.07451784933864\\
0.0038	1.0708882979988\\
0.0039	1.06726807021863\\
0.004	1.06365709422935\\
0.0041	1.06005538528959\\
0.0042	1.05718064282627\\
0.0043	1.0535954816935\\
0.0044	1.05001950247883\\
0.0045	1.04716528026977\\
0.0046	1.04360570369175\\
0.0047	1.04076458947976\\
0.0048	1.037221343906\\
0.0049	1.03439324723111\\
0.005	1.03086623958699\\
0.0051	1.02805110269519\\
0.0052	1.02524169026515\\
0.0053	1.02173795919167\\
0.0054	1.01894139694631\\
0.0055	1.01615052328262\\
0.0056	1.01336532410344\\
0.0057	1.01058577623244\\
0.0058	1.00781187016741\\
0.0059	1.00435240759533\\
0.006	1.00159115931149\\
0.0061	0.998835512791084\\
0.0062	0.996085450680636\\
0.0063	0.993340953867761\\
0.0064	0.991286240708731\\
0.0065	0.988551463920247\\
0.0066	0.985822224806582\\
0.0067	0.983098510008174\\
0.0068	0.98038030623027\\
0.0069	0.9776676002425\\
0.007	0.975636668716021\\
0.0071	0.972933549340932\\
0.0072	0.97023589175112\\
0.0073	0.967543682953785\\
0.0074	0.965528087946333\\
0.0075	0.962845380756533\\
0.0076	0.960168086998889\\
0.0077	0.958163653988114\\
0.0078	0.955495800486314\\
0.0079	0.95349843476325\\
0.008	0.950839982919492\\
0.0081	0.948849655760831\\
0.0082	0.946200567297703\\
0.0083	0.944217250217884\\
0.0084	0.941577487174951\\
0.0085	0.939601151923624\\
0.0086	0.936970676654055\\
0.0087	0.935001295213591\\
0.0088	0.933034887939782\\
0.0089	0.930417614963857\\
0.009	0.928458121891911\\
0.0091	0.926501572121833\\
0.0092	0.923897434642371\\
0.0093	0.921947756232244\\
0.0094	0.920001018447377\\
0.0095	0.917409933388607\\
0.0096	0.915470031850755\\
0.0097	0.91353305447427\\
0.0098	0.911598996360042\\
0.0099	0.909024769014653\\
0.01	0.907097494964775\\
};
\addlegendentry{$ \underline{\hat{E}}_0  $}

\addplot [color=mycolor2, dashed, line width=1.0pt]
  table[row sep=crcr]{%
0.0001	1.30828563211542\\
0.0002	1.28573082473555\\
0.0003	1.26842391232431\\
0.0004	1.25383351024222\\
0.0005	1.24097910796105\\
0.0006	1.22935783998744\\
0.0007	1.21867098115236\\
0.0008	1.20872389548248\\
0.0009	1.19938138836903\\
0.001	1.19054502544024\\
0.0011	1.18214049669652\\
0.0012	1.17411007066\\
0.0013	1.166407836517\\
0.0014	1.15899656995626\\
0.0015	1.15184559280807\\
0.0016	1.14492926649583\\
0.0017	1.13822590396643\\
0.0018	1.13171696622942\\
0.0019	1.125386457482\\
0.002	1.11922046193348\\
0.0021	1.11320678375374\\
0.0022	1.10733466339875\\
0.0023	1.10159455139788\\
0.0024	1.09597792598627\\
0.0025	1.09047714462264\\
0.0026	1.08508532200066\\
0.0027	1.0797962289957\\
0.0028	1.07460420831611\\
0.0029	1.06950410360329\\
0.003	1.06449119944892\\
0.0031	1.05956117034256\\
0.0032	1.05471003697635\\
0.0033	1.04993412865112\\
0.0034	1.04523005077376\\
0.0035	1.04059465662773\\
0.0036	1.03602502274944\\
0.0037	1.03151842736311\\
0.0038	1.02707233142212\\
0.0039	1.02268436188197\\
0.004	1.01835229689187\\
0.0041	1.01407405264271\\
0.0042	1.00984767165033\\
0.0043	1.00567131228734\\
0.0044	1.00154323940443\\
0.0045	0.997461815905913\\
0.0046	0.993425495163253\\
0.0047	0.989432814167047\\
0.0048	0.985482387331389\\
0.0049	0.98157290087625\\
0.005	0.977703107723289\\
0.0051	0.9738718228489\\
0.0052	0.970077919045377\\
0.0053	0.96632032304726\\
0.0054	0.962598011985103\\
0.0055	0.958910010133477\\
0.0056	0.955255385923905\\
0.0057	0.951633249196808\\
0.0058	0.948042748669512\\
0.0059	0.9444830695999\\
0.006	0.940953431627533\\
0.0061	0.937453086776064\\
0.0062	0.933981317602427\\
0.0063	0.930537435479851\\
0.0064	0.927120779003055\\
0.0065	0.923730712505151\\
0.0066	0.920366624676861\\
0.0067	0.917027927279522\\
0.0068	0.913714053944237\\
0.0069	0.91042445905021\\
0.007	0.907158616675974\\
0.0071	0.903916019617821\\
0.0072	0.900696178470225\\
0.0073	0.89749862076355\\
0.0074	0.894322890154745\\
0.0075	0.891168545667083\\
0.0076	0.888035160975379\\
0.0077	0.884922323733393\\
0.0078	0.881829634940417\\
0.0079	0.878756708344295\\
0.008	0.875703169878346\\
0.0081	0.87266865712986\\
0.0082	0.869652818838037\\
0.0083	0.866655314419401\\
0.0084	0.863675813518862\\
0.0085	0.860713995584763\\
0.0086	0.857769549466368\\
0.0087	0.854842173032337\\
0.0088	0.851931572808895\\
0.0089	0.849037463636447\\
0.009	0.846159568343503\\
0.0091	0.843297617436877\\
0.0092	0.840451348807146\\
0.0093	0.837620507448497\\
0.0094	0.834804845192082\\
0.0095	0.832004120452109\\
0.0096	0.829218097983932\\
0.0097	0.826446548653446\\
0.0098	0.82368924921716\\
0.0099	0.820945982112339\\
0.01	0.818216535256665\\
};
\addlegendentry{$ \underline{\tilde{E}}_0 $}

\end{axis}

\end{tikzpicture}%
	\caption{Worst-case achievable type-\RNum{1} error exponent with a relative entropy ball of radius $r$. The solid line corresponds to the optimization problems in \eqref{eq:worst_case}, and the dashed line corresponds to the approximated exponent using Theorem \ref{thm:lowerworstseq}. } 
	\label{fig:worstRsenseq}
\end{figure}
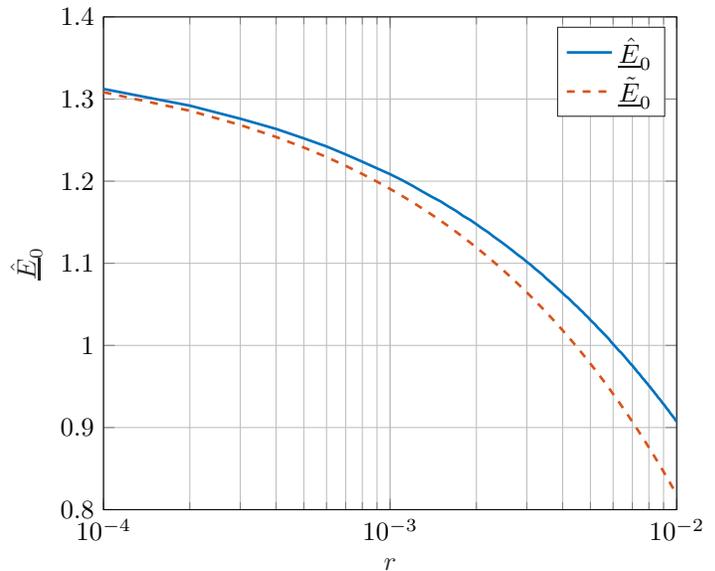

\section{Adversarial Setting}
\label{sec:adv}

In this section, we study the sensitivity of hypothesis testing under a perturbation of the observed samples by an adversary.

\subsection{Likelihood Ratio test}

We consider the worst-case scenario where an adversary can change the sample type to $\Th'$, where the change is assumed to be limited to a divergence ball around the type of actual sequence $\Th$ generated by either of the hypothesis, i.e., $d(\Th,\Th')\leq r$.

Similar to the case with the distribution mismatch, by direct application of Sanov's theorem, we can find the worst-case exponents by solving the optimization
\begin{align}
\underline{\hat{E}}_0(r)=  \min_{\substack{ \hat{Q} \in \mathcal{Q}^{\rm adv}_0  \\ Q:d(Q,\hat{Q})\leq r }} D(Q\|P_0),\label{eq:testworst1}\\
\underline{\hat{E}}_1(r)=  \min_{\substack{ \hat{Q} \in \mathcal{Q}^{\rm adv}_1 \\ Q:d(Q,\hat{Q})\leq r }} D(Q\|P_1),\label{eq:testworst2}
\end{align}
where	
\begin{align}
\mathcal{Q}^{\rm adv}_0&=  \big\{\hat{Q}\in \mathcal{P}(\mathcal{X}): D(\hat{Q}\| {P}_0)-D(\hat{Q}\|{P}_{1}) \geq {\gamma} \big \}, \label{eq:Q1set}\\
\mathcal{Q}^{\rm adv}_1&=  \big\{\hat{Q}\in \mathcal{P}(\mathcal{X}): D(\hat{Q}\| {P}_0)-D(\hat{Q}\|{P}_{1}) \leq {\gamma} \big \}. \label{eq:Q2set}
\end{align}
Furthermore, in the case where the distance $d$ is the $f$-divergence ball or the R\'enyi divergence ball of order $\alpha \in[0,1]$, $d(Q,\hat{Q})$ is  jointly convex \cite{Harremos,Csiszar}. Hence, \eqref{eq:testworst1} is a convex optimization problem and, the KKT conditions are also sufficient. Writing the Lagrangian for $\alpha=1$ we have
\begin{align}\label{eq:lagrange}
L(Q,\hat{Q},&\lambda_1,\lambda_2, \nu_1, \nu_2)= D(Q\|P_0) + \lambda_1 \big( D(\hat{Q}\|{P}_1) \nonumber \\
&~~~~~-D(\hat{Q}\|P_0) +\gamma \big) +  \lambda_2 \big ( D(Q\|\hat{Q}) -r\big )  \nonumber \\
&+ \nu_1 \Big ( \sum_{x\in \mathcal{X}} Q(x)-1\Big )+ \nu_2 \Big ( \sum_{x\in \mathcal{X}} \hat{Q}(x)-1\Big ).
\end{align}
Differentiating with respect to $Q(x)$ and  $\hat{Q}(x)$ and setting the derivatives to zero, we have
\begin{align}
1+\log \frac{Q(x)}{P_0(x)} + \lambda_2 \Bigg (1+\log \frac{Q(x)}{\hat{Q}(x)}\Bigg) + \nu_1&=0,\label{eq:lagrangetest1}\\
\lambda_1 \log \frac{{P}_0(x)}{{P}_1(x)} - \lambda_2 \frac{Q(x)}{\hat{Q}(x)}  + \nu_2&=0,\label{eq:lagrangetest2}
\end{align} 
respectively. Solving equations \eqref{eq:lagrangetest1}, \eqref{eq:lagrangetest2} for every $x\in\Xc$, we obtain
\begin{align}\label{eq:lowerworstKKT1}
Q_{\lambda_1,\lambda_2}(x)=& \frac{ P_0(x)} {\Big( 1-\frac{\lambda_1}{\lambda_2}\gamma+\frac{\lambda_1}{\lambda_2}\log \frac{P_0(x)}{P_1(x)}  \Big ) ^ {\lambda_2}} \times   \nonumber \\
&\Bigg( \sum_{a\in \Xc} \frac{ P_0(a)} {\Big( 1-\frac{\lambda_1}{\lambda_2}\gamma+\frac{\lambda_1}{\lambda_2}\log \frac{P_0(a)}{P_1(a)}  \Big ) ^ {\lambda_2}}   \Bigg)^{-1},\\
\hat{Q}_{\lambda_1,\lambda_2}(x)=&\frac{ P_0(x)} {\Big( 1-\frac{\lambda_1}{\lambda_2}\gamma+\frac{\lambda_1}{\lambda_2}\log \frac{P_0(x)}{P_1(x)}  \Big ) ^ {1+\lambda_2}}\times    \nonumber \\
&\Bigg( \sum_{a\in \Xc} \frac{ P_0(a)} {\Big( 1-\frac{\lambda_1}{\lambda_2}\gamma+\frac{\lambda_1}{\lambda_2}\log \frac{P_0(a)}{P_1(a)}  \Big ) ^ {1+\lambda_2}}   \Bigg)^{-1} ,
\end{align}
where $\lambda_1, \lambda_2$ can be find by solving  $D(Q_{\lambda_1,\lambda_2}\|\hat{Q}_{\lambda_1,\lambda_2})=r$ and $D(\hat{Q}_{\lambda_1,\lambda_2}\|{P}_1) -D(\hat{Q}_{\lambda_1,\lambda_2}\|P_0) = \gamma$. Next, we will study the error exponents' worst-case sensitivity when the radius of the divergence ball  is small.
\begin{theorem}\label{thm:adverLRT}
	For every 	$r\geq0$
	we have	
	\begin{equation}
	\underline{\hat{E}}_i(r)= E_i - \sqrt{r \cdot \theta_i(P_0,P_1,\gamma) }+ o \big(\sqrt{r} \big),  
	\end{equation}
	where	
	\begin{equation}\label{eq:sampsen}
	\theta_i(P_0,P_1,\gamma)=\frac{2}{\alpha}  {\rm Var}_{Q_{\lambda}} \bigg(\log \frac{Q_{\lambda}(X)}{{P}_i(X)}  \bigg), 
	\end{equation}
	and ${Q}_{\lambda}(X)$ is the minimizing distribution in \eqref{eq:tilted}.
\end{theorem}

\begin{remark}
	
	Unlike the distribution mismatch where $\frac{\partial \theta_0^{\rm dist}(\hat{P}_0,\hat{P}_1,\hat{\gamma})}{\partial \gamma} \geq 0 , \frac{\partial \theta_1^{\rm dist}(\hat{P}_0,\hat{P}_1,\hat{\gamma})}{\partial \gamma} \leq 0 $, the sensitivities of the likelihood ratio test towards sample mismatch are not strictly-increasing nor strictly-decreasing. Instead, it can be shown that the derivative of the sample sensitivities respect to the threshold $\hat{\gamma}$  are proportional to the skewness of the random variable $\log \frac{P_1(X)}{P_0(X)}$ under the distribution $Q_\lambda(X)$ which can have any sign, depending on $P_0$ and $P_1$. However, similarly to the case of distribution mismatch, the sensitivities under both distributions are equal when $\lambda=\frac{1}{2}$, i.e., $\theta_0(\hat{P}_0,\hat{P}_1,\hat{\gamma})=\theta_1(\hat{P}_0,\hat{P}_1,\hat{\gamma})$. More generally, for any $\lambda$, by substituting $Q_\lambda$, we have
	
\begin{align}
\frac{\theta_0(\hat{P}_0,\hat{P}_1,\hat{\gamma})}{\theta_1(\hat{P}_0,\hat{P}_1,\hat{\gamma})} &=\frac{{\rm Var}_{Q_{\lambda}} \bigg(\log \frac{Q_{\lambda}(X)}{{P}_0(X)}  \bigg)}{{\rm Var}_{Q_{\lambda}} \bigg(\log \frac{Q_{\lambda}(X)}{{P}_1(X)}  \bigg)}\\
&=\frac{{\rm Var}_{Q_{\lambda}} \bigg( \kappa+\log \frac{ { P_{0}^{1-\lambda}(x) P_{1}^{\lambda}(x) }}{{P}_0(X)}  \bigg)}{{ \rm Var}_{Q_{\lambda}} \bigg(\kappa+\log \frac{{ P_{0}^{1-\lambda}(x) P_{1}^{\lambda}(x) }}{{P}_1(X)}  \bigg)}\\
&=\frac{{\rm Var}_{Q_{\lambda}} \bigg( -\lambda \log \frac{ {  P_{0}(x) }}{{P}_1(X)}  \bigg)}{{\rm Var}_{Q_{\lambda}} \bigg( (1-\lambda) \log \frac{{ P_{0}(x)  }}{{P}_1(X)}  \bigg)}\\
&=\Big(\frac{\lambda}{1-\lambda} \Big)  ^2.
\end{align}
where $\kappa =-\log{\sum_{a \in \Xc }  P_{0}^{1-\lambda}(a) P_{1}^{\lambda}(a) }.$
\end{remark}

 Next, we will compare the sensitivity of the likelihood ratio test toward distribution mismatch and sample mismatch. 
\begin{corollary}\label{cor:samp_dist}
	For every $P_0, P_1 \in \Pc(\Xc)$ and $i\in\{0,1\}$, we have
	\begin{equation}\label{eq:lowersenadv}
 	\Bigg(\min_{x \in \Xc} \frac{P_i(x)}{Q_\lambda(x)} \Bigg )  \theta_i^{\rm dist} - E_i^2\leq \theta_i^{\rm adv}\leq\theta_i^{\rm dist},
	\end{equation}
	where $\theta_i^{\rm dist}, \theta_i^{\rm adv}$ are the likelihood ratio test distribution and adversarial setting sensitivities in \eqref{eq:sensitivity} and \eqref{eq:sampsen}.

\end{corollary}	

As can be seen from the above result, having distribution mismatch renders the test performance more sensitive than an adversary tampering with the observation when the divergence ball radii are equal.

\begin{example}

	Consider $\hat{P}_0 =\text{Bern}(0.1)$ ,    $\hat{P}_1 =\text{Bern}(0.8)$. By finding the optimizing distribution $Q_\lambda$ we have $\theta_0^{\rm dist}=1.136$, $\theta_0^{\rm adv} = 0.854$ and the lower bound in \eqref{eq:lowersenadv} equals to $0.151$.
	
\end{example}

Observe that for small radii, the approximations for worst case mismatch and worst case adversarial settings are accurate. In this case, by equating the approximations of mismatched and adversarial error exponents, we have that if $\frac{r^{\rm adv} }{r_i^{\rm dist} } =\frac{\theta_i^{\rm dist} }{\theta_i^{\rm adv} }$ then, the worst case error exponents under both scenarios are equal.

\subsection{Generalized Likelihood Ratio Test}

Next, we consider the worst-case Hoeffding's generalized likelihood test in the adversarial setting. Similarly to the likelihood ratio test, we assume the observer receives samples with the type $\Th'$ which satisfies $d(\Th,\Th') \leq r$ and $\Th$ is the type of the original sequence generated by the unknown hypothesis. By direct application of Sanov's theorem, we have
\begin{align}
\underline{\hat{E}}_0(r)=  \min_{\substack{ \hat{Q}: D(\hat{Q}\|P_0)\geq  \gamma \\ Q:d(Q,\hat{Q})\leq r}} D(Q\|P_0),\label{eq:testworstH1}\\
\underline{\hat{E}}_1(r)=  \min_{\substack{ \hat{Q}: D(\hat{Q}\|P_0)\leq  \gamma \\ Q:d(Q,\hat{Q})\leq r }} D(Q\|P_1),\label{eq:testworstH2}
\end{align}
In this scenario, unlike the case with distribution mismatch, the adversary can change both error exponents. It is clear that $\underline{E}_0$ is non-convex optimization, and for the $f$-divergence or the R\'enyi divergence of order $\alpha \in [0,1]$, $\underline{E}_1$ is a convex optimization and hence easy to solve. Like previous sections, we will look into the sensitivity of error exponents when the divergence ball radius is small.

\begin{theorem}\label{thm:advGLRT}
		For every 	$r\geq0$, the worst-case error exponents can be approximated as 
	\begin{equation}\label{eq:worstapproxHoef}
	\underline{\hat{E}}_i(r)= E_i -  \sqrt{ r \cdot \theta_i(P_0,P_1, \gamma) }+o(\sqrt{r_0}),  
	\end{equation}
	where  
	\begin{align}\label{eq:hoefsensamp}
	\theta_0(P_0, P_1, \gamma) &= \frac{2}{\alpha}  \max_{\substack{ \hat{Q}: D(\hat{Q}\|P_0) = \gamma}} {\rm Var}_{\hat{Q}} \bigg(\log \frac{\hat{Q}(X)}{P_0(X)}  \bigg),\\
	\theta_1(P_0,P_1,\gamma)&=\frac{2}{\alpha}  {\rm Var}_{{Q}_{\mu}} \bigg(\log \frac{{Q}_{\mu}(X)}{P_1(X)}  \bigg), \label{eq:Hoefadver1}
	\end{align}	
	are the type-\RNum{1} and type-\RNum{2}  error exponents' sensitivities of the Hoeffding's test, and ${Q}_{\mu}(X)$ is the minimizing distribution in \eqref{eq:tiltedH}. 
\end{theorem}
In this scenario, both exponents will be affected by a change in the observation type as opposed to distribution mismatch, where only the first exponent would change as a result of the mismatch. Similarly to the likelihood ratio test, we have that distribution mismatch is more sensitive than adversarial observation perturbation.
\begin{corollary}\label{cor:sen_hoef_comapre}
	For every $P_0, P_1 \in \Pc(\Xc)$ we have
	\begin{equation}\label{eq:lowersen}
 \theta_0^{\rm adv}\leq\theta_0^{\rm dist},
\end{equation}
	where $\theta_0^{\rm dist}, \theta_0^{\rm adv}$ are the generalized likelihood ratio test distribution and sample sensitivities in \eqref{eq:hoefsen} and \eqref{eq:hoefsensamp}.	
\end{corollary}	
It is easy to see that Corollary \ref{cor:sen_hoef_comapre} is an immediate result of Corollary \ref{cor:samp_dist}  by taking the maximum of the both distribution and sample sensitivities over all distributions $\hat{Q}$ such that $D(\hat{Q}\|P_0) = \gamma$.  

\subsection{Sequential Probability Ratio Test}

Finally, we study the impact of adversarial observation perturbation  on the error exponents of the sequential probability ratio test due to sample mismatch. In this setting, we assume that the adversary can perturb the received sequence type in a divergence ball with a specific radius when the tests stop. 
Formally, the adversary can observe all the future samples $(x_i) _{i=1}^{\infty}$ and alter the samples to $(x'_i) _{i=1}^{\infty}$, with the constraint that for any positive $n$ the sample types and perturbed sample types of  $(x_i) _{i=1}^{n}$ ,  $(x'_i) _{i=1}^{\infty}$ satisfy $d(\Th,\Th')\leq r$, and the sequential probability ratio test calculates the log-likelihood ratio using the perturbed samples, i.e.,
\begin{align}
\underline{\hat{S}}_{n}=\sum_{i=1}^n \log \frac{{P}_0(x_i')}{{P}_1(x_i')}.
\end{align}
The worst-case error exponent trade-off in the adversarial setting can be defined as
\begin{align}
 \underline{\hat{E}}_1(\underline{\hat{E}}_0,r) \triangleq &\inf_{d(\Th,\Th')\leq r} \sup \Big \{ \underline{\hat{E}}_1 \in \mathbb{R}_{+}: \exists \gamma_0, \gamma_1,  \nonumber \\
& \exists n_0, n_1  \in \ZZ_{+},  \text{s.t.}    \mathbb{E}_{P_0} [\underline{\hta}]   \leq n_0, \mathbb{E}_{P_1}[\underline{\hta}]  \leq n_1, \nonumber \\
 &~~~~~~~\epsilon_0({\Phi}) \leq 2^{- n_0  \underline{\hat{E}}_0}  ~ \text{and} ~  \epsilon_1({\Phi}) \leq 2^{- n_1  \underline{\hat{E}}_1} \Big \}.
\end{align}
where the worst case stopping time is defined as 
\begin{align}
\underline{\hta}=\inf \{n\geq1: {\underline{\hat{S}}_n} \geq{ \gamma}_0  ,  {\underline{\hat{S}}_n} \leq{ \gamma}_1  \}.
\end{align}  
 This model is the worst possible adversarial setting as the adversarial can see future samples and change the whole sequence until the stopping time to maximize the stopping time and also minimize the error exponents. We find a lower bound on the error exponents and an upper bound on the average stopping time in this adversarial setting. 

\begin{theorem}\label{thm:SPRTadver}
	For every $r\geq0$, $i\in\{0,1\}$, define $\bar\imath=1-i$ to be the complement of index $i$. The worst-case error exponents can be approximated as 
	\begin{align}
	\underline{\hat{E}}_0 (r)	\underline{\hat{E}}_1 (r)  \geq& \Big(D(P_1\|P_0)-2\sqrt{r\cdot \theta_0(P_0,P_1)  }\Big) \times  \nonumber \\
	&\Big(D(P_0\|P_1)-2\sqrt{r\cdot \theta_1(P_0,P_1)  }\Big)   +o(\sqrt{r}).
	\end{align}
	where 
	\begin{equation}
	\theta_i(P_0,P_1)=\frac{2}{\alpha} \text{Var}_{P_{\bar\imath}} \Bigg(\log \frac{P_{\bar\imath}}{{P}_i}  \Bigg) 
	\end{equation}

\end{theorem}

Theorem \ref{thm:SPRTadver} provides only a lower bound to the worst-case error exponents in the adversarial setting. Thus, we cannot compare the sensitivities derived in the mismatched case with the adversarial worst-case sensitivity.

\appendices

\section{Proof of Theorem \ref{thm:mismatchLRT}}\label{apx:TmismatchLRT}

We show the result  for $	\hat{E}_0$ and similar steps are valid for $\hat{E}_1$. The type-\RNum{1} probability of error can be written as 
\begin{align}	
\hat{\epsilon}_0  = \sum_{\substack{\bx\in\Xc^n\\D(\Th\|\hat{P}_0)-D(\Th\|\hat{P}_1)  \geq \hat{\gamma}}} P_{0}^n(\xv).
\label{eq:tailmismatch}
\end{align}
Applying Sanov's Theorem  to \eqref{eq:tailmismatch} to get  \eqref{eq:LRTmis1} is immediate. The optimization problem in  \eqref{eq:LRTmis1}  consists of the minimization of a convex function over linear constraints. Therefore, the KKT conditions are also sufficient \cite{Boyd}. Writing the Lagrangian, we have
\begin{align}\label{eq:lagrangeMM}
L(Q,\lambda,\nu)= &D(Q\|P_0) + \lambda \big ( D(Q\|\hat{P}_1)-D(Q\|\hat{P}_0) +\hat{\gamma} \big )  \nonumber \\
&+\nu \Big ( \sum_{x\in \Xc} Q(x)-1 \Big).
\end{align}
Differentiating with respect to $Q(x)$ and setting to zero we have
\begin{equation}\label{eq:lagrangeder}
1+\log \frac{Q(x)}{P_0(x)} +\lambda \log \frac{\hat{P}_0(x)}{\hat{P}_1(x)} + \nu=0.
\end{equation} 
Solving equations \eqref{eq:lagrangeder} for every $x\in\Xc$ we obtain \eqref{eq:tiltedMM1}. Moreover, from the complementary slackness condition if \cite{Boyd}
\begin{equation}\label{eq:threshcondition}
D(P_0\|\hat{P}_0)- D(P_0\|\hat{P}_1) \leq \hat{\gamma},
\end{equation}
then \eqref{eq:KKTgamma1} should hold. Otherwise, if \eqref{eq:threshcondition} does not hold then $\lambda$ in \eqref{eq:lagrangeder} should be zero and hence $\hat{Q}_{\lambda_0}=P_0$, $\hat{E}_0=0$.  Finally, substituting the minimizing distribution $\hat{Q}_{\lambda_0}$ \eqref{eq:tiltedMM1} into \eqref{eq:lagrangeMM} we get the dual expression
\begin{equation}\label{eq:lagrange}
g(\lambda)= \lambda \hat{\gamma} - \log \Big (  \sum_{x\in \Xc} P_0 \hat{P}_0^{-\lambda}(x) \hat{P}_1^{\lambda}(x) \Big ).  
\end{equation}
Since the optimization problem in \eqref{eq:LRTmis1} is convex, then the duality gap is zero \cite{Boyd}, and this proves the \eqref{eq:dual}. 
%

\section{Proof of Theorem \ref{thm:stein}}\label{apx:Tstein}

For convenience, in this section, we make explicit the dependence of  $\hat{\Qc}_1$ and $\hat{E}_1$ on the threshold of the test $\hat \gamma$, and denote them by $\hat{\Qc}_1({\hat{\gamma}})$ and $\hat{E}_1({\hat{\gamma}})$.

First, notice that  $\hat{E}_1$ is a non-increasing function of $\hat{\gamma}$  since for every $\hat{\gamma}_1 \leq \hat{\gamma}_2 $  we have
\begin{equation}
\hat{\Qc}_1({\hat{\gamma}_1}) \subset   \hat{\Qc}_1({\hat{\gamma}_2}),
\end{equation}
hence 
\begin{equation}
\hat{E}_1({\hat{\gamma}_2}) \leq \hat{E}_1({\hat{\gamma}_1}).
\end{equation}
Therefore, in the Stein's regime we are looking for the smallest threshold such that $\limsup_{n\rightarrow \infty}  \hat{\epsilon}_0  \leq \epsilon$. Let
\begin{equation}\label{eq:steinthresh}
\hat{\gamma}= D(P_0\|\hat{P}_0) - D(P_0\|\hat{P}_1) -   \sqrt {\frac{ V(P_0, \hat{P}_0,\hat{P}_1) } {n} }    \Qsf^{-1}(\epsilon),
\end{equation}	
where 
\begin{align}
&V(P_0, \hat{P}_0,\hat{P}_1)= {\rm Var}_{P_0}\bigg[   \log \frac{\hat{P}_0}{\hat{P}_1} \bigg ] \nonumber\\
 &= \sum_{x\in \Xc} P_0(x)  \bigg ( \log \frac{\hat{P}_0}{\hat{P}_1} \bigg )^2 - 
\big ( D(P_0\|\hat{P}_1) - D(P_0\|\hat{P}_0)   \big)^2,
\end{align}
and $\Qsf^{-1}(\epsilon)$ is the inverse cumulative distribution function of a zero-mean unit-variance Gaussian random variable. For such $\hat{\gamma}$, the type-\RNum{1}  error probability of the mismatched likelihood ratio test is
\begin{align}	
\hat{\epsilon}_0=&\PP_0 \Bigg [\frac{1}{n} \sum_{i=1}^{n} \log \frac{\hat{P}_0(X_i)}{\hat{P}_1(X_i)}  \leq  
D(P_0\|\hat{P}_1) - D(P_0\|\hat{P}_0)    \nonumber \\
&+\sqrt {\frac{ V(P_0, \hat{P}_0,\hat{P}_1) } {n} }    \Qsf^{-1}(\epsilon) \Bigg ].
\end{align}
Observe that $D(P_0\|\hat{P}_1) - D(P_0\|\hat{P}_0) =  \mathbb{E}_{P_0} \Big [ \log \frac{\hat{P}_0(X)}{\hat{P}_1(X)} \Big ]$. Let $  \hat{S}_n=\frac{1}{n}\sum_{i=1}^n \hat\imath(x_i)$, where $\hat\imath(x_i)=\log \frac{\hat{P}_0(x_i)}{\hat{P}_1(x_i)}$.  Letting $Z$ be a zero-mean unit-variance Gaussian random variable, then, by the central limit theorem we have
\begin{align}
&\limsup_{n\rightarrow \infty}  \hat{\epsilon}_0   &\notag\\
&= \limsup_{n\rightarrow \infty} \PP_0\Bigg [ \frac{ \sqrt{n} \big ( \hat{S}_n- \mathbb{E}_{P_0} [\hat\imath(X)] \big )}{\sqrt {V(P_0, \hat{P}_0,\hat{P}_1) }  }   \leq    \Qsf^{-1}(\epsilon)     \Bigg]\\
&=\Pp\big [Z \leq \Qsf^{-1}(\epsilon) \big]\\
&= \epsilon. 
\end{align} 
Therefore, asymptotically, the type-\RNum{1} error probability of mismatched likelihood ratio test with $\hat{\gamma}$ in \eqref{eq:steinthresh} is equal to $\epsilon$. 

Next, we need to show that for any threshold $\hat{\gamma}$ and $\varepsilon>0$ such that
\begin{equation}\label{eq:limsupthresh}
\limsup_{n \rightarrow \infty} \hat{\gamma} +\varepsilon\leq D(P_0\|\hat{P}_0) - D(P_0\|\hat{P}_1),
\end{equation}
the type-\RNum{1} probability of error tends to $1$ as the number of observation approaches infinity, which implies that $D(P_0\|\hat{P}_0) - D(P_0\|\hat{P}_1)$ is the lowest possible threshold that meets the constraint $\limsup_{n\rightarrow \infty}  \hat{\epsilon}_0  \leq \epsilon$. Hence, the corresponding $\hat E_1({\hat\gamma})$ is this highest type-\RNum{2} exponent that meets the constraint.
In order to show this, define the following sets 
\begin{align}
\Ec_{\delta} = \Big \{ \bx\in\Xc^n:\, & \| \Th(x) -P_0(x)\|_\infty   < \delta \Big     \},\\   
\Dc = \big\{  \bx\in\Xc^n: \,  & \big|D(\Th\|\hat{P}_0)-D(\Th\|\hat{P}_1)  \nonumber \\
&-D(P_0\|\hat{P}_0)+D(P_0\|\hat{P}_1) \big | <  \varepsilon \big\}, \\
\bar \Dc = \big\{  \bx\in\Xc: \, &D(\Th\|\hat{P}_0)-D(\Th\|\hat{P}_1)  \nonumber \\
 &- D(P_0\|\hat{P}_0)+D(P_0\|\hat{P}_1) \geq - \varepsilon \big\}.
\end{align}
where $\|.\|_{\infty}$ is the norm infinity. From the continuity of $D(.\|\hat{P})$ we have that for any $\varepsilon >0$ such that
\begin{equation}\label{eq:epsilondelta}
\big |D(\Th\|\hat{P}_0)-D(\Th\|\hat{P}_1) - D(P_0\|\hat{P}_0)+D(P_0\|\hat{P}_1) \big  | < \varepsilon. 
\end{equation} 
there exists $\delta>0$ such that for all $\Th$ satisfying
\begin{equation}
\| \Th(x) -P_0(x)\|_\infty   < \delta, 
\end{equation}
\eqref{eq:epsilondelta} holds. Therefore, when  \eqref{eq:limsupthresh} holds 
\begin{align}
\liminf_{n\rightarrow \infty} \epsilon_0 (\hat{\phi} ) \geq&  \liminf_{n\rightarrow \infty}  \sum_{x\in\bar \Dc} P_0^n(\xv)\\
\geq  &\liminf_{n\rightarrow \infty}   \sum_{x\in \Dc} P_0^n(\xv).   
\end{align}
Now from the continuity argument, there exists a $\delta$ such that 
\begin{equation}
\sum_{x\in\Dc} P_0^n(\xv) \geq   \sum_{x\in\Ec_{\delta}} P_0^n(\xv).
\end{equation}
Set $\delta_n=\sqrt{\frac{\log n}{n}}$. Thus, for sufficiently large $n$, $\delta_n \leq \delta$, Therefore, we have 
\begin{align}
\liminf_{n\rightarrow \infty} \epsilon_0 (\hat{\phi} )  &\geq \liminf_{n \rightarrow \infty} \sum_{\xv\in\Ec_{\delta_n}} P_0^n(\xv)\\
%
&\geq \lim_{n \rightarrow \infty} 1- \frac{2|\Xc|}{n}\\
&=1.
\end{align}
where the last step is by Hoeffding's inequality \cite{Hoeffdingineq} and union bound. Therefore, for any $\hat{\gamma} < D(P_0\|\hat{P}_0) - D(P_0\|\hat{P}_1)$ type-\RNum{1} error goes to unity which concludes the theorem.


\section{Proof of Theorem \ref{thm:lowerworst}}\label{apx:Tlowerworst}

We will use the following two lemmas which are the local approximation of the R\'enyi Divergences and $f$-divergences.
\begin{lemma}\cite{Harremos} 
Let $P$ and $Q$ be two probability distributions over the same alphabet $\Xc$, the R\'enyi divergence of order $\alpha$ can be locally approximated by
\begin{equation}
D_\alpha(P\|Q)=\frac{\alpha}{2} \sum_{x\in\Xc} \frac{\big(P(x)-Q(x)\big )^2}{P(x)}+ o(\|P(x)-Q(x))\|^2).
\end{equation}
\end{lemma}

\begin{lemma}\cite{Csiszar}
	Let $P$ and $Q$ be two probability distributions over the same alphabet $\Xc$, and let the convex function $f(t)$ to be twice differentiable at $t = 1$. Then the $f$-divergence using such function can be locally approximated by
	\begin{align}
	D_f(P\|Q)&=\frac{f''(1)}{2} \sum_{x\in\Xc} \frac{\big(P(x)-Q(x)\big )^2}{P(x)}\notag\\&~~~~~+ o(\|P(x)-Q(x))\|^2)
	\end{align}
\end{lemma}

We show the result under the first hypothesis, and similar steps are valid for the second hypothesis. Also, since the second-order approximation of the  R\'enyi divergence and the family of twice differentiable $f$-divergences are only different in a constant, by setting $\alpha$ to be the order of the R\'enyi divergence and $\alpha$ to be $f''(1)$ we can prove the result for both divergences simultaneously.  Consider the first minimization in \eqref{eq:MMlower1} over $Q$, i.e.,
\begin{equation}\label{eq:LRTmis1perturb}
\hat{E}_0= \min_{ Q \in \hat{\Qc}_0  } D(Q\|P_0).
\end{equation}
Observe that by assumption, $\hat P_0(x)>0$ for each $x\in\Xc$. Therefore, for every $\alpha$ there exists a positive $\bar{r}_0$ such that $P_0(x)>0$ for every  $P_0 \in \Bc(\hat P_0,\bar{r}_0)$ (for example in the case of the R\'enyi divergence with $\alpha\geq1$, $P_0(x)>0$ for every finite $r_0$). Hence, for $P_0 \in \Bc(\hat P_0,\bar{r}_0)$, the relative entropy $D(Q\|P_0)$ is continuously differentiable in both $Q, P_0$ for some positive $\bar{r}_0$.  Moreover, the constraints in  \eqref{eq:LRTmis1perturb} are continuously differentiable with respect to $Q$ and also trivially with respect to $P_0$, since the constraints do not depend on $P_0$. Hence, the optimization in \eqref{eq:LRTmis1perturb} is minimizing a continuously differentiable function  over a compact set with continuously differentiable constraints.  Hence, by the maximum theorem \cite{Walker},
$\hat{E}_0$ is  a continuous function of $P_0$  for all $P_0 \in \Bc(\hat P_0,\bar{r}_0)$ with finite radius $\bar{r}_0$. Also, by the envelope theorem\cite{Segal} we have
\begin{equation}
\frac{ \partial \hat{E}_0 }{\partial P_0(x)}= -\frac{\hat{Q}_{\lambda_0}(x)}{P_0(x)}.  
\end{equation} 
Define the vectors
\begin{align}
\nabla \hat{E}_0 &= \bigg( -\frac{\hat{Q}_{\lambda}(x_1)}{\hat{P}_0(x_1)},\dotsc, -\frac{\hat{Q}_{\lambda}(x_{|\Xc|})}{\hat{P}_0(x_{|\Xc|})}\bigg)^T,\\
\thetav_{P_0}		 &= \big(P_0(x_1)-\hat{P}_0(x_1),\dotsc,P_0(x_{|\Xc|})-\hat{P}_0(x_{|\Xc|})\big)^T.
\end{align}			
Applying the Taylor expansion to $\hat{E}_0$ around $P_0=\hat{P}_0$, we obtain
\begin{align}\label{eq:linearapprox}
\hat{E}_0=E_0  +     \thetav_{P_0}^{T}   \nabla \hat{E}_0     + o(\| \thetav_{P_0} \|_{\infty}).
\end{align}
By substituting the expansion \eqref{eq:linearapprox}   for  the first minimization in  \eqref{eq:MMlower1} we obtain
\begin{equation}\label{eq:approxworstlrt}
\underline{\hat{E}}_0 (r_0)  = \min_{P_0 \in \Bc(\hat P_0,r_0)   }  E_0  +     \thetav_{P_0}^{T}   \nabla\hat{E}_0     + o(\|\thetav_{P_0} \|_{\infty}).
\end{equation}

Now, we further approximate the outer minimization constraint in \eqref{eq:MMlower1}. By approximating $d(\hat{P}_0, P_0 )$ we get \cite{Zheng}
\begin{equation}\label{eq:KLapprox}
d(\hat{P}_0 , P_0 ) = \frac{1}{2}  \thetav_{P_0}^T \Jm(\hat{P}_0) \thetav_{P_0} + o (\| \thetav_{P_0} \|^2_{\infty}),
\end{equation}
where 
\begin{equation}
\Jm(\hat{P}_0)=\diag\bigg( \frac{\alpha}{\hat{P}_0(x_1)},\dotsc,\frac{\alpha}{\hat{P}_0(x_{|\Xc|})}\bigg)
\end{equation}
is the Fisher information matrix. Therefore,   \eqref{eq:approxworstlrt} can be approximated as
\begin{align}\label{eq:worstapproxoptlrt}
&\underline{\hat{E}}_0 (r_0)  =  \nonumber \\
& \min_{\substack{ \frac{1}{2} \thetav_{P_0}^T \Jm(\hat{P}_0) \thetav_{P_0} +o (\| \thetav_{P_0} \|^2_{\infty}) \leq r_{0} \\ \onev^T\thetav_{P_0}=0}} \Big \{ E_0  +     \thetav_{P_0}^{T}   \nabla \hat{E}_0+o(\|\thetav_{P_0} \|_{\infty}) \Big \} \\
&~~~~~~~~~= \min_{\substack{ \frac{1}{2} \thetav_{P_0}^T \Jm(\hat{P}_0) \thetav_{P_0}  \leq r_{0} \\ \onev^T\thetav_{P_0}=0}} \Big \{ E_0  +     \thetav_{P_0}^{T}   \nabla \hat{E}_0 \Big \}+o(\sqrt{r_0}),  \label{eq:approxerrorlrt}
\end{align}	
	where to get \eqref{eq:approxerrorlrt} we have taken $o(\|\thetav_{P_0} \|_{\infty})$ out of the minimization and substitute it with $o(\|\thetav^*_{P_0}(r_0) \|_{\infty})$ where $\thetav^*_{P_0} $ is the optimizing solution to the minimization. Moreover, in approximating the inequality constraint, we incur an error of the order $o(\sqrt{r_0})$ in $\|\thetav^*_{P_0}\|_{\infty}$. Also, from the inequality constraint and the restriction it imposes on the length of the vector $\thetav_{P_0}  $ we have that $\|\thetav^*_{P_0}\|_{\infty} \leq c\sqrt{r_0}+o(\sqrt{r_0})$ where $c$ is  independent from $r_0$, from which we obtain \eqref{eq:approxerrorlrt}.

The optimization problem in \eqref{eq:approxerrorlrt} is  convex and hence the KKT conditions are sufficient.  The corresponding Lagrangian is given by
\begin{align}\label{eq:laglrt}
L(\thetav_{P_0}, \lambda,\nu) &= E_0  +     \thetav_{P_0}^{T}   \nabla \hat{E}_0   \nonumber \\
&+ \lambda \Big (\frac{1}{2}\thetav_{P_0}^T \Jm(\hat{P}_0) \thetav_{P_0}  - r_{0} \Big) +\nu ( \onev^T\thetav_{P_0} ).
\end{align}
Differentiating with respect to $\thetav_{P_0}$ and setting to zero, we have
\begin{equation}\label{eq:KKTsenlrt}
\nabla \hat{E}_0 + \lambda \Jm(\hat{P}_0)\thetav_{P_0} +\nu \onev=0.
\end{equation}
Therefore,
\begin{equation}\label{eq:deltaPsolutionlrt}
\thetav_{P_0} =-\frac{1}{\lambda}   \Jm^{-1}(\hat{P}_0) \big (\nabla \hat{E}_0   +\nu \onev \big ).
\end{equation}
Note that if $\lambda=0$ then from \eqref{eq:KKTsenlrt}    $  \nabla \hat{E}_0= -\nu \onev$ which cannot be true for thresholds satisfying \eqref{eq:threshcodsen} since $\hat{Q}_{\lambda} \neq \hat{P}_0$. Therefore, from the complementary slackness condition \cite{Boyd} the inequality constraint in \eqref{eq:approxerrorlrt} should be satisfied with equality.  By solving $\bold{1}^T\thetav_{P_0} =0 $ we get
\begin{align}
 \nu &= -  \frac{\onev^T\Jm^{-1}(\hat{P}_0) \nabla \hat{E}_0}{\onev^T\Jm^{-1}(\hat{P}_0) \onev^T  } \\
 &=\frac{\frac{1}{\alpha} \sum_{x\in\Xc} \hat{P}_0(x) \frac{\hat{Q}_{\lambda}(x)}{\hat{P}_0(x)}}{ \frac{1}{\alpha} \sum_{x\in\Xc} \hat{P}_0(x)   }\\
 &=1  
\end{align}
Also by letting $\frac{1}{2}\thetav_{P_0}^T \Jm(\hat{P}_0) \thetav_{P_0} = r_0$, we have
\begin{align}
2 r_0 {\lambda^2} &= \big (\nabla \hat{E}_0   + \onev \big )^T \Jm^{-1}(\hat{P}_0)  \Jm(\hat{P}_0)         \Jm^{-1}(\hat{P}_0) \big (\nabla \hat{E}_0   + \onev \big )\\
&=\Big (  \nabla \hat{E}_0^T \Jm^{-1}(\hat{P}_0) \nabla \hat{E}_0  \nonumber \\
&~~~~~~~~~~+ 2 \onev^T \Jm^{-1}(\hat{P}_0) \nabla \hat{E}_0  +\onev^T \Jm^{-1}(\hat{P}_0) \onev \Big )\\
&= \frac{1}{\alpha} \sum_{x\in\Xc} \hat{P}_0(x) \frac{\hat{Q}^2_{\lambda}(x)}{\hat{P}_0^2(x)} - \frac{1}{\alpha}\\
&= \frac{1}{\alpha}  {\rm Var}_{\hat{P}_0} \bigg(\frac{\hat{Q}_{\lambda}(X)}{\hat{P}_0(X)}  \bigg), 
\end{align}
and hence
\begin{equation}\label{eq:deltaPlrt}
\thetav_{P_0} =-  \sqrt{\frac{2r_0}{ \frac{1}{\alpha}  {\rm Var}_{\hat{P}_0} \bigg(\frac{\hat{Q}_{\lambda}(X)}{\hat{P}_0(X)}  \bigg)}} \Jm^{-1}(\hat{P}_0)\big  (\nabla \hat{E}_0 +\onev\big ).
\end{equation}
Substituting \eqref{eq:deltaPlrt} into  \eqref{eq:worstapproxoptlrt} yields
\begin{align}
\underline{\hat{E}}_0 (r_0) &= E_0  +     \thetav_{P_0}^{T}   \nabla \hat{E}_0 + o(\sqrt{r_0})\\
&=E_0 -  \sqrt{\frac{2r_0}{ \frac{1}{\alpha}  {\rm Var}_{\hat{P}_0} \bigg(\frac{\hat{Q}_{\lambda}(X)}{\hat{P}_0(X)}  \bigg)}} \cdot  \frac{1}{\alpha}  {\rm Var}_{\hat{P}_0} \bigg(\frac{\hat{Q}_{\lambda}(X)}{\hat{P}_0(X)}  \bigg)  \nonumber \\
&~~+ o(\sqrt{r_0})\\
&= E_0 -  \sqrt{ \frac{2r_0}{\alpha} {\rm Var}_{\hat{P}_0} \bigg(\frac{\hat{Q}_{\lambda}(X)}{\hat{P}_0(X)}  \bigg)  } + o(\sqrt{r_0})
\end{align}	
which concludes the proof.


\section{Proof of Proposition \ref{cor:varderivative}}\label{apx:Lvarderivative}

We show the result  under the first hypothesis and similar steps are valid under the second hypothesis. To prove the Theorem we need the following lemma. 
\begin{lemma}\label{lem:convex}
	Consider the following optimization problem
	\begin{equation}
	E(\gamma)= \min_{ \mathbb{E}_Q [X] \geq \gamma } D(Q\|P),
	\end{equation} 
	where $P,Q \in \Pc(\Xc)$, and $X$  takes values in $\Xc$. Then	$E(\gamma)$ is convex in $\gamma$.
\end{lemma}		

\begin{proof}
	Let 
	\begin{equation}
	Q^{*}_{1} = \argmin_{ \mathbb{E}_Q [X] \geq \gamma_1} D(Q\|P) ~~ Q^{*}_{2} = \argmin_{ \mathbb{E}_Q [X] \geq \gamma_2} D(Q\|P).
	\end{equation}
	From the  convexity of the relative entropy, 	for any   $\beta \in (0,1)$,
	\begin{align}
	&D\big(\beta Q^*_1 +(1-\beta) Q^*_2 \| P \big) \leq \beta  D( Q^*_1  \| P) +(1-\beta)  D( Q^*_2 \| P)\\
	&= \beta  \min_{ \mathbb{E}_Q [X] \geq \gamma_1 } D(Q\|P) +(1-\beta)  \min_{ \mathbb{E}_Q [X] \geq \gamma_2 } D(Q\|P).
	\end{align}			
	Furthermore,  since $Q^*_1, Q^*_2$  satisfy their corresponding optimization constraints, then $\mathbb{E}_{Q^*_1}[X] \geq   \gamma_1$, $\mathbb{E}_{Q^*_2}[X]  \geq  \gamma_2$ , hence
	\begin{equation}
	\mathbb{E}_{\beta Q^*_1 +(1-\beta) Q^*_2}[X]  \geq \beta \gamma_1+ (1-\beta) \gamma_2.
	\end{equation}
	Therefore,  $\beta Q^*_1 +(1-\beta) Q^*_2$ satisfies the optimization constraint when $\gamma= \beta \gamma_1 + (1-\beta) \gamma_2$, then 
	\begin{align}
	&\min_{	\mathbb{E}_{Q} [X] \leq  \beta \gamma_1+(1-\beta) \gamma_2} D(Q\|P) \leq D(\beta Q^*_1 +(1-\beta) Q^*_2 \| P)\\
	&\leq \beta  \min_{ \mathbb{E}_Q [X] \geq \gamma_1 } D(Q\|P) +(1-\beta)  \min_{ \mathbb{E}_Q [X] \geq \gamma_2 } D(Q\|P).
	\end{align}	 
	Hence $E(\gamma)$ is  convex in $\gamma$. 	
\end{proof}
From above lemma we can show that $\lambda$  is a non-decreasing function of $\hat{\gamma}$. From the envelope theorem \cite{Segal}
\begin{equation}
\frac{\partial 	\hat{E}_0  }{\partial \hat{\gamma}} = \lambda^*,
\end{equation}
where $\lambda^*$ is the optimizing $\lambda$ in \eqref{eq:tilted} for the test $\hat{\phi}$. Therefore
\begin{equation}
\frac{\partial  \lambda^* }{\partial \hat{\gamma}}=  \frac{\partial ^2	\hat{E}_0  }{\partial \hat{\gamma}^2} \geq 0,
\end{equation}
where the inequality is from convexity of 	$\hat{E}_0$ respect to $\hat{\gamma}$. Therefore, we only need to consider the behavior of ${\rm Var}_{\hat{P}_0} \Big[\frac{\hat{Q}_{\lambda}(X)}{\hat{P}_0(X)}  \Big]$  as $\lambda$ changes. Taking the derivative of variance respect to $\lambda$, we have
\begin{align}
\frac{\partial }{\partial \lambda}{\rm Var}_{\hat{P}_0} \Bigg[\frac{\hat{Q}_{\lambda}(X)}{\hat{P}_0(X)}  \Bigg]=&\sum_{x\in \Xc}  \frac{2{\hat{Q}_{\lambda}}(x)}{\hat{P_0}(x)} \frac{\partial \hat{Q}_{\lambda}(x) }{\partial \lambda}\\
=& \sum_{x\in \Xc}  \frac{2{\hat{Q}_{\lambda}}(x)}{\hat{P_0}(x)}  \Bigg( \hat{Q}_{\lambda}(x)  \log \frac{\hat{P}_1(x)}{\hat{P}_0(x)}  \nonumber \\ 
&-\hat{Q}_{\lambda}(x)  \sum_{x' \in \Xc} \hat{Q}_{\lambda}(x')\log \frac{\hat{P}_1(x')}{\hat{P}_0(x')} \Bigg ) \\
=&2 \mathbb{E}_{\hat{Q}_{\lambda}} \bigg[ \frac{\hat{Q}_{\lambda}(X)}{\hat{P}_0(X)} \log  \frac{\hat{P}_1(X)}{\hat{P}_0(X)} \bigg  ] \nonumber \\
 &-2 \mathbb{E}_{\hat{Q}_{\lambda}} \bigg [ \frac{\hat{Q}_{\lambda}(X)}{\hat{P}_0(X)}\bigg ] \mathbb{E}_{\hat{Q}_{\lambda}} \bigg [ \log  \frac{\hat{P}_1(X)}{\hat{P}_0(X)} \bigg ].
\end{align}
Substituting $\hat{Q}_{\lambda}(X)$ as a function of $\lambda$ we get
\begin{align}
&\frac{\sum_{a\in \Xc} \hat{P}_0^{1-\lambda}(a) \hat{P}_1^{\lambda}(a)  }{2} \frac{\partial }{\partial \lambda}{\rm Var}_{\hat{P}_0} \bigg[\frac{\hat{Q}_{\lambda}(X)}{\hat{P}_0(X)}  \bigg]\nonumber \\
=&\mathbb{E}_{\hat{Q}_{\lambda}} \bigg [ \bigg (\frac{\hat{P}_{1}(X)}{\hat{P}_0(X)} \bigg  )^{\lambda} \log  \frac{\hat{P}_1(X)}{\hat{P}_0(X)} \bigg  ] \nonumber \\
& -  \mathbb{E}_{\hat{Q}_{\lambda}} \bigg [ \bigg (\frac{\hat{P}_1(X)}{\hat{P}_0(X)} \bigg ) ^{\lambda} \bigg ] \mathbb{E}_{\hat{Q}_{\lambda}} \bigg [ \log  \frac{\hat{P}_1(X)}{\hat{P}_0(X)}  \bigg ]. \label{eq:varlogsum}
\end{align}
Let $r(X)= \Big (\frac{\hat{P}_1(X)}{\hat{P}_0(X)}\Big )^{\lambda}$, then
\begin{align} \label{eq:varlogsum}
\mathbb{E}_{\hat{Q}_{\lambda}} &\bigg [ \bigg (\frac{\hat{P}_{1}(X)}{\hat{P}_0(X)} \bigg  )^{\lambda} \log  \frac{\hat{P}_1(X)}{\hat{P}_0(X)} \bigg  ]  \notag \\  
&~~~~~~~~~~~~~~~-\mathbb{E}_{\hat{Q}_{\lambda}} \bigg [ \bigg (\frac{\hat{P}_1(X)}{\hat{P}_0(X)} \bigg ) ^{\lambda} \bigg ] \mathbb{E}_{\hat{Q}_{\lambda}} \bigg [ \log  \frac{\hat{P}_1(X)}{\hat{P}_0(X)}  \bigg ]  \notag \\
&=\frac{1}{\lambda} \mathbb{E}_{\hat{Q}_{\lambda}} [ r(X)  \log r(X) ]- \frac{1}{\lambda}  \mathbb{E}_{\hat{Q}_{\lambda}} [ r(X) ] \mathbb{E}_{\hat{Q}_{\lambda}} [ \log r(X)  ]. 
\end{align} 
Note that $\hat{Q}_{\lambda}(x),r(x)$ are positive for all $x\in \Xc$. Therefore, using the log-sum inequality \cite{Cover} for the first term and Jensen inequality \cite{Cover} for the second term in \eqref{eq:varlogsum}, we obtain
\begin{align}
&\frac{\lambda \sum_{a\in \Xc} \hat{P}_0^{1-\lambda}(a) \hat{P}_1^{\lambda}(a)  }{2}\frac{\partial }{\partial \lambda}{\rm Var}_{\hat{P}_0} \Bigg[\frac{\hat{Q}_{\lambda}(X)}{\hat{P}_0(X)}  \Bigg] \nonumber \\ 
& \geq	\mathbb{E}_{\hat{Q}_{\lambda}} [ r(X) ] \log \mathbb{E}_{\hat{Q}_{\lambda}} [  r(X)  ]-  \mathbb{E}_{\hat{Q}_{\lambda}} [ r(X) ] \mathbb{E}_{\hat{Q}_{\lambda}} [ \log r(X)  ]  \\
&\geq		\mathbb{E}_{\hat{Q}_{\lambda}} [ r(X) ] \log \mathbb{E}_{\hat{Q}_{\lambda}} [  r(X)  ]-  \mathbb{E}_{\hat{Q}_{\lambda}} [ r(X) ] \log \mathbb{E}_{\hat{Q}_{\lambda}} [  r(X)  ]\\
&=0.
\end{align} 
Also, the above inequalities are met with equality when both log-sum and Jensen's inequalities are met with equality, which happens when $\lambda=0$. Therefore, for $ \lambda>0$,  ${\rm Var}_{\hat{P}_0} \Big[\frac{\hat{Q}_{\lambda}(X)}{\hat{P}_0(X)}  \Big]$ is an increasing function of $\lambda$  and consequently 
\begin{equation}
\frac{\partial }{\partial \hat{\gamma}}\theta_0(\hat{P}_0,\hat{P}_1,\hat{\gamma}) \geq 0.
\end{equation}


\section{Proof of Theorem \ref{thm:hoeffupper}}\label{apx:Thoeffupper}

	By Sanov's theorem, the error exponent is given by
	\begin{align}\label{eq:hoefmisexp}
	\hat{E}_0&=\min_{Q: D(Q\|\hat{P}_0) \geq \hat{\gamma}} D(Q\|P_0).
	\end{align}
	The above optimization problem corresponds to the minimization of a convex function over the complement of a convex set, which achieves its optimum value on its boundary. It is clear that when $D(P_0\|\hat{P}_0) \geq \hat{\gamma}$ the error exponent equals to zero since $Q=P_0$ satisfies the optimization constraint in \eqref{eq:hoefmisexp} and $D(P_0\|P_0)=0$. Hence, we assume $D(P_0\|\hat{P}_0) < \hat{\gamma}$. Note that since $P_0$ lies in the relative entropy ball centred at $\hat{P}_0$, the line passing through $P_0$ and $\hat{P}_0$ passes through the interior of the bounded convex set $\Bc(\hat{P}_0,\hat{\gamma})$  and therefore cuts the set's boundary in exactly two points, and there exist a point $Q^*$ such that $P_0$ lies in between of intersecting point and $\hat{P}_0$. By writing $P_0$ as linear combination of $\hat{P}_0$, $Q^*$  such that $D(Q^*\|\hat{P}_0) = \hat{\gamma}$ and $P_0=\beta Q^* +(1-\beta) \hat{P}_0$ where $ 0 \leq \beta \leq 1$ we get
	\begin{align}
	\hat{E}_0 &=  \min_{Q: D(Q\|\hat{P}_0) \geq \hat{\gamma}} D(Q\|\beta Q^* +(1-\beta) \hat{P}_0)  \\
	&\leq  \min_{Q: D(Q\|\hat{P}_0) \geq \hat{\gamma}} \beta D(Q\| Q^*)+(1-\beta) D(Q\| \hat{P}_0)  \\
	&= (1-\beta) \hat{\gamma}, \label{eq:upperHoef} 
	\end{align}
	where the inequality is by the convexity of relative entropy. In order to get equality \eqref{eq:upperHoef} we lower bound the minimization by
	\begin{align}
	\min_{Q: D(Q\|\hat{P}_0) \geq \hat{\gamma}} &\beta D(Q\| Q^*)+(1-\beta) D(Q\| \hat{P}_0)  \\
	& \geq \min_{Q: D(Q\|\hat{P}_0) \geq \hat{\gamma}} (1-\beta) D(Q\| \hat{P}_0)\\
	& \geq (1-\beta) \hat{\gamma},
	\end{align}
	and by choosing $Q=Q^*$  we can achieve this lower bound. Next we find a lower bound to  $\beta$. By definition of $Q^*$ we have
	\begin{equation}
	D\Big(\frac{1}{\beta} P_0 - \frac{1-\beta}{\beta} \hat{P}_0 \| \hat{P}_0 \Big) =\hat{\gamma}.  
	\end{equation}
	Next, by Pinsker's inequality and lower bounding the relative entropy we get
	\begin{equation}\label{eq:alphalower} 
	2\Big\|\frac{1}{\beta} P_0 - \frac{1-\beta}{\beta} \hat{P}_0 - \hat{P}_0 \Big\|_{\rm TV}^2=\frac{2}{\beta^2}\| P_0-\hat{P}_0\|^2_{\rm TV}  \leq \hat{\gamma}.
	\end{equation}
	By equations \eqref{eq:upperHoef} and  \eqref{eq:alphalower} we conclude the result.


\section{Proof of Theorem \ref{thm:lowerworstHoef}}\label{apx:TlowerworstHoef}

	As opposed to the likelihood ratio test, we first consider the minimization over $P_0$ for a fixed $Q$, and proceed with a Taylor expansion of $D(Q\|P_0)$ around $P_0=\hat{P}_0$  to get
	\begin{equation}
	\underline{\hat{E}}_0 (r_0)=  \min_{\substack{P_0: D(\hat{P}_0\| P_0)\leq r_0 \\ Q: D(Q\|\hat{P}_0) = \hat{\gamma}}} D(Q\|\hat{P}_0) +  \thetav_{P_0}^{T}   \nabla {E}_0  +o(\| \thetav_{P_0} \|_{\infty}),
	\end{equation}
	where
	\begin{align}
	\nabla E_0 &= \bigg( -\frac{Q(x_1)}{\hat{P}_0(x_1)},\dotsc, -\frac{{Q}(x_{|\Xc|})}{\hat{P}_0(x_{|\Xc|})}\bigg)^T,\\
	\thetav_{P_0}		 &= \big(P_0(x_1)-\hat{P}_0(x_1),\dotsc,P_0(x_{|\Xc|})-\hat{P}_0(x_{|\Xc|})\big)^T,
	\end{align}	
	and we replaced the inequality constraint by an equality one, since  the problem is that of optimizing a convex function over the complement of the convex set which attains its optimal value on the boundary of the set. Next by optimizing over $P_0$ and similarly to the proof of Theorem \ref{thm:lowerworst}, by substituting $Q$ with $\hat{Q}_{\lambda}$ we get
	
	\begin{align}\label{eq:worstapproxopthoef}
\underline{\hat{E}}_0 (r_0) =& \min_{\substack{ Q: D(Q\|\hat{P}_0) = \hat{\gamma}}}   \min_{\substack{ \frac{1}{2} \thetav_{P_0}^T \Jm(\hat{P}_0) \thetav_{P_0} +o (\| \thetav_{P_0} \|^2_{\infty}) \leq r_{0} \\ \onev^T\thetav_{P_0}=0}} \Big \{ E_0   \nonumber \\
&~~~~~~~~~~~~~~~~~~~ +     \thetav_{P_0}^{T}   \nabla {E}_0+o(\|\thetav_{P_0} \|_{\infty}) \Big \} \\
=&\min_{\substack{ Q: D(Q\|\hat{P}_0) = \hat{\gamma}}}    \min_{\substack{ \frac{1}{2} \thetav_{P_0}^T \Jm(\hat{P}_0) \thetav_{P_0}  \leq r_{0} \\ \onev^T\thetav_{P_0}=0}} \Big \{ E_0  +     \thetav_{P_0}^{T}   \nabla \hat{E}_0 \Big \} \nonumber \\
&+o(\sqrt{r_0}),  \label{eq:approxerrorhoef}
\end{align}	
	where to get \eqref{eq:approxerrorhoef} we have taken the similar steps to the proof of  Theorem \ref{thm:lowerworst}, with the difference that now the vector $\thetav^*_{P_0}  $ is depending on the choice of $Q$. However, this dependency does not change the remainder term order. Observe that, from the inequality constraint and the restriction it imposes on the length of the vector $\thetav_{P_0}  $ we have that   $\frac{1}{2} \thetav_{P_0}^T \Jm(\hat{P}_0) \thetav_{P_0} \leq r_{0} +o(\sqrt{r_0}) $ and hence $\frac{\alpha}{2}\|\thetav^*_{P_0}\|_{2}^2  \leq \max_{x\in \Xc} \hat{P}_0(x) (r_{0} +o(\sqrt{r_0})) $, and by equivalency of $p$-norms, we obtain the remainder term. Finally, similarly to  the Theorem \ref{thm:lowerworst} by solving the inner optimization problem, we have
	\begin{align}
	\underline{\hat{E}}_0 (r_0)=& \min_{\substack{ Q: D(Q\|\hat{P}_0) = \hat{\gamma}}} D(Q\|\hat{P}_0)   - \sqrt{\frac{2}{\alpha}  {\rm Var}_{\hat{P}_0} \Bigg[\frac{Q(X)}{\hat{P}_0(X)}  \Bigg] r_0} \nonumber \\ 
	&+ o(\sqrt{r_0})  \\
	=& \hat{\gamma} - \max_{\substack{ Q: D(Q\|\hat{P}_0) = \hat{\gamma}}} \sqrt{\frac{2}{\alpha}  {\rm Var}_{\hat{P}_0} \Bigg[\frac{Q(X)}{\hat{P}_0(X)}  \Bigg]   r_0} + o(\sqrt{r_0}), 
	\end{align}
which completes the proof.

\section{Proof of Proposition \ref{cor:comparing}}\label{apx:Ccomparing}

	We are comparing the sensitivity of the likelihood ratio test and Hoeffding's test when they achieved the same type-\RNum{1} error exponents, and where the likelihood ratio test uses the test distribution $\hat{P}_1$. Let $\hat{Q}_{\lambda}$ be the minimizing distribution of likelihood ratio test. By  letting $Q=\hat{Q}_{\lambda}$ in \eqref{eq:hoefsen} we get $D(\hat{Q}_{\lambda} \|\hat{P}_0)=\hat{\gamma}$ which satisfies the maximization constraint, therefore
	\begin{align}
	\theta_0^{\rm h}(\hat{P}_0,\hat{\gamma}^{\rm h})  &\geq  \sqrt{\frac{2}{\alpha}  {\rm Var}_{\hat{P}_0} \Bigg[\frac{Q(X)}{\hat{P}_0(X)}  \Bigg]} \\
	&= \theta_0^{\rm lrt}(\hat{P}_0,\hat{P}_1,\hat{\gamma}^{\rm lrt}),
	\end{align}
	which gives the lower bound. To upper bound the sensitivity ratio, we have
	\begin{align}
	\frac{\theta_0^{\rm h}(\hat{P}_0,\hat{\gamma}^{\rm h})}{\theta_0^{\rm lrt}(\hat{P}_0,\hat{P}_1,\hat{\gamma}^{\rm lrt})}\leq   \frac{\max_{{ Q: D(Q\|\hat{P}_0) = {\hat{\gamma}^{\rm h}}}} \sqrt{ {\rm Var}_{\hat{P}_0} \Big[\frac{Q(X)}{\hat{P}_0(X)}  \Big]}  }{\min_{\hat{P}_1: D(\hat{Q}_\lambda\|\hat{P}_0) = \hat{\gamma}^{\rm h} } \sqrt{  {\rm Var}_{\hat{P}_0} \Big[\frac{\hat{Q}_\lambda(X)}{\hat{P}_0(X)}  \Big]} },
	\end{align}
	 where the constraint in the minimization term is the set of all test distributions $\hat{P}_1$ and corresponding $\hat{\gamma}^{\rm lrt}$ such that the likelihood ratio test achieves the same type-\RNum{1} error exponents as Hoeffding's test which equals to $\hat{\gamma}^{\rm h}$. Also, note that ${\rm Var}_{\hat{P}_0} \Big[\frac{Q(X)}{\hat{P}_0(X)}  \Big] = \chi^2(Q\|\hat{P}_0)$ where $\chi^2$ is the chi-squared distance.  For every $\hat{P}_0, Q$ we have $D(Q\|\hat{P}_0) \leq \chi^2(Q\|\hat{P}_0)$ \cite{Choosing}. Therefore, 
	\begin{equation}\label{eq:lower}
	\min_{\hat{P}_1 : D(\hat{Q}_\lambda\|\hat{P}_0) = \hat{\gamma}^{\rm h} } \sqrt{  {\rm Var}_{\hat{P}_0} \Bigg[\frac{\hat{Q}_\lambda(X)}{\hat{P}_0(X)}  \Bigg]} \geq \sqrt{\hat{\gamma}^{\rm h}}.
	\end{equation}
	In addition, we  upper bound the chi-squared distance as follows
	\begin{align}
	\chi^2(Q\|\hat{P}_0)&\leq \frac{1}{\min_{x\in \Xc} \hat{P}_0(x)} \|Q-\hat{P}_0\|_2^2  \\
	&\leq \frac{1}{\min_{x\in \Xc} \hat{P}_0(x)}  \| Q-\hat{P}_0\|_1^2  \\
	&\leq  \frac{4}{\min_{x\in \Xc} \hat{P}_0(x)} D(Q\|\hat{P}_0), \label{eq:chi2KL}
	\end{align}
	where we used Pinsker's inequality \cite{Cover} in the last step. Hence, we have
	\begin{equation}\label{eq:upper}
	\max_{{ Q: D(Q\|\hat{P}_0) = {\hat{\gamma}^{\rm h}}}} \sqrt{ {\rm Var}_{\hat{P}_0} \Bigg[\frac{Q(X)}{\hat{P}_0(X)}  \Bigg] }  \leq \sqrt{\frac{4}{\min_{x\in \Xc} \hat{P}_0(x)} }.
	\end{equation}
	Finally, from \eqref{eq:lower}, \eqref{eq:upper} we conclude \eqref{eq:cor}.  We can also improve the chi-squared upper bound using the following $f$-divergences inequality \cite{Sason}
	\begin{align}
	\chi^2(Q\|\hat{P}_0)&\leq \frac{1}{\kappa(\beta)}  D(Q\|\hat{P}_0)   
	\end{align}
	where for $x \in (1,\infty)$ 
	\begin{equation}
	\kappa(x)=\frac{x\log x +1-x}{(x-1)^2},
	\end{equation}
	and $\beta = \max_{x \in \Xc} \frac{Q(x)}{\hat{P}_0(x)}$. We need to find $\beta$ such that for all distribution $Q$ satisfying the $D(Q\|\hat{P}_0)=\hat{\gamma}^{\rm h}$ this inequality holds.  It is easy to show that $\kappa(x)$ is a non-increasing function for $x \in (1,\infty)$. To show this we take the derivative of $\kappa(x)$,
	\begin{align}
	\frac{d \kappa(x)}{dx}&= \frac{2(1-x)+(x+1)\log x}{(1-x)^3}\\
	&\leq  \frac{2(1-x)+(x+1)(x-1)}{(1-x)^3}\\
	&=\frac{1}{1-x}\\
	&\leq 0
	\end{align}
	where we have used the inequality $\log x \leq x-1$. Therefore for all $D(Q\|\hat{P}_0)=\hat{\gamma}^{\rm h}$
	\begin{equation}
	\chi^2(Q\|\hat{P}_0)\leq \frac{1}{\kappa \Big (\max_{Q: D(Q\|\hat{P}_0)=\hat{\gamma}^{\rm h}} \beta\Big )}  D(Q\|\hat{P}_0). \label{eq: chi2upper} 
	\end{equation}
	For every $x >0$  by the inequality $x \log x \geq x-1$,  for every $Q$ such that $D(Q\|\hat{P}_0)=\hat{\gamma}^{\rm h} $we get
	\begin{align}
	\hat{\gamma}^{\rm h} &\geq \sum_{x\in \Xc} \hat{P}_0(x) \Big(\frac{Q(x)}{\hat{P}_0(x)}-1\Big ).
	\end{align}
	Therefore
	\begin{align}
	\max_{Q \in \Pc(\Xc) : D(Q\|\hat{P}_0)=\hat{\gamma}^{\rm h}} & \max_{x \in \Xc} \frac{Q(x)}{\hat{P}_0(x)}  \nonumber \\ 
	&\leq \max_{Q: \sum_{x\in \Xc} \hat{P}_0(x) \big(\frac{Q(x)}{\hat{P}_0(x)}-1\big) \leq \hat{\gamma}^{\rm h}  }  \max_{x \in \Xc} \frac{Q(x)}{\hat{P}_0(x)} \\
	&\leq \frac{\min(1, \hat{\gamma}^{\rm h} +\min_{x\in \Xc} \hat{P}_0(x))  }{\min_{x\in \Xc} \hat{P}_0(x)  } \label{eq:maxmax} 
	\end{align}  
	 where in the last step we have dropped the $Q \in \Pc(\Xc)$ condition and assigned all the mass to $Q(x_{\rm min})$ where $x_{\rm min} =\argmin_{x\in \Xc} \hat{P}_0(x)$. Finally, substituting \eqref{eq:maxmax} into \eqref{eq: chi2upper} we get \eqref{eq:lrttohoef}.


\section{Proof of Theorem \ref{thm:seqMM}}\label{apx:TseqMM}

The proof described below holds in general and can be used for continuous probability distributions. From the absolute continuity assumption, let the log-likelihood ratio be  bounded by a positive constant $c$, i.e.,
\begin{equation} \label{eq:llrbound}
\bigg |\log \frac{\hpo}{\hpt}\bigg| \leq c ~~~ \forall x.
\end{equation}
We use the following results.
\begin{theorem}[\cite{Woodroofe}] \label{thm:converge}
	Let $S_n=\sum_{i=1}^{n} Z_i$ be a random walk where $Z_i$ is some non-lattice random variable\footnote{A random variable $Z$ is said to be lattice if and only if $\sum_{k=-\infty}^{\infty} \text{Pr}[ Z=a+kd]=1$ for some non-negative $a,d$. Otherwise, it is said to be non-lattice.} generated in the i.i.d fashion with $\mathbb{E}[Z_i] >0$.  For $\gamma>0$, let
	\begin{equation}
	\tau=\inf\{n\geq 1: S_n \geq  \gamma\}.
	\end{equation}
	Also, let $R_{\gamma}\triangleq S_{\tau}-\gamma$.
	Then $R_{\gamma}$ converges in distribution to a random variable $R$ with distribution $Q$ as $\gamma\to\infty$. 
	Moreover, if $Z$ is lattice random variable, then $R_{\gamma}$ has a limiting distribution $Q_d$ as $\gamma\to \infty$ through multiples of $d$. 
\end{theorem}

The next result shows that under conditions \eqref{eq:posdrift}, the mismatched sequential probability ratio test stops at a finite time.

\begin{lemma}\label{lem:finite}
	
	Let $\hta_0$ be the the smallest time that the mismatched sequential probability ratio test crosses threshold $\hgo$, i.e.,  
	\begin{align}\label{eq:tau1}
	\hta_{0}=\inf \{n\geq1: \hat{S}_n& \geq \hgo\}.  
	\end{align}   
	Also, assume that conditions \eqref{eq:posdrift} hold. Then,
	\begin{equation}\label{eq:finiteupper}
	\PP_0[\hta_0 \geq n] \leq      e^{d\hgo} e^{-(n-1) E(0)},
	\end{equation}
	where $E(0),d >0$. Also, as $ \hgo\to  \infty $,  $\hta_0\to  \infty$ almost surely.
\end{lemma}

\begin{proof}
	By Chernoff bound \cite{Dembo}, the probability of passing the threshold  under the first hypothesis at a time after $n$ can be upper bounded by
	\begin{align}
	\PP_0[\hta_0 \geq n] &\leq \PP_0 \Big[ \hat{S}_{n-1}  \leq \hgo\Big  ]\\
	&= \PP_0 \Bigg[   \sum_{i=1}^{n-1} \log \frac{\hat{P}_0(x_i)}{\hat{P}_1(x_i)} \leq \hgo\Bigg  ] \\
	& \leq e^{-(n-1)  E\big(\frac{\hgo}{n-1}\big)} \label{eq:chernoff},
	\end{align}
	where 
	\begin{equation}\label{eq:lagrangelem}
	E(\gamma)= \sup_{s \geq 0} \Big \{- s{\gamma} - \hat{\kappa}(s) \Big\} , 
	\end{equation}
	and
	\begin{equation}
	\hat{\kappa}(s)= \log \E_{P_0} \Bigg[ \frac{\hat{P}_1^s}{\hat{P}_0^s} \Bigg],
	\end{equation}
	is the cumulant function of the mismatched log-likelihood ratio. Note that for each $s$ the objective function in \eqref{eq:lagrangelem} is linear in $\gamma$, so $E(\gamma)$ is the pointwise supremum of a family of linear functions, hence convex \cite{Boyd}. By the convexity of $E(\gamma)$ we have the following lower bound
	\begin{equation}\label{eq:linearseq}
	E \Big (\frac{\hgo}{n-1} \Big ) \geq   E(0)+    \frac{\partial  E(\gamma)}{\partial \gamma} \Bigg|_{\gamma =0}  \frac{\hgo}{n-1}.
	\end{equation}
	In order to show that $ E(0)>0$ it suffices to show that $\hat{\kappa}'(s=0) <0$. Taking derivative of $\hat{\kappa}(s)$ respect to $s$ and setting $s=0$,  we have
	\begin{align}
	\hat{\kappa}'(0)&= \E_{P_0} \Bigg [\log \frac{\hat{P}_1}{\hat{P}_0} \Bigg ]\\
	&= D(P_0\|\hat{P}_0)- D(P_1\|\hat{P}_1)<0,
	\end{align}
	where the last step is by the assumption in \eqref{eq:posdrift}. Finally, by the envelope theorem \cite{Segal} we get 
	\begin{align}
	\frac{\partial  E(\gamma)}{\partial \gamma} \Bigg|_{\gamma =0} = -s^*(\gamma=0),
	\end{align} 
	where $s^*({\gamma=0})$ is the optimizing value of $s$  in \eqref{eq:lagrangelem} evaluated when $\gamma=0$. (Note that this value is unique since $\hat{\kappa}(s)$ is strictly convex in $s$ \cite{Dembo}). By the constraint of the optimization problem in \eqref{eq:lagrangelem}, we have  $s\geq0$. Also, from $\hat{\kappa}'(0)<0$, we get $s^*({\gamma=0})\neq0$. Therefore, we conclude that $s^*(\gamma=0) >0$ and
	\begin{align}
	d&\triangleq { \Bigg | \frac{\partial  E(\gamma)}{\partial \gamma} \Big|_{\gamma =0} \Bigg | } = s^*(\gamma=0) > 0. 
	\end{align} 
	Finally, substituting \eqref{eq:linearseq} into \eqref{eq:chernoff} we get \eqref{eq:finiteupper}. Furthermore,
	\begin{align}
	\PP_0[\hta_0 \leq n] &\leq  \PP_0 \Bigg[ \sum_{i=1}^{k} \log \frac{\hat{P}_0(x_i)}{\hat{P}_1(x_i)} \geq \hgo, k \leq n\Bigg  ]\\
	&\leq  \PP_0 \Big[   kc \geq \hgo, k\leq n \Big]\\
	&\leq \PP_0[nc\geq \hgo]. 
	\end{align}
	Taking $\hgo > nc$, then $\PP_0[\hta_0  \leq n]=0$. Therefore, as $\hgo\rightarrow  \infty$ then $\hta_0\rightarrow \infty$ a.s.
	%
\end{proof}

We now proceed with the proof of the Theorem. We show the result for the type-\RNum{2} error probability; a similar proof holds for the type-\RNum{1} case. The type-\RNum{2} probability of error of mismatched sequential probability ratio test is
\begin{align}
\het &=\E_{P_1} \big[ \mathds{1} \{\hat{S}_{\hta} \geq \hgo\} \big ]\\
&=\E_{P_0} \big [ e^{-{S}_{\hta}}  \mathds{1} \{\hat{S}_{\hta} \geq \hgo\}  \big ], \label{eq:errprobseq1}
\end{align}
where $S_{\hta}$ is the log-likelihood ratio under no mismatch in \eqref{eq:LLR} evaluated at the time where the mismatched test stops. Recall the definition of  $\hta_{0}$ in \eqref{eq:tau1}
and $\hat{R}_{\hgo}=\hat{S}_{\hta_{0}}-\hgo$. Observe that, if $\hat{S}_{\hta} \geq \hgo$, then we have $\hta=\hta_{0}$. Multiplying the exponent in \eqref{eq:errprobseq1}  by $\frac{\hat{S}_{\hta_{0}} }{\hat{S}_{\hta_{0}} }$ and substituting $\hro$ we get
\begin{align}\label{eq:errprobseq2}
\het &= \E_{P_0} \Big [ e^{-{S}_{\hta_{0}}}    \mathds{1} \{\hat{S}_{\hta} \geq \hgo\}   \Big ]  \\
&= \E_{P_0} \Big[ e^{-\frac{{S}_{\hta_{0}}}{\hat{S}_{\hta_{0}}}.\hat{S}_{\hta_{0}} }  \mathds{1} \{\hat{S}_{\hta} \geq \hgo\} \Big ]  \\
&=\E_{P_0} \Big[ e^{-\frac{{S}_{\hta_{0}}}{\hat{S}_{\hta_{0}}}(\hro+\hgo)} \mathds{1}\{\hat{S}_{\hta} \geq \hgo\} \Big ]. \label{eq:err}
\end{align}
Let $\mu=\frac{S_{\hta_{0}}}{\hta_{0}}$, $\hat{\mu}=\frac{\hat{S}_{\hta_{0}}}{\hta_{0}}$. By Lemma \ref{lem:finite},  $\hta_{0} \rightarrow \infty$ as $\hgo \rightarrow \infty$ a.s., and therefore by the WLLN
\begin{align}
\mu &\xrightarrow[]{p} D(P_0\|P_1),\\
\hat{\mu} &\xrightarrow[]{p} D(P_0\|\hat{P}_1)-D(P_0\|\hat{P}_0). 
\end{align} 
Also, since $\hat{\mu}>0$ almost surely, by using the continuous mapping theorem \cite{Resnick} we have
\begin{align}
\frac{\mu}{\hat{\mu}}= \frac{S_{\hta_{0}}}{\hat{S}_{\hta_{0}}}\xrightarrow[]{p} \frac{D(P_0\|P_1)}{ D(P_0\|\hat{P}_1)-D(P_0\|\hat{P}_0)}. 
\end{align} 		  
Moreover, by Theorem \ref{thm:converge},  $\hro$  converges in distribution to a random variable $\hat{R}_0$ with limiting distribution $\hat{Q}_0$ under $P_0$ (through multiples of d in the lattice case). By the Slutsky's theorem \cite{Resnick},
\begin{equation}
{\frac{{S}_{\hta_{0}}}{\hat{S}_{\hta_{0}}}\cdot\hro} \xrightarrow[]{d} \frac{D(P_0\|P_1)}{ D(P_0\|\hat{P}_1)-D(P_0\|\hat{P}_0)} \hat{R}_0,
\end{equation}
Writing \eqref{eq:err} as
\begin{align}
&\het  e^{\frac{D(P_0\|P_1)}{ D(P_0\|\hat{P}_1)-D(P_0\|\hat{P}_0)} \hgo} \nonumber \\
&= \E_{P_0} \Big[ e^{-\frac{{S}_{\hta_{0}}}{\hat{S}_{\hta_{0}}}(\hro+\hgo)+ \frac{D(P_0\|P_1)}{ D(P_0\|\hat{P}_1)-D(P_0\|\hat{P}_0)} \hgo  } \mathds{1}\{\hat{S}_{\hta} \geq \hgo\} \Big ], 
\end{align}
and letting $\hgo \rightarrow \infty$, by Slutsky's theorem we get 
\begin{align}
\lim_{\hgo \rightarrow \infty}\het  e^{\frac{D(P_0\|P_1)}{ D(P_0\|\hat{P}_1)-D(P_0\|\hat{P}_0)} \hgo} &= \E_{P_0} \Big [ e^{-\frac{D(P_0\|P_1)}{ D(P_0\|\hat{P}_1)-D(P_0\|\hat{P}_0)} \hat{R}_0} \Big ]\\
&=\hat{c}_1  .  
\end{align}

To prove \eqref{eq:SPRTthresh}, we show the converges of $\hta_0$ in probability as well as its uniform integrability. Therefore, we can conclude its convergence in $L^1$ norm  (and hence in expectation). Finally, from the convergence of $\hta_0$, we obtain the convergence of $\hta$. First, by the finiteness of $\hta_0$ for every $\hgo$ and definition of $\hta_0$, there exist a finite $\hta_0$ with probability one such that
\begin{align}\label{eq:convergpseq}
\hat{s}_{\hta_0-1} < \hgo \leq \hat{s}_{\hta_0}~~~   \text{w.p}.1.
\end{align}

Also, let $n = \frac{\hat{\gamma}_0}{c+1}$, where $c$ is the log-likelihood ratio bound in \eqref{eq:llrbound},  we have
\begin{align}
&\PP\Bigg [ \Bigg |\frac{\hat{S}_{\hta_0}}{\hta_0}- \big (D(P_0\|\hat{P}_1)-D(P_0\|\hat{P}_0)\big ) \Bigg | > \epsilon  \Bigg] \nonumber \\ 
&= \PP\Bigg [ \Bigg |\frac{\hat{S}_{\hta_0}}{\hta_0}-\big(D(P_0\|\hat{P}_1)-D(P_0\|\hat{P}_0)\big) \Bigg | > \epsilon , \hta_0 >n  \Bigg] \nonumber \\
&+\PP\Bigg [ \Bigg |\frac{\hat{S}_{\hta_0}}{\hta_0}-\big(D(P_0\|\hat{P}_1)-D(P_0\|\hat{P}_0)\big) \Bigg | > \epsilon , \hta_0 \leq n \Bigg]  \\
& \leq \PP\Bigg [ \bigcup_{t=n}^\infty \Bigg |\frac{\hat{S}_{t}}{t}-\big(D(P_0\|\hat{P}_1)-D(P_0\|\hat{P}_0)\big) \Bigg | > \epsilon \Bigg]  \nonumber \\ 
&~~~+ \PP\big [\hta_0 \leq n \big] 
\end{align}
where $ \PP\big [\hta_0 \leq n \big]=0$ for  the choice of $n$ by Lemma \ref{lem:finite}. Now by letting $\hgo \rightarrow \infty$, and hence $n\rightarrow \infty$ we get
\begin{align}
&\lim_{\hgo \rightarrow \infty} \PP\Bigg [ \Bigg |\frac{\hat{S}_{\hta_0}}{\hta_0}- \big (D(P_0\|\hat{P}_1)-D(P_0\|\hat{P}_0)\big ) \Bigg | > \epsilon  \Bigg]  \nonumber \\ 
&\leq  \PP\Bigg [\bigcap_{n=1}^{\infty} \bigcup_{t=n}^\infty \Bigg |\frac{\hat{S}_{t}}{t}-\big(D(P_0\|\hat{P}_1)-D(P_0\|\hat{P}_0)\big) \Bigg | > \epsilon \Bigg] \\
&=0
\end{align}
by the strong law of large numbers. Therefore, we have  
\begin{align}
\frac{\hat{S}_{\hta_0}}{\hta_0} \xrightarrow[]{p} D(P_0\|\hat{P}_1)-D(P_0\|\hat{P}_0),\label{eq:convergp1}\\
\frac{\hat{S}_{\hta_0-1}}{\hta_0-1} \xrightarrow[]{p} D(P_0\|\hat{P}_1)-D(P_0\|\hat{P}_0).\label{eq:convergp2}
\end{align}
Therefore, by \eqref{eq:convergpseq}, \eqref{eq:convergp1}, \eqref{eq:convergp2} we can conclude that 
\begin{equation}\label{eq:convP1adv}
\frac{\hta_0}{\hgo}  \xrightarrow[]{p}   \frac{1}{ D(P_0\|\hat{P}_1)-D(P_0\|\hat{P}_0)} 
\end{equation}  	
as	$\hgo\rightarrow  \infty$.

To show the convergence in $L^1$ we only need to prove the uniform integrability of the sequence of random variables $\frac{\hta_0}{ \hgo}$, where $\hta_0$ is a random variable that depends on  the strictly positive parameter $\hgo$. Equivalently, for some $\epsilon >0$, we need to show that,
\begin{equation}\label{eq:uniform}
\lim_{t \rightarrow \infty} \sup_{\hgo \geq \epsilon} \mathbb{E}_{P_0} \Bigg [\frac{\hta_0}{ \hgo} \mathds{1}\Big \{\frac{\hta_0}{ \hgo} \geq t \Big  \}     \Bigg] =0.
\end{equation}
We can upper bound the given expectation in \eqref{eq:uniform} as  
\begin{align}
\underbrace{\mathbb{E}_{P_0} \Bigg [\frac{\hta_0- \floor{t  \hgo}   }{\hgo} \mathds{1}\Big \{{\hta_0} \geq \floor{t  \hgo}  \Big  \}     \Bigg]}_{A}+  \underbrace{ t \mathbb{E}_{P_0} \Bigg [  \mathds{1}\Big \{\frac{\hta_0}{ \hgo} \geq t \Big  \}     \Bigg]}_{B}.
\end{align}
The second term can be upper bounded by \eqref{eq:finiteupper} as
\begin{align}
B&= t \PP_0[ \hta_0 \geq t  \hgo ] \leq t e^{E(0)} e^{-\hgo(tE(0)-d)}.
\end{align}
The first expectation can be also written as the following sum 
\begin{align}
A&= \frac{1}{\hgo} \sum_{m=1}^{\infty} \PP_0\big [\hta_0-\floor{t\hgo}\geq m \big ],
\end{align}
and by  \eqref{eq:finiteupper} 
\begin{align}
A\leq \frac{1}{\hgo} t e^{-\hgo (t E(0)-d)}  \sum_{m=1}^{\infty}     e^{-(m-2) E(0)}.
\end{align}
Hence $A$ and $B$ are vanishing as $t\rightarrow\infty$ for every $\hgo$ 
giving the uniform integrability of $\frac{\hta_0}{ \hgo}$, and hence convergence in $L^1$  \cite{Bill}, i.e,
\begin{equation} \label{eq:convergT1}
\lim_{\hgo \rightarrow \infty} \mathbb{E}_{P_0} \Bigg [ \Bigg|\frac{\hta_0}{ \hgo}- \frac{1}{D(P_0\|\hat{P}_1)-D(P_0\|\hat{P}_0)} \Bigg | \Bigg]=0.
\end{equation}	

Finally, we prove the convergence of $\hta$. By \eqref{eq:MMexp1}, \eqref{eq:convP1adv} and the union bound, we obtain 
\begin{align}
&\PP_0 \bigg[\bigg|\frac{\hta}{ \hgo}- \frac{1}{D(P_0\|\hat{P}_1)-D(P_0\|\hat{P}_0)} \bigg | \geq \epsilon  \bigg] \nonumber \\ 
&\leq \PP_0 \bigg[\bigg|\frac{\hta}{ \hgo}- \frac{1}{D(P_0\|\hat{P}_1)-D(P_0\|\hat{P}_0)} \bigg | \geq \epsilon , \hat \phi =0 \bigg] \nonumber \\
& + \PP_0[\hat{\phi}=1 ]\\
&=\PP_0 \bigg[\bigg|\frac{\hta_0}{ \hgo}- \frac{1}{D(P_0\|\hat{P}_1)-D(P_0\|\hat{P}_0)} \bigg | \geq \epsilon\bigg] + \heo,  
\end{align}	
which tends to $0$ as $\hgo \rightarrow \infty$, establishing the convergence of $\frac{\hta}{ \hgo}$ in probability. Now, using that $\hta\leq \hta_0$ we have	
\begin{align}\label{eq:upper1}
\mathbb{E}_{P_0} \Bigg [\frac{\hta}{ \hgo} \mathds{1}\Big \{\frac{\hta}{ \hgo} \geq t \Big  \}     \Bigg] \leq  \mathbb{E}_{P_0} \Bigg [\frac{\hta_0}{ \hgo} \mathds{1}\Big \{\frac{\hta_0}{ \hgo} \geq t \Big  \}     \Bigg].
\end{align}
Therefore, uniform integrability of $\hta_0$ gives the uniform integrability of $\hta$, and hence convergence in $L^1$ norm and also expectation of $\frac{\hta}{ \hgo}$, which concludes the proof.


\section{Proof of Theorem \ref{thm:negdrift}}\label{apx:Tnegdrift}

\begin{proof}
	Defining $\hta_1$ similar to \eqref{eq:tau1}, we have
	\begin{equation}
	\hta_{1}=\inf \{n\geq1: \hat{S}_n < -\hgt\}. 
	\end{equation} 
	The probability of making the right decision can be bounded as	
	\begin{align}
	\PP_0[\hat{\phi}=0] &=\PP_0\big [\hat{S}_n ~ \text{passes } \hgo \text{ before passing }  - \hgt     \big ] \\
	& \leq    \sum_{n=1}^{\infty} \PP_0 \big [ \hta_0 \leq n,  \hta_1 > n\big]\\
	&\leq    \sum_{n=1}^{\infty}  \min\big\{ \PP_0  [ \hta_0 \leq n ],  \PP_0  [\hta_1 > n]  \big  \}. 
	\end{align}
	We can bound both terms similar to  the proof of lemma \ref{lem:finite} by
	\begin{align}
	\PP_0[\hat{\phi}=1]& \leq  \sum_{n=1}^{\infty}  \min\Big\{ e^{-a\hat{\gamma}_0} e^{-n \tilde{E}(0)},  e^{-a\hat{\gamma}_1} e^{-n \tilde{E}(0)}   \Big  \} \\
	& =\min\Big\{{c_0}e^{-a\hat{\gamma}_0} , {c_1}e^{-a\hat{\gamma}_1} \Big \},
	\end{align}
	where $a>0,$ $\tilde{E}(0) >0$ due to \eqref{eq:negdrift}. Finally, 
	\begin{equation}
	\heo \geq 1- \max\Big\{{c_1}e^{-a\hat{\gamma}_0} , {c_2}e^{-a\hat{\gamma}_1} \Big \}
	\end{equation}
	which goes to $1$ as $\hat{\gamma}_0, \hat{\gamma}_1$ approach infinity.
\end{proof}


\section{Proof of Corollary \ref{cor:exp}}\label{apx:Cexp}

\begin{proof}	
	Theorem \ref{thm:seqMM} gives an asymptotic expression for the error probability  $\hat{\epsilon}_i$ and expected stopping time $\mathbb{E}_{P_i}[\tau]$  of the mismatched sequential probability ratio test  for $i\in \{0,1\}$ as a function of thresholds. To find the largest error exponents of the test as defined in \eqref{eq:tradeseqMM1} and \eqref{eq:tradeseqMM2} we should find the largest thresholds $\hgo, \hgt$ such that they satisfy the expected time condition since type-\RNum{1} and type-\RNum{2}  error exponents are increasing function of $\hgo, \hgt$ by \eqref{eq:MMexp1}, \eqref{eq:MMexp2}. It is easy to check that thresholds in corollary are the largest thresholds satisfying the expected stopping time conditions, which concludes the proof.
\end{proof}



\section{Proof of Theorem \ref{thm:lowerworstseq}}\label{apx:Tlowerworstseq}
We show the result under hypothesis $0$, and similar steps are valid for hypothesis $1$. Observe that \eqref{eq:threshMMM1} can be written as 
\begin{align}
\hEo &= D(P_0\|P_1)\notag\\ 
&~~\times \min\Bigg \{  \frac{D(\hat{P}_1\|\hat{P}_0)}{D(\hat{P}_0\|\hat{P}_1)}, \frac{D(P_1\|\hat{P}_0)-D(P_1\|\hat{P}_1)}{D(P_0\|\hat{P}_1)-D(P_0\|\hat{P}_0)} \Bigg \}. \label{eq:threshMMM1sim} 
\end{align}	
From \eqref{eq:worst_case} and \eqref{eq:threshMMM1sim}, we need to compute two minimizations, the first of which over $P_0$. To this end, we exchange the order of these minimizations and  apply a Taylor series expansion to the first term of \eqref{eq:threshMMM1sim} around $P_0=\hat{P}_0, P_1=\hat{P}_1$ we obtain

\begin{align}
&\hEo =\notag\\
& \Big ( D(\hat{P}_0\| \hat{P}_1) + \dv_0^T  \thetav_{P_0}  + \dv_1^T  \thetav_{P_1} + o(\| \thetav_{P_0} \|_{\infty}+\| \thetav_{P_1} \|_{\infty}) \Big )  \nonumber \\ 
& ~~~\times  \min\Bigg \{  \frac{D(\hat{P}_1\|\hat{P}_0)}{D(\hat{P}_0\|\hat{P}_1)}, \frac{D(\hat{P}_1\|\hat{P}_0)}{D(\hat{P}_0\|\hat{P}_1)} \Big (1+ \frac{(\dv_2-\rho_0 \dv_1)^T \thetav_{P_1} }{D(\hat{P}_1\|\hat{P}_0)} \nonumber \\
& ~~~- \frac{\dv_0^T \thetav_{P_0} }{D(\hat{P}_0\|\hat{P}_1)}  + o(\| \thetav_{P_0} \|_{\infty}+\| \thetav_{P_1} \|_{\infty})  \Big)   \Bigg \} \\
&=D(\hat{P}_1\|\hat{P}_0) + \min \Big \{  \rho_0   \dv_0^T  \thetav_{P_0}   +\rho_0   \dv_1^T  \thetav_{P_1} ,   \dv_2^T  \thetav_{P_1}  \Big \}  \nonumber \\
&~~~+ o(\| \thetav_{P_0} \|_{\infty}+\| \thetav_{P_1} \|_{\infty}),
\label{eq:expansion}
\end{align}	
where for $i=0,1$,
\begin{align}
\thetav_{P_i}&= \big(P_i(x_1)-\hat{P}_i(x_1),\dotsc,P_i(x_{|\Xc|})-\hat{P}_i(x_{|\Xc|})\big)^T,\\
\dv_0&= \bigg( 1+\log \frac{\hat{P}_{0}(x_1)}{\hat{P}_1(x_1)},\dotsc,1+\log \frac{\hat{P}_{0}(x_1)}{\hat{P}_1(x_1)}  \bigg)^T,\\
\dv_1&= \bigg( -\frac{\hat{P}_{0}(x_1)}{\hat{P}_1(x_1)},\dotsc, -\frac{\hat{P}_{0}(x_{|\Xc|})}{\hat{P}_1(x_{|\Xc|})}\bigg)^T,\\
\dv_2&= \bigg( 1+\log \frac{\hat{P}_{1}(x_1)}{\hat{P}_0(x_1)},\dotsc,1+\log \frac{\hat{P}_{1}(x_1)}{\hat{P}_0(x_1)}  \bigg)^T+\rho_0\dv_1,
\end{align}	
and $\rho_0=\frac{D(\hat{P}_1\|\hat{P}_0)}{D(\hat{P}_0\|\hat{P}_1)}$.
By substituting expansion \eqref{eq:expansion} into  \eqref{eq:worst_case} we obtain
\begin{align}\label{eq:approxworst}
\underline{\hat{E}}_0(r_0) = D(&\hat{P}_1 \|\hat{P}_0) + \min \Big \{ \rho_0 \min_{\substack{P_0 \in \Bc(\hat P_0,r_0)\\ P_1 \in \Bc(\hat P_1,r_1)}}    \dv_0^T  \thetav_{P_0}   +   \dv_1^T  \thetav_{P_1}, \nonumber \\
& \min_{\substack{P_0 \in \Bc(\hat P_0,r_0)\\ P_1 \in \Bc(\hat P_1,r_1)}}    \dv_2^T  \thetav_{P_1}  \Big \}   + o(\| \thetav_{P_0} \|_{\infty}+\| \thetav_{P_1} \|_{\infty}).
\end{align}
Now, we further approximate the outer minimization constraint in \eqref{eq:worst_case}, or, equivalently, the minimizations over the divergence balls in \eqref{eq:approxworst} to get
\begin{align}\label{eq:worstapproxopt}
\underline{\hat{E}}_0(r_0) &= D(\hat{P}_1\|\hat{P}_0) + \min \Bigg \{ \rho_0 \min_{\substack{P_0 \in \underline\Bc(\hat P_0,r_0)\\ P_1 \in \underline\Bc(\hat P_1,r_1)}}    \dv_0^T  \thetav_{P_0}   +   \dv_1^T  \thetav_{P_1},   \nonumber \\
&~~~~~~~~~\min_{\substack{P_0 \in \underline\Bc(\hat P_0,r_0)\\ P_1 \in \underline\Bc(\hat P_1,r_1)}}    \dv_2^T  \thetav_{P_1}  \Bigg \}   + o(\| \thetav_{P_0} \|_{\infty}+\| \thetav_{P_1} \|_{\infty}).
\end{align}	
where
\begin{equation}
\underline\Bc(\hat P_i,r_i) = \big\{\thetav_i \in \R^{|\Xc|}: \thetav_{P_i}^T \Jm_i \thetav_{P_i}  \leq2 r_{i} ,\onev^T\thetav_{P_i}=0 \big\}.
\end{equation}
and
\begin{equation}
\Jm_i=\diag\bigg( \frac{\alpha}{\hat{P}_i(x_1)},\dotsc,\frac{\alpha}{\hat{P}_i(x_{|\Xc|})}\bigg)
\end{equation}
is the Fisher information matrix corresponding to hypothesis $i$. Next by optimizing over $P_0$ and $P_1$ and similarly to the proof of Theorem \ref{thm:lowerworst}, by substituting $P_i$ with $\hat{P}_{i}$ we get \eqref{eq:worstapproxseq}.

\section{Proof of Theorem \ref{thm:adverLRT}}\label{apx:adverLRT}

	 Assume $\hat{Q}$ fixed. Let
	\begin{align}
	\underline{\hat{E}}_0(\hat{Q},r)=  \min_{\substack{ Q: d(Q,\hat{Q})\leq r}} D(Q\|P_0).
	\end{align}
	 To derive the Taylor expansion of the optimization we expand the $d(\hat{Q},Q)$ and  $D(Q\|P_0)$ around $Q=\hat{Q}$  to get 
	\begin{equation}
	\underline{\hat{E}}_0 (\hat{Q},r)= \min_{\substack{Q: \frac{1}{2}  \thetav_{\hat{Q}}^T	 \Jm(\hat{Q})  \thetav_{\hat{Q}}	 \leq r \\ \onev^T \thetav_{\hat{Q}}=0	 }} D(\hat{Q}\|P_0) +  \thetav_{\hat{Q}}^{T}   \nabla {E}_0  +o(\| \thetav_{\hat{Q}} \|_{\infty})
	\end{equation}
	where
	\begin{align}
	\nabla E_0 &= \bigg(1+\log \frac{\hat{Q}(x_1)}{{P}_0(x_1)},\dotsc, 1+\log \frac{\hat{Q}(x_{|\Xc|})}{{P}_0(x_{|\Xc|})}\bigg)^T,\\
	\thetav_{\hat{Q}}		 &= \big(Q(x_1)-\hat{Q}(x_1),\dotsc,Q(x_{|\Xc|})-\hat{Q}(x_{|\Xc|})\big)^T,\\
	\Jm&=\diag\bigg( \frac{\alpha}{\hat{Q}(x_1)},\dotsc,\frac{\alpha}{\hat{Q}(x_{|\Xc|})}\bigg).
	\end{align}	
	Solving this convex optimization, we obtain
	\begin{equation}
	\underline{\hat{E}}_0 (\hat{Q},r)=  D(\hat{Q}\|P_0) - \sqrt{ \frac{2}{\alpha}\text{Var}_{\hat{Q}} \Big( \log \frac{\hat{Q}}{P_0} \Big ) r} +o(\sqrt{r}).
	\end{equation}
	Next, minimizing over $\hat{Q}$ we have
	\begin{equation}
	\underline{\hat{E}}_0 (r)= \min_{\hat{Q} \in \mathcal{Q}^{\rm adv}_0 }  D(\hat{Q}\|P_0) - \sqrt{ \frac{2}{\alpha}\text{Var}_{\hat{Q}} \Big( \log \frac{\hat{Q}}{P_0} \Big ) r} +o(\sqrt{r}).
	\end{equation}
	We can expand $\underline{E}_0$ around $r=0$ as
	\begin{equation}
	\underline{\hat{E}}_0 (r)=  \underline{\hat{E}}_0 (r=0)+ \frac{\partial  \underline{\hat{E}}_0 (r) }{\partial \sqrt{r}} \Bigg|_{r=0} \sqrt{r} +o(\sqrt{r}).
	\end{equation}
	By the envelope theorem \cite{Segal}, we have that
	\begin{equation}\label{eq:lrtsensamp}
	\frac{\partial  \underline{\hat{E}}_0 (r) }{\partial \sqrt{r}} \Bigg|_{r=0}= -\sqrt{ \frac{2}{\alpha}\text{Var}_{{Q}_\lambda} \Big( \log \frac{{Q}_\lambda}{P_0} \Big ) },
	\end{equation}
	where we used the fact that $\hat{Q}= Q_{\lambda}$ when $r=0$, where $Q_\lambda$ is the optimizing distribution in \eqref{eq:tilted}. Finally, setting $\underline{\hat{E}}_0 (r=0)=E_0$, concludes the proof.


\section{Proof of Corollary \ref{cor:samp_dist}}  \label{apx:samp_dist}
We prove the result for $i=0$, and the same holds for the type-\RNum{2} sensitivity. We can write the sensitivity of type-\RNum{1} error exponent to sample mismatch as 
\begin{equation}\label{eq:thetadv}
\theta_0^{\rm adv}= \sum_{a \in \Xc } P_0(a) \frac{Q_\lambda(a)}{P_0(a)} \log ^2 \Bigg ( \frac{Q_\lambda(a)}{P_0(a)} \Bigg ) -D^2(Q_\lambda\|P_0).
\end{equation}
For every $x\geq0$, we can show the following inequality
\begin{equation}
x\log^2 x \leq (x-1)^2.
\end{equation}
Using this we have
\begin{equation}
\theta_0^{\rm adv}\leq \sum_{a \in \Xc } P_0(a) \Bigg(\frac{Q_\lambda(a)}{P_0(a)}-1   \Bigg)^2 =\chi^2(Q_\lambda\|P_0)=\theta_0^{\rm dist}. 
\end{equation}
 To prove the inequality let $f_1(x)= (x-1)^2 - x\log^2 x$. Taking the second derivative we have
\begin{equation}
f_1''(x)=2\Big(1-\frac{\log x}{x} -\frac{1}{x}\Big) \geq 0,
\end{equation}
where we used $\log x \leq x-1$. Hence  $f_1(x)$ is convex and  the first order condition is  sufficient to find the minimum of $f_1(x)$. Setting $x=1$ we get $f_1(1)=0, f_1'(1)=0$ and therefore $f_1(x)\geq 0$.

To prove the lower bound, for every $x\geq 0$   we have
\begin{equation}
x-1 \leq x \log x.
\end{equation}
 Applying this inequality to \eqref{eq:thetadv} we get
 \begin{align}
&\theta_0^{\rm adv}\notag\\ 
&\geq \sum_{a \in \Xc } P_0(a) \Bigg (\frac{Q_\lambda(a)}{P_0(a)}\Bigg )^{-1} \Bigg(\frac{Q_\lambda(a)}{P_0(a)}-1   \Bigg)^2  -D^2(Q_\lambda\|P_0)\\
 &\geq  \Bigg(\min_i\frac{P_0(i)}{Q_\lambda(i)} \Bigg )  \theta_0^{\rm dist}  -E_0^2,
 \end{align}
which concludes the proof.


\section{Proof of Theorem \ref{thm:advGLRT}}  \label{apx:advGLRT}

	We show the result under the first hypothesis; similar steps are valid for the second hypothesis.  Unlike the likelihood ratio test, we first consider the minimization over $P_0$ for a fixed $Q$ and we perform a Taylor expansion of $D(Q\|P_0)$ around $Q=\hat{Q}$  to get
	\begin{equation}\label{eq:optsenH}
	\underline{\hat{E}}_0 (r)= \min_{\substack{ Q, \hat{Q}: D(\hat{Q}\|P_0)=  \gamma \\ d(Q,\hat{Q})\leq r  }} D(\hat{Q}\|P_0) +  \thetav_{Q}^{T}   \nabla {E}_0 +o(\| \thetav_{Q} \|_{\infty}),
	\end{equation}
	where
	\begin{align}
	\nabla E_0 &= \bigg( 1+\log \frac{\hat{Q}(x_1)}{{P}_0(x_1)},\dotsc, 1+\log \frac{\hat{Q}(x_{|\Xc|})}{P_0(x_{|\Xc|})}\bigg)^T,\\
	\thetav_{Q}		 &= \big(Q(x_1)-\hat{Q}(x_1),\dotsc,Q(x_{|\Xc|})-\hat{Q}(x_{|\Xc|})\big)^T.
	\end{align}	
	We have replaced the inequality constraint with equality since the the optimal value of the minimization will be attained at the boundary. Next, by solving \eqref{eq:optsenH} over $Q$ for fixed $\hat{Q}$, we get
	
	\begin{align}
	\underline{\hat{E}}_0 (r)=& \min_{\substack{ \hat{Q}: D(\hat{Q}\|P_0) = \gamma}} D(\hat{Q}\|P_0)   - \sqrt{\frac{2}{\alpha}  {\rm Var}_{\hat{Q}} \Bigg[\log \frac{\hat{Q}(X)}{P_0(X)}  \Bigg] r}  \nonumber \\
	&+ o(\sqrt{r})  \\
	=& \gamma- \max_{\substack{ \hat{Q}: D(\hat{Q}\|P_0) = \gamma}} \sqrt{\frac{2}{\alpha}  {\rm Var}_{\hat{Q}} \Bigg[\log \frac{\hat{Q}(X)}{P_0(X)}  \Bigg]   r}   + o(\sqrt{r}). 
	\end{align}
	Similarly for the type-\RNum{2} error exponent, we have
	\begin{align}
	\underline{\hat{E}}_1 (r)=& \min_{\substack{ \hat{Q}: D(\hat{Q}\|P_0) \leq  \gamma}} D(\hat{Q}\|P_1)   - \sqrt{\frac{2}{\alpha}  {\rm Var}_{\hat{Q}} \Bigg[\log \frac{\hat{Q}(X)}{P_1(X)}  \Bigg] r}  \nonumber \\
	& + o(\sqrt{r}). 
	\end{align}
	Next, by the envelope theorem  \cite{Segal} and similarly to the proof of Theorem \ref{thm:adverLRT}  we get \eqref{eq:Hoefadver1}.


\section{Proof of Theorem \ref{thm:SPRTadver}}  \label{apx:SPRTadver}

	Assume samples are drawn by $P_0$. First, we find a bound to the probability of error as a function of the threshold. The type-\RNum{1} probability of error of sequential probability ratio test under disturbed samples $\Th'$ can be upper bounded by
	\begin{align}
	&\epsilon_0 \notag\\&\leq  \sum_{t=1}^{\infty} \PP_0\Big [ t(D(\Th' \| {P}_0)-D(\Th'\| {P}_1)) \geq  \tilde{\gamma}_1 ,  \Th \in \Bc(\Th',r)  \Big ].
	\end{align} 
	By the method of types we have
	\begin{align}\label{eq:uppersensamp}
	\epsilon_0 &\leq \sum_{t=1}^{\infty} \sum_{\substack{\hat{Q} \in		\hat{\mathcal{Q}}_{\tilde{\gamma}_1}(t)\\ Q:  d(Q,\hat{Q})\leq r } } e^{ -t D(Q\| P_0) }\\
	&\leq \sum_{t=1}^{\infty} (t+1)^{|\mathcal{X}|} e^{- \underline{E}_{ \tilde{\gamma}_1}(r,t)},
	\end{align} 
	where 
	\begin{equation}\label{eq:set}
	\hat{\mathcal{Q}}_{\gamma}(t)=\Big \{\hat{Q}:   D(\hat{Q}\| {P}_0)-D(\hat{Q}\| {P}_1)   \geq \frac{\gamma}{t} \Big \},
	\end{equation} 
	\begin{equation}\label{eq:Eadverseq}
	\underline{E}_{\gamma}(r,t)=t \min_{\hat{Q} \in		\hat{\mathcal{Q}}_{{\gamma}}(t)}  \min_{ Q: d(Q,\hat{Q})\leq r } D(Q\|P_0).
	\end{equation}	
	Let $ \tilde{\gamma}_1=\gamma_1+\frac{ |\Xc|+2}{\lambda^*_1} \log (t+1)$,  where $\lambda_1^*$ is the optimal Lagrange multiplier corresponding to the constraint in \eqref{eq:set}  of optimization in \eqref{eq:Eadverseq} when $\gamma=\gamma_1$. We expand   $\underline{E}_{ \tilde{\gamma}_1}(r,t)$ around $ \tilde{\gamma}_1=\gamma_1$. Similarly to Lemma \ref{lem:convex}  it can be shown that $\underline{E}_{\gamma}(r,t)$ is convex in $\gamma$, hence 
	\begin{equation}
	\underline{E}_{\tilde{\gamma}_1}(r,t) \geq \underline{E}_{\gamma_1}(r,t) + \frac{\partial  \underline{E}_{\gamma}(r,t) }{\partial \gamma}\Bigg |_{\gamma=\gamma_1}  \frac{ |\Xc|+2}{\lambda^*_1} \log (t+1).
	\end{equation}
	  By the envelope theorem we have
	\begin{equation}
	\frac{\partial  \underline{E}_{\gamma}(r,t) }{\partial \gamma}\Bigg |_{\gamma=\gamma_1} ={\lambda_1^*} \geq 0.
	\end{equation}
	 Furthermore, the inequality is strict if $\frac{\gamma_1}{t}\geq -D(P_0\|P_1)$, hence by choosing $\gamma_1\geq0$  for every $t$ this condition is satisfied. Hence we can upper bound \eqref{eq:uppersensamp} by
	\begin{align}
	\epsilon_0 & \leq \sum_{t=1}^{\infty} (t+1)^{-2} e^{- \underline{E}_{\gamma_1}(r,t)},\\
	& \leq \frac{\pi^2}{6} e^{- \min_{t\geq1} \underline{E}_{\gamma_1}(r,t)}. \label{eq:uppersensamp_app}
	\end{align} 
	Next, by Taylor expanding $\underline{E}_{\gamma_1}(r,t)$, we have 
	\begin{equation}\label{eq:Taylorseqsamp}
	\underline{E}_{\gamma_1}(r,t)= \underline{E}_{\gamma_1}\big(r=0,t\big)+\frac{\partial  \underline{E}_{\gamma_1}(r,t) }{\partial \sqrt{r}}\Bigg |_{r=0,t} \sqrt{r} +o(\sqrt{r}).
	\end{equation}
	Also, using \eqref{eq:lrtsensamp} in the proof of theorem \ref{thm:adverLRT},  we get
	\begin{equation}\label{eq:senseqsam}
	\frac{\partial  \underline{E}_{\gamma_1}(r,t) }{\partial \sqrt{r}}\Bigg |_{r=0,t}=-t \sqrt{\frac{2}{\alpha}\text{Var}_{Q_{\lambda(t)}} \Bigg(\log \frac{Q_{\lambda(t)}}{{P}_0}  \Bigg)}, 
	\end{equation}
	where $Q_{\lambda(t)}$ is the optimizing distribution  in \eqref{eq:tilted} for the case where $\gamma=\frac{\gamma_1}{t}$ in \eqref{eq:constraint1}. 
	Let 
	\begin{equation}\label{eq:exptime}
	\underline{E}_{\gamma_1}(r)=\min_{t \geq 1}  \underline{E}_{\gamma_1}(r,t).
	\end{equation}  
	Taking derivatives respect to $\sqrt{r}$ we have
	\begin{align}
	\frac{d \underline{E}_{\gamma_1}(r) }{d \sqrt{r}} &=\frac{\partial  \underline{E}_{\gamma_1}(r,t^*(r)) }{\partial \sqrt{r}}  + \frac{\partial  \underline{E}_{\gamma_1}(r,t^*(r))}{\partial t^*(r)} \cdot \frac{d t^*(r)}{d \sqrt{r}}.
	\end{align}
	where $t^*(r)$ is the minimizing $t$ in \eqref{eq:exptime} as the function of $r$. By the first order condition $ \frac{\partial  \underline{E}_{\gamma_1}(r,t^*(r))}{\partial t^*(r)} =0$ , we get
	\begin{align}
	\frac{d \underline{E}_{\gamma_1}(r) }{d \sqrt{r}} = \frac{\partial  \underline{E}_{\gamma_1}(r,t^*(r)) }{\partial \sqrt{r}},  
	\end{align}
	 hence
	\begin{equation}\label{eq:envseq}
	\frac{d \underline{E}_{\gamma_1}(r) }{d \sqrt{r}}\Bigg |_{r=0}=  \frac{\partial  \underline{E}_{\gamma_1}(r,t^*(r)) }{\partial \sqrt{r}} \Bigg |_{r=0,t^*(r=0)}.  
	\end{equation}
	To find the $t^*(r=0)$, note that $D({Q}\| {P}_1)\geq 0$, hence
	\begin{align}
	\underline{E}_{\gamma_1}(r=0,t)&=\min_{{Q}:   D({Q}\| {P}_0)  \geq \frac{\gamma_1}{t}+D({Q}\| {P}_1)  }  tD(Q\|P_0) \label{eq:optime}\\
	&\geq  \gamma_1.
	\end{align}	
	Letting $\gamma_1=nD(P_1\|P_0)$,  $t=n$ will achieve this minimum. Additionally, $t^*(r=0)=n$ is the unique solution. To see this we can write the optimization in \eqref{eq:optime} in the dual form as

	\begin{equation}
	\underline{E}_{\gamma_1}(r=0,t)= \max_{\lambda \geq 0} \gamma_1\lambda -t \log \Big (  \sum_{x\in \Xc} P_0^{1-\lambda}(x) P_1^{\lambda}(x) \Big ). 
	\end{equation}     
	Since $\underline{E}_{\gamma_1}(r=0,t)$ is the supremum of linear functions in $t$, therefore it is convex in $t$. Also by the envelope theorem, we have
	\begin{equation}
	\frac{\partial \underline{E}_{\gamma_1}(r=0,t)}{\partial t}=- \log \Big (  \sum_{x\in \Xc} P_0^{1-\lambda^*}(x) P_1^{\lambda^*}(x) \Big ), 
	\end{equation}
	and setting this to zero, we can conclude that first order condition only satisfies if $\lambda=0$ or $\lambda=1$, i.e., $Q_\lambda=P_0$ or $Q_\lambda=P_1$ should be the optimizer in  \eqref{eq:optime}, and it is clear that only $\frac{\gamma_1}{t}=D(P_1\|P_0)$ can satisfy this condition which  shows the uniqueness of the solution. Then by \eqref{eq:senseqsam},  \eqref{eq:envseq} and substituting $t^*(r=0)=n$ , we obtain that
	\begin{equation}\label{eq:firstder}
	\frac{d \underline{E}_{\gamma_1}(r) }{d \sqrt{r}}=-n\sqrt{\frac{2}{\alpha}\text{Var}_{P_1} \Bigg(\log \frac{P_1}{{P}_0}  \Bigg)}.
	\end{equation}
	Also, we have
	\begin{equation}\label{eq:zeroder}
	\underline{E}_{\gamma_1}(r=0)=\underline{E}_{\gamma_1}\big(r=0,t^*(r=0)\big)=nD(P_1\|P_0).
	\end{equation}
	Finally, By  \eqref{eq:firstder}, \eqref{eq:zeroder} and Taylor expanding $\underline{E}_{\gamma_1}(r)$ around $r=0$ as the function of $\sqrt{r}$, we get
	\begin{align}\label{eq:seqsampsenexp}
	\epsilon_0 \leq c \cdot e^{- n\Big (D(P_1\|P_0)-\sqrt{\frac{2}{\alpha}\text{Var}_{P_1} \big(\log \frac{P_1}{{P}_0}  \big) } \Big )},
	\end{align} 
	where $c$ is a positive constant. Next, we find the  worst-case expected stopping time $	\mathbb{E}_{P_0}[\underline{\hta}], 	\mathbb{E}_{P_1}[\underline{\hta}]$. We can write the accumulated log-likelihood ratio $\hat{S}_{n}$  evaluated at adversarial samples with the type $\Th'$ by
	\begin{equation}
	\frac{\hat{S}_{n}}{n}=\frac{S_{n}}{n}+ \thetav_{\Th}^{T}   \nabla \hat{S},   
	\end{equation}
	where
	\begin{align}
	\nabla \hat{S} &= \bigg( \log \frac{P_{0}(x_1)}{P_1(x_1)},\dotsc, \log \frac{P_{0}(x_{|\Xc|})}{P_1(x_{|\Xc|})}\bigg)^T,\\
	\thetav_{\Th}		 &= \big(\Th'(x_1)-\Th(x_1),\dotsc,\Th'(x_{|\Xc|})-\Th(x_{|\Xc|})\big)^T,
	\end{align}			
	and $\Th$ is the type of the original samples at time $n$. Assume $\Th$ is fixed and the adversary is trying to maximize the stopping time, or equivalently reducing the $\hat{S}_n$ under the first hypothesis. Therefore, letting $\underline{\hat{S}}_n= \min_{ \hat{Q}: d(\Th, \hat{Q})\leq r}\hat{S}_n $ we have
	\begin{align}
	\frac{\underline{\hat{S}}_n}{n}&=\frac{S_n}{n}+ \min_{ \hat{Q}: d(\Th,\hat{Q})\leq r} \thetav_{\Th}^{T}   \nabla \hat{S}
	\end{align}
	Note that, this is a convex optimization problem, and similarly to the proof of Theorem \ref{thm:adverLRT} we get
	\begin{align}
	\frac{\underline{\hat{S}}_n}{n}=\frac{S_n}{n}- \sqrt{\frac{2}{\alpha}\text{Var}_{Q} \Bigg(\log \frac{P_0}{{P}_1} \Bigg)r } +o(\sqrt{r}).
	\end{align}
	Let       
	\begin{align}\label{eq:tau1a}
	\underline{\hta}_{0}=\inf \{n\geq1: {\underline{\hat{S}}_n}& \geq\tilde{ \gamma}_0 \},  
	\end{align}   
	to be the worst-case stopping time under the adversarial perturbation.  Similarly to proof of the Theorem \ref{thm:lowerworstseq}, it is easy to show that the worst case stopping time  $\underline{\hta}_{0}$ tends to infinity as $\tilde{\gamma}_0\rightarrow \infty$, and also for every finite $\tilde{\gamma}_0$ the stopping time is finite with probability one. 	Let $\tilde{\gamma}_0 =\gamma_0 +\frac{ |\Xc|+2}{\lambda^*_2} \log (t+1)$  where $\lambda_2^*$ is the optimal Lagrange multiplier defined similarly to $\lambda_1^*$ for the type-\RNum{2} error exponent. By SLLN  and  continuous mapping theorem as $\gamma_0 \rightarrow \infty$ we have
	\begin{align}
	\frac{\underline{\hat{S}}_{\underline{\hta}_{0}}}{\underline{\hta}_{0}}   \xrightarrow[]{p}  D(P_0\|P_1)- \sqrt{\frac{2}{\alpha}\text{Var}_{P_0} \Bigg(\log \frac{P_0}{{P}_1}  \Bigg) r } +o(\sqrt{r}),\label{eq:convergp1samp}\\ 
	\frac{\underline{\hat{S}}_{\underline{\hta}_{0}-1}}{\underline{\hta}_{0}-1}   \xrightarrow[]{p}  D(P_0\|P_1)- \sqrt{\frac{2}{\alpha}\text{Var}_{P_0} \Bigg(\log \frac{P_0}{{P}_1}  \Bigg)r } +o(\sqrt{r}). \label{eq:convergp2samp}
	\end{align}
     It is easy to show that there exist a finite $\underline{\hta}_{0}$ with probability one such that
	\begin{align}\label{eq:convergp}
	\hat{s}_{\underline{\hta}_{0}-1} \leq \tilde{\gamma}_0 < \hat{s}_{\underline{\hta}_{0}}~~~   \text{with probability } 1.
	\end{align}
	Therefore, from \eqref{eq:convergp1samp}, \eqref{eq:convergp2samp}, and \eqref{eq:convergp}, we conclude that 
	\begin{equation}\label{eq:convP1}
	\frac{\underline{\hta}_{0}}{\gamma_0}  \xrightarrow[]{p}   \Bigg ( D(P_0\|P_1)- \sqrt{\frac{2}{\alpha}\text{Var}_{P_0} \Bigg(\log \frac{P_0}{{P}_1}  \Bigg)r } +o(\sqrt{r}) \Bigg )^{-1}, 
	\end{equation}  	
	as	$\gamma_0 \rightarrow  \infty$.

	Similarly to the proof of Theorem \ref{thm:lowerworstseq}  we can show the convergence in expectation by proving the uniform integrability of $\frac{\underline{\hta}_0}{ \gamma_0}$ as $\gamma_0 \rightarrow \infty$ as well as convergence of $\frac{\underline{\hta}}{\gamma_0}$ by using the convergence of $\frac{\underline{\hta}_{0}}{\gamma_0}  $  . Hence
	\begin{align} \label{eq:SPRTadvtime}
	\mathbb{E}_{P_0}[\underline{\hta}]=& \frac{ \gamma_0}{D(P_0\|P_1)}+\frac{\gamma_0 \sqrt{\frac{2}{\alpha}\text{Var}_{P_0} \Big(\log \frac{P_0}{{P}_1}  \Big)r }}{D^2(P_0\|P_1)}  \nonumber \\
	& +o(1)+o(\sqrt{r}),\\ 
	\mathbb{E}_{P_1}[\underline{\hta}]=& \frac{ \gamma_1}{D(P_1\|P_0)}+\frac{\gamma_1 \sqrt{\frac{2}{\alpha}\text{Var}_{P_1} \Big(\log \frac{P_1}{{P}_0}  \Big)r }}{D^2(P_1\|P_0)}  \nonumber \\
	& +o(1)+o(\sqrt{r}),
	\end{align}
	Finally letting $\gamma_0=nD(P_0\|P_1), \gamma_1=D(P_1\|P_0)$ we get
	\begin{align} \label{eq:SPRTadvtime}
	\mathbb{E}_{P_0}[\underline{\hta}]&=n+\frac{ \sqrt{\frac{2}{\alpha}\text{Var}_{P_0} \Big(\log \frac{P_0}{{P}_1}  \Big)r }}{D(P_0\|P_1)}n +o(1)+o(\sqrt{r}),\\ 
	\mathbb{E}_{P_1}[\underline{\hta}]&=n+\frac{ \sqrt{\frac{2}{\alpha}\text{Var}_{P_1} \Big(\log \frac{P_1}{{P}_0}  \Big)r }}{D(P_1\|P_0)}n +o(1)+o(\sqrt{r}),
	\end{align}
	 and using \eqref{eq:seqsampsenexp} the worst case error exponent will satisfy
	\begin{align}
	&\underline{\hat{E}}_0 (r)	\underline{\hat{E}}_1 (r) \geq D(P_0\|P_1)D(P_1\|P_0)   \nonumber\\ 
	& \times\Bigg (1- \frac{ 2\sqrt{\frac{2}{\alpha}\text{Var}_{P_0} \big(\log \frac{P_0}{{P}_1}  \big)r }}{D(P_0\|P_1)}-\frac{2 \sqrt{\frac{2}{\alpha}\text{Var}_{P_1} \big(\log \frac{P_1}{{P}_0}  \big)r }}{D(P_1\|P_0)}\Bigg)\notag\\& +o(\sqrt{r}),
	\end{align}
	which concludes the proof.

\bibliographystyle{ieeebib}
\bibliographystyle{ieeetr}
\bibliography{journal_abbr,Mismatch_Journal}

\newpage

\begin{IEEEbiographynophoto}{Parham Boroumand} received a bachelor's degree in electrical engineering from the Sharif University of Technology in 2018. He is currently working towards the Ph.D. degree in the Signal Processing and Communications Group at the Department of Engineering, University of Cambridge, Cambridge, U.K. His research interests include the areas of statistics, information theory, and signal processing.  \end{IEEEbiographynophoto}

\vspace{-140mm}

\begin{IEEEbiographynophoto}{Albert Guill\'en i F\`abregas}
(S--01, M--05, SM--09, F--22) received the Telecommunications Engineering Degree and
the Electronics Engineering Degree from Universitat Polit\`ecnica de
Catalunya and Politecnico di Torino, respectively in 1999, and the Ph.D.
in Communication Systems from \'Ecole Polytechnique F\'ed\'erale de
Lausanne (EPFL) in 2004.

In 2020, he returned to a full-time faculty position at the Department
of Engineering, University of Cambridge, where he had been a full-time
faculty and Fellow of Trinity Hall from 2007 to 2012. Since 2011 he has
been an ICREA Research Professor at Universitat Pompeu Fabra (currently
on leave). He has held appointments at the New Jersey Institute of
Technology, Telecom Italia, European Space Agency (ESA), Institut
Eur\'ecom, University of South Australia, Universitat Pompeu Fabra, University of Cambridge, as
well as visiting appointments at EPFL, \'Ecole Nationale des
T\'el\'ecommunications (Paris), Universitat Pompeu Fabra, University of
South Australia, Centrum Wiskunde \& Informatica and Texas A\&M
University in Qatar. His specific research interests are in the areas of
information theory, communication theory, coding theory, statistical inference.

Dr. Guill\'en i F\`abregas is a Member of the Young Academy of Europe,
and received the Starting and Consolidator Grants from the European
Research Council, the Young Authors Award of the 2004 European Signal
Processing Conference, the 2004 Best Doctoral Thesis Award from the
Spanish Institution of Telecommunications Engineers, and a Research
Fellowship of the Spanish Government to join ESA. Since 2013 he has been an Editor of the
Foundations and Trends in Communications and Information Theory, Now
Publishers and was an Associate Editor of the \sc{IEEE Transactions on
Information Theory} (2013--2020) and \sc{IEEE Transactions on Wireless
Communications} (2007--2011).
\end{IEEEbiographynophoto}

\end{document}